\documentclass[openleft,12pt,a4paper,twoside]{report}
\usepackage[usenames]{color} %used for font color
\usepackage[UKenglish]{babel}
\usepackage{amssymb} %maths
\usepackage{amsmath} %maths
\usepackage{amsthm}
\newtheorem{theorem}{Theorem}[section]

\newtheorem*{repeatedtheorem}{Theorem \ref{LinDeFinetti}}
\newtheorem{lemma}{Lemma}[section]
\newtheorem*{repeatedlemma1}{Lemma \ref{lemmasym}}
\newtheorem*{repeatedlemma2}{Lemma \ref{lemma2}}
\newtheorem*{remark}{Remark}
\newtheorem{observation}{Observation}
\newtheorem{definition}{Definition}
\newtheorem{optimize}{Optimization problem}[section]

\DeclareMathOperator{\trace}{tr}
\DeclareMathOperator{\tr}{tr}
\DeclareMathOperator{\Tr}{tr}
\usepackage[utf8]{inputenc} %useful to type directly diacritic characters
\usepackage{graphicx} % for inserting graphics
\usepackage{braket} %for bras and kets
\usepackage{caption}
\usepackage{mathtools}
\usepackage{appendix}
\usepackage{subcaption} %for subplots
\usepackage{mathrsfs}
\usepackage{bbm} %identity operater?
\usepackage{float} %to keep figures in place
\usepackage{geometry} %geometry of paper
\usepackage{setspace} %fur 1,5fachen Zeilenabstand
\usepackage[colorlinks]{hyperref}
\usepackage{xcolor} %so that I can edit color of Refs
\hypersetup{
    colorlinks = true,
    linkcolor=magenta,
    linkbordercolor = white,
    citecolor=blue,
    citebordercolor=white,
    urlcolor=red,
    urlbordercolor=white,
}
\usepackage[normalem]{ulem} %for underlining s.t. i und j have same
\usepackage{titlesec}
\titleformat{\chapter}[block]
  {\normalfont\huge\bfseries}{\thechapter}{1em}{\Huge} 
\titlespacing*{\chapter}{0pt}{-20pt}{40pt}

\usepackage{fancyhdr}
\begin{document}
\begin{titlepage}
   %ijäaefäef opewrjövdkmaäe efojamclyxjo eafdvlmyxcüaerh  sefdlcpöoqhl weoijvnxlc efwüjsdlx ürdfj eofjdnjkvnfibwep oejfrnfdkokn
        %\vspace*{0cm}
      
        \noindent
        \begin{minipage}[c]{0.5\textwidth}
\includegraphics[width=0.87\textwidth]{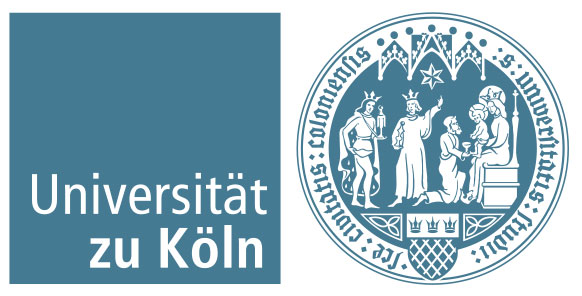}

        \Large{\textsc{University of Cologne }}
\end{minipage}    
\hspace{0.2cm}
\begin{minipage}[c]{0.5\textwidth}
\includegraphics[width=\textwidth]{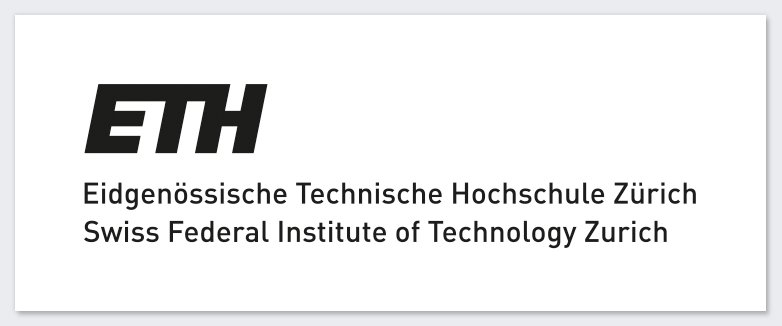}

        \flushright{ \Large{\textsc{ETH Zürich}}}
        
\end{minipage}

        \vspace{0.2cm}
       
        %\Large{\textsc{Department of Physics}}
                
           \vspace{1.5cm}
            %\begin{center}
   %     \begin{center}
%\textbf{{\LARGE{Studying Symmetries
%}}}
%\end{center}
%\begin{center}
%\textbf{{\LARGE{in Tensor Powers of Stabilizer States
%}}}
%\end{center}
%\begin{center}
%{\textbf{\LARGE{and Possible Applications}}}
%\end{center}
%\vspace{1cm}

%choice1
        \begin{center}
\textbf{\LARGE{Studying Stabilizer de Finetti Theorems
}}
\end{center}
\begin{center}
\textbf{\LARGE{and Possible Applications
}}
\end{center}
\begin{center}
\textbf{\LARGE{in Quantum Information Processing
}}
\end{center}

        \vspace{1cm}

        %\textbf{University of Cologne}
        
        %\begin{figure*}
        %\centering
       % \includegraphics[width=7cm]{Unikoln}
        %\end{figure*}

        \begin{center}
        
        \Large{
    
Master Thesis by Paula Belzig

 Date: 30th of April 2020
 }
 \end{center}
         \vfill
        \begin{center}
\Large{
        First assessor: Prof. Dr. David Gross%, University of Cologne

        Second assessor: Prof. Dr. Renato Renner%, ETH Zürich
        
        Co-supervisor: Dr. Joe Renes%, ETH Zürich
        }
     
 \end{center}
        \begin{flushright}

    \end{flushright}
\end{titlepage}
\newpage
\thispagestyle{empty}
\cleardoublepage
\newpage
\pagenumbering{gobble}
\setlength{\parindent}{0.3cm} %indentt/
\section*{Abstract}

Symmetries are of fundamental interest in many areas of science. In quantum information theory, if a quantum state is invariant under permutations of its subsystems, it is a well-known and widely used result that its marginal can be approximated by a mixture of tensor powers of a state on a single subsystem. Applications of this \textit{quantum de Finetti theorem} range from quantum key distribution (QKD) to quantum state tomography and numerical separability tests.
Recently, it has been discovered by Gross, Nezami and Walter that a similar observation can be made for a larger symmetry group than permutations: states that are invariant under stochastic orthogonal symmetry are approximated by tensor powers of stabilizer states, with an exponentially smaller overhead than previously possible. This naturally raises the question if similar improvements could be found for applications where this symmetry appears (or can be enforced). Here, two such examples are investigated.

Using the postselection technique developed by Christandl, König and Renner and generalized by Leverrier, we show that the new version of the quantum de Finetti theorem leads to an improvement of known bounds on the diamond norm. Subsequently, these bounds can be used in the context of a QKD protocol to infer the security of general attacks (where an adversary can manipulate all signals in any way they want) from the security of collective attacks (an adversary is restricted to acting independently and identically on each signal) with a smaller overhead on the security parameter than previously possible. %Both the case for two party protocols as well as N party protocols are discussed.

Moreover, quantum de Finetti theorems naturally give rise to a way of approximating separable quantum states by a hierarchy of semi-definite programs (SDP). This facilitates (for example) the approximation of the maximum fidelity of a quantum communication channel, which is an indicator for the success of a quantum error correction procedure.
Since the new version of the quantum de Finetti Theorem describes closeness to separable tensor powers of stabilizer states rather than arbitrary separable states, there is a clear motivation to study if it can also lead to a similar SDP hierarchy for optimal Clifford operations. Here, we find that it does, with a minimal change in convergence speed.

\vspace{5cm}
\noindent\textcolor{red}{Update [13 March 2024]: Our proposed fix for the proof from \cite{CKR} for the application of the postselection technique to QKD in Section \ref{cohtocoll} contains an error that changes the scaling of the security parameter. We comment on this at the appropriate places and refer to \cite{NTZLT} for a detailed discussion.}

\newpage
\thispagestyle{empty}
\cleardoublepage
\newpage

\tableofcontents
\thispagestyle{empty} 
\newpage
\pagenumbering{arabic}
\chapter{Introduction and Motivation}
\setcounter{page}{1}
\setcounter{section}{0}
\label{introduction}

As information processing devices are decreasing in size, the impact of quantum effects increases in relevance. In addition, the emergence of novel devices like the quantum computer increase the need for understanding of the rules and limits of information processing in a quantum setting, in particular pertaining to secure and correct data transfer.

%One important notion in this context is the notion of independent and identically distributed quantum states (\textit{i.i.d.\ states}), because their analysis is well-known. de Finetti.

This work focuses on two important aspects of quantum information processing - quantum key distribution (QKD) and quantum error correction (QEC) - and how they can be investigated using various mathematical tools, particularly quantum de Finetti theorems.

%Quantum de Finetti theorems are an important result and tool.
%They have applications, which we will investigate here.
Quantum de Finetti theorems are an important result and a ubiquitous tool in quantum information connecting permutation invariant quantum states to independent and identically distributed quantum states (\textit{i.i.d.\ states}): If a quantum state that is spread over multiple subsystems is invariant under swappings of the subsystems, its marginal can be approximated by a convex mixture of i.i.d.\ states, where the approximation improves with increased number of traced out systems \cite{RennerPhD,Renner07,KM09,CKMR07}.
Thereby, this theorem can be used to justify one of the most basic assumptions in physics, namely that the properties of a large system can be inferred from experiments conducted on a small part of it: if a physical law holds true in one subsystem, it is justified to assume that it holds in all other subsystems and independently of the subsystem. Similarly, an assumption of an i.i.d.\ structure is also at the basis of many quantum information theoretical problems, like tomography \cite{ChrRen12} and cryptography \cite{RennerPhD,Renner07}. Inferring the structure of a state (or a key encoded in it) requires that its subsystems, which are measured separately, are in fact identically distributed and independent from one another.

Additionally to their use in justifying important basic assumptions, quantum de Finetti theorems also find application in many attempts at studying a system's quantum state - and here, we will study two such examples. %, for example in terms of an eavesdropper's ability to guess %, in terms of its closeness to an i.d or separability.%. On the one hand, they find application in QKD, where two trusted parties try to generate a string of bits (a key) which is only known to them, to approximate how much an eavesdropper may infer or guess about the secret key. On the other hand, de Finetti theorems facilitate the study of separable states, which are interesting from a fundamental point of view as well as many applications, for example in studying maximum channel fidelity, which contains separability in the cut between encoder and decoder.
\\
\\
On the one hand, there is a direct and most natural application of quantum de Finetti theorems to QKD security \cite{RennerPhD}. %, to approximate how much an eavesdropper may guess about the quantum state of two trusted parties.

In this day and age, security of our communication systems has become more important than ever. Sharing personal data online always comes at the risk of revealing sensitive information to potentially bad actors. To safely share an encrypted message, a \textit{secret}, between two trusted parties (e.g. you and your friend/your bank/your boss), we use a protocol called key distribution - at the end of which both parties should have an identical string of bits (a key, a password), while an adversary has no information about this shared bitstring. Current classical cryptosystems usually rely on the fact that a very large computing time prohibits a potential adversary from decrypting messages \cite{crypto_book}. However, the physical laws of quantum mechanics can provide trusted parties with an advantage for communication, which can be exploited for QKD. %QKD is concerned with the following cryptographic problem: Given an authenticated (but possibly tapped or bugged) classical channel, can two distant parties, Alice and Bob, transform a shared input state into two identical bit strings that are unknown to any third party?

While correlations between the trusted parties can become stronger in the quantum mechanical framework, the potential eavesdropper also obtains an advantage: entanglement can also be used to attack. There are two different categories of attack: A most powerful \textit{general attack}, where the adversary has all resources available to them, and a \textit{collective attack}, where the adversary can only act identically and independently on each separate signal. One strong result in QKD is the fact that the security of collective attacks implies the security of general attacks, where the security parameter (and thus the chance of information being revealed to the attacker) changes by a multiplicative factor \cite{RennerPhD,CKR,Lev17}. 

This is a direct result of quantum de Finetti theorems and leads to one of the main results of this thesis: If a QKD protocol has a certain symmetry, namely stochastic orthogonal symmetry introduced in \cite{Gross}, this factor becomes much smaller than for previous attempts relying on permutation invariance, and thus becomes much more achievable in a realistically small setup.% - the protocol becomes more secure against general attacks.
\\
\\
On the other hand, quantum de Finetti theorems find application in approximating separable quantum states \cite{DPS02,DPS04,BBFS18}.%de Finetti theorems can be used in the context of QEC, particularly pertaining to maximum channel fidelity.

%Realistically, the safe transfer of a 
The transfer of a secret message is not only susceptible to an enemy's meddling, but also disturbances in the communication channel, like faulty cables. Some such influences that corrupt the data can be counteracted by error correction (at its simplest: exchanging a broken cable). If one cannot pinpoint the exact instance where an error occurs, or there exist different ways of correcting it, the success of an error correcting procedure can be measured by determining the maximum success probability for transmitting a uniform message over the channel \cite{BF18}. Then, by comparing the maximum success probabilities of different error correcting procedures, the best one can be identified.

In a quantum setting, the safe transfer of data is threatened by classical noise as well as disturbances at the quantum level, like thermal fluctuations \cite{fern2009limitations}, which can be counteracted by QEC. Instead of maximum success probability, the maximum channel fidelity of a given noisy channel is the property that allows for comparison of the success of different QEC procedures \cite{LM15}. Maximum channel fidelity is the bilinear optimization problem of finding the best possible combination of encoder and decoder for a given (and possibly or partly corrected) noisy channel, and comparing input and output state of the whole transfer to see how faithfully the state was recovered.

Instead of the complicated task of optimizing the combination of encoder and decoder, the problem can be recast as a task for finding the best possible separable state, which can be approximated by a hiearchy of converging semidefinite programming (SDP) relaxations \cite{BBFS18}.
This hierarchy is chiefly possible because of a quantum de Finetti theorem approximating states with permutation invariance on one side (for example the decoder's) by states with separability between encoder and decoder.
 In this thesis, we show that an analogous de Finetti theorem derived from \cite{Gross} and a subsequent hierarchy can be found for a different symmetry, leading to a hierarchy for maximum channel fidelity with optimal Clifford decoder (or encoder, or both) instead of optimal arbitrary encoder and decoder, which is interesting for studying Clifford operations. %This is particularly interesting because of the close relation between Clifford operations and stabilizer codes used in QEC.

\newpage

\chapter{Introducing Stabilizer de Finetti Theorems}
\label{Preliminaries}

Symmetries appear in various physics related contexts: When symmetries exist, systems often become much more straightforward to treat (e.g. crystal structure), and when symmetries cease to exist, it is a key sign of critical behaviour (e.g. phase transitions). Likewise, the same is true for many problems in quantum information theory, where symmetry considerations can lead to important information about the system's state (e.g. concerning bosonic or fermionic systems). One important and widely used result in quantum information theory is the so called quantum de Finetti theorem, which relates symmetric quantum states to i.i.d.\ states, which are often the desired and most well-understood states for the analysis of many applications \cite{RennerPhD,Renner07}.

This chapter contains a short introduction to symmetry groups in Section \ref{sec-symgroups}, focussing on permutations and stochastic orthogonal symmetry in particular, before outlining how the study of de Finetti theorems emerged and evolved in Section \ref{sec-deFinetti}, which ultimately lead to the discovery of a novel de Finetti theorem in \cite{Gross} which lies at the basis of this thesis.

\section{Symmetry Groups}
\label{sec-symgroups}

When there exists a set of operations or transformations which leave a mathematical object unchanged, this property is referred to as a \textit{symmetry} of that object. 
This mathematical property occurs in numerous mathematical contexts, including geometry, calculus and linear algebra. In this thesis, the focus lies on symmetry in the context of group theory and representation theory.

The \textit{symmetry group} $\mathcal{S}$ of an object is the group of all transfomations $s \in \mathcal{S}$ that leave the object invariant. The simplest example is a sphere, which will remain exactly the same under any kind of rotation about its center. Its symmetry group then consists of all these rotations.

Within the space on which $\mathcal{S}$ acts, there could be multiple objects that are left invariant by it. Such a set of objects, which are left invariant under one and the same set of symmetry transformations $s \in \mathcal{S}$ spans a subspace of the whole space, called the \textit{$\mathcal{S}$-invariant subspace}:

\begin{definition}[$\mathcal{S}$-invariant subspace]
Let $\mathcal{H}$ be a vector space, and let $\mathcal{S}$ be a symmetry group acting on $\mathcal{H}$. Then, the $\mathcal{S}$-invariant subspace $(\mathcal{H})^s \subseteq \mathcal{H}$ is defined as the vector space spanned by the projection $\frac{1}{|\mathcal{S}|} \sum_{s\in\mathcal{S}} s$ applied to $\mathcal{H}$.
\end{definition}

%Generally, when studying a problem with a particular symmetry, one is interested in the group elements, i.e. the set

Two symmetries that are central to this project are permutation invariance (conventionally used in de Finetti type arguments) and stochastic orthogonal invariance (appearing in \cite{Gross}, and possibly providing an advantage over permutation invariance in de Finetti type considerations).

\subsection{Permutation Invariance}

Instead of working with a singular state $\rho$ on some Hilbert space $\mathcal{H}$, quantum information processing tasks often consider many copies of the same state, $\rho^{\otimes n}$ on a composite Hilbert space $\mathcal{H}^{\otimes n}$, as an input to a protocol (e.g. teleportation \cite{ChrRen12}, quantum key distribution \cite{Ekert91}).
States with this type of structure are generally referred to as i.i.d.\ states (as alluded to in Chapter \ref{introduction}). Because many results on different information processing tasks rely on assuming an i.i.d.\ state as input, it is of great importance and interest to analyse how an arbitrary state differs from it. For many cases, such an analysis comes in the form of de Finetti theorems, which show that a task's symmetry can be utilized to justify an approximation by a mixture of i.i.d.\ states (see Section \ref{sec-deFinetti}).

There are two main symmetry groups associated with such $n$-fold tensor powers: the symmetric group $\mathcal{P}_n$ and the unitary group $\mathcal{U}(d)$, which act on the $n$-fold copy of a $d$-dimensional Hilbert space $\mathcal{H}$ in the following ways:

\begin{definition}[Action of the symmetric group]
Let $\mathcal{H}$ be a $d$-dimensional complex vector space. Then, the action of the symmetric group $\mathcal{P}_n$ on objects in $\mathcal{H}^{\otimes n}$ is defined by 
the permutations
 $\pi \in\mathcal{P}_n$ with
\[{\pi}: \ket{\phi_1} \otimes \dots \otimes \ket{\phi_n} \mapsto\ket{\phi_{\pi_1}} \otimes \dots \otimes \ket{\phi_{\pi_n}}.\]
\end{definition}

\begin{definition}[Action of the tensor power unitary group]
Let $\mathcal{H}$ be a $d$-dimensional complex vector space. Then, the action of the unitary group $\mathcal{U}(d)$ on objects in $\mathcal{H}^{\otimes n}$ is defined by tensor powers of
the unitary matrices
 $U^{\otimes n}$ for $ U\in \mathcal{U}(d)$ with
\[U^{\otimes n}: \ket{\phi_1} \otimes \dots \otimes \ket{\phi_n}\mapsto U\ket{\phi_{1}} \otimes \dots \otimes U\ket{\phi_{n}}.\]
\end{definition}

When considering a $n$-fold tensor power of some state $\rho$, the resulting state $\rho^{\otimes n}$ is obviously invariant under permutation (i.e. switching) of the subsystems, and therefore invariant under the action of the symmetric group $\mathcal{P}_n$.
Furthermore, any problem that involves the eigenvalues of a state (like computing an entropy or a trace) will be invariant under unitary operations $U: \rho\mapsto U\rho U^{\dagger}$ on each subsystem $\rho\in\mathcal{H}$. Consequently, the $n$-fold tensor product $\rho^{\otimes n}$ of the state will also be invariant under tensor powers of unitaries, and thus invariant under the action of $U^{\otimes n}$.

%Since both tensor powers of unitaries and permutations of subsystems leave $n$-fold tensor powers of one identical state invariant,

There exists a special group theoretic duality between permutations of subsystems and $n$-fold tensor powers of unitary operators $U^{\otimes n}$, called \textit{Schur-Weyl-Duality} (see, for example \cite{reptheorybook}). This duality emerges from the fact that the two groups' irreducible representations are double commutants; the space of operators commuting with $n$-fold tensor powers $U^{\otimes n}$ is spanned by permutations of the $n$ tensor factors, i.e. the groups determine each other. Schur-Weyl duality is an important tool appearing with applications in various areas of quantum information theory and mathematics, for determining the spectrum of many copies of a density operator \cite{HM02,CG05}, studying the properties of the Haar-random state vector \cite{Harrow05} and (most importantly for this work) proving quantum de Finetti theorems \cite{KM09}.%Many interesting results emerge from Schur-Weyl duality in quantum information theory

One property of Schur-Weyl-duality is that the permutations and the unitaries act on a state in different ways, namely with transversality. While the permutations exchange the whole subsystem, a unitary $U$ acts on a single subsystem, which could also contain multiple qudits, for example $r$. This transversality is sketched in Figure \ref{transversality-perm}.

\begin{figure}[htbp]
   \centering
       \includegraphics[width=10cm]{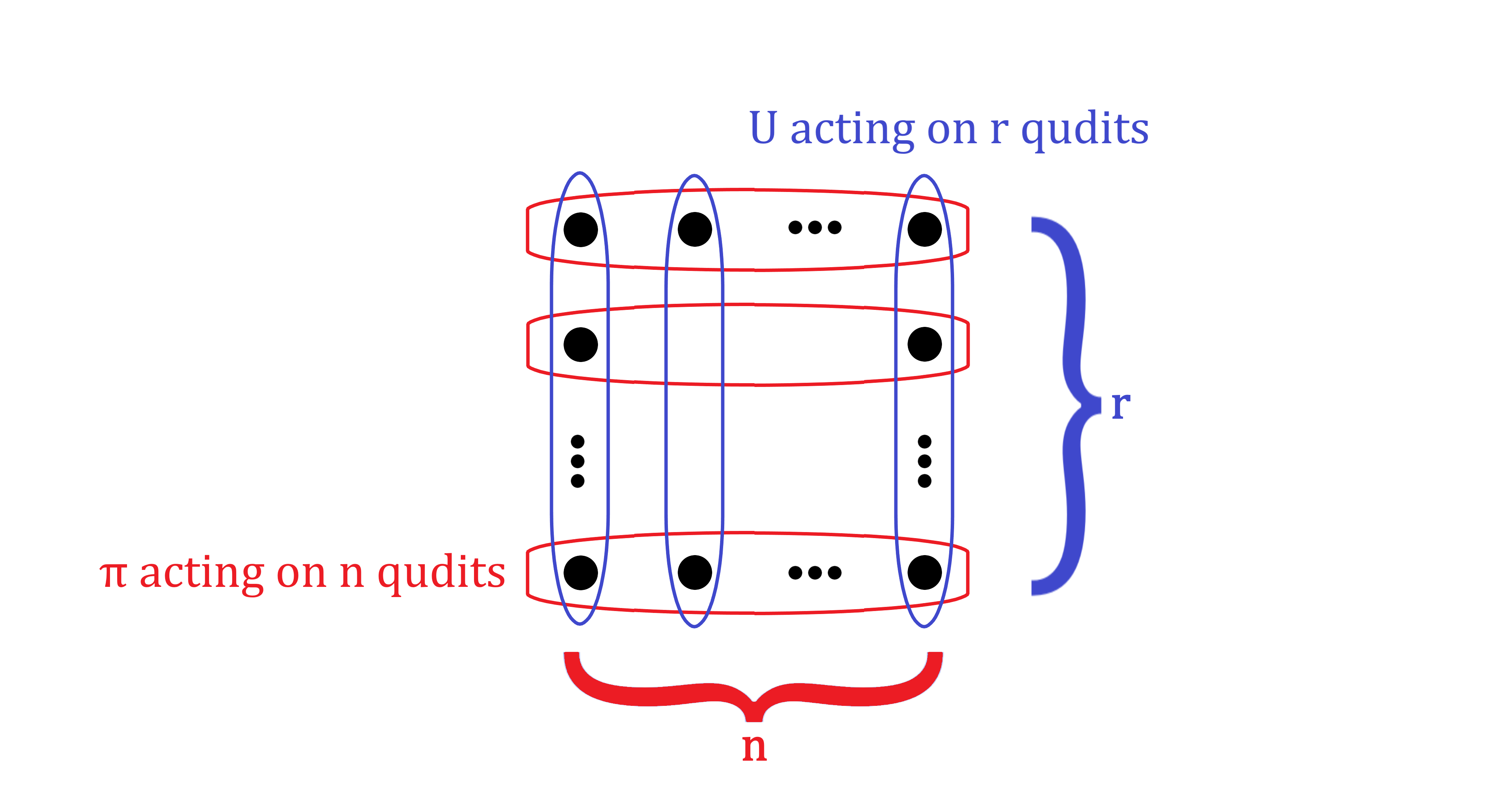}
\caption{Sketch to illustrate how unitaries $U$ and permutations $\pi$ affect the different subsystems of a state. Each black dot corresponds to a qudit on the Hilbert space $\mathcal{H} = \mathbbm{C}^d$. While the unitary operations act on a subsystem containing $r$ qudits, i.e. on a Hilbert space $\mathcal{H}^{\otimes r}$, the permutation operations permute $n$ subsystems, which can be regarded as an action on a Hilbert space $\big(\mathcal{H}^{\otimes r}\big)^{\otimes n}$, or as a separate permutation of all the first qudits in $\mathcal{H}^{\otimes r}$, all the second qubits in $\mathcal{H}^{\otimes r}$, and so on.}
\label{transversality-perm}
\end{figure}

%Schur-Weyl duality appears in various applications in mathematics and quantum information theory, e.g. for .
%Schur-Weyl duality is an important concept that appears in in man

%Schur-Weyl duality is an important concept with many applications in mathematics and quantum information theory, including %This Schur-Weyl duality is intimately related to de Finetti type considerations: \textcolor{green}{Why do I mention this? Importance of Schur Weyl duality, and in de Finetti. maybe one sentence why this relates.}

\subsection{Stochastic Orthogonal Invariance}
\label{newsym}

Many aspects of quantum information theory make use of permutation symmetry and its Schur Weyl duality to the unitary group, and its intimate connection to $n$-fold copies of quantum states. However, one could consider a special, interesting subgroup of the unitary group, the \textit{Clifford group}, which is the group of unitary operations that map the Pauli group onto itself under conjugation. It appears in many subfields of quantum information science and is intimately connected to a special set of states, which are called \textit{stabilizer states}, as they can be generated by applying Clifford operations to the state $\ket{00\cdots 0}$ \cite{Gottesman}.
In fact, as it turns out, there is a close relation between restricting oneself to $n$ copies of Clifford unitaries $U_C^{\otimes n}$, which are a subset of unitary operations $U^{\otimes n}$, and considering the set of $n$-fold tensor powers of stabilizer states $\sigma^{\otimes n }$ instead of the set of $n$-fold tensor powers of arbitrary states.
Stabilizer states are a central object of quantum coding and are frequently used as input states of quantum information processing tasks, for example in entanglement based QKD \cite{Ekert91}, where one is interested in such $n$-fold copies.

To further study the subset of Clifford unitaries and applications connected to it, finding a version of Schur Weyl duality for the Clifford group by identifying the commutant of tensor powers of Clifford unitaries is of great interest, which was achieved by Gross, Nezami and Walter in \cite{Gross}. Since the group of Clifford unitaries is a subgroup of the whole unitary group, its commutant contains permutations, but is not restricted to them. Therefore, to construct the commutant, permutations were used as a basis and extended, which eventually lead to the appearance of a new group that leaves tensor powers of stabilizer states invariant: the \textit{stochastic orthogonal group}. For a Clifford unitary acting on $r$ qudits, the commutant of $n$-fold Clifford tensor powers contains $r$-fold tensor powers of the action of the stochastic orthogonal group, which acts on $n$ qudits. Details can be found in Chapter 4 of \cite{Gross}. 

It must be noted that the commutant of tensor power Clifford unitaries is not exclusively spanned by representations of tensor powers of the stochastic orthogonal group, but also contains tensor powers of orthogonal projections onto CSS codes. However, the additional basis elements are not unitary (and not even invertible), and all unitary basis elements correspond to tensor powers of the stochastic orthogonal group. Therefore, it is sufficient to consider this group in the context of the de Finetti theorem. %However, since all basis elements which are unitary correspond to the stochastic orthogonal group, it is sufficient to consider this group for the purposes of the de Finetti theorem.%, a restriction to the stochastic orthogonal group is admissable, since it is%the commutant's unitary subgroup, the stochastic orthogonal group, suffices. Notably, tensor powers of the stochastic orthogonal group stabilize all stabilizer tensor powers, and also constitute the minimal perfectly complete test for stabilizer states \cite{Gross}.% The other elements of the commutant remain somewhat of a mystery \textcolor{green}{(at least to me)}, and it is an open problem to obtain a complete duality theory in positive characteristic \textcolor{green}{yea right, you don't know what this means!}.)

Importantly, while the Schur-Weyl duality is not exact for all cases, the property of transversality is recovered in this theory: stochastic orthogonal transformations and Clifford unitaries act transversally, as sketched in Figure \ref{transversality-o}. This means that there are generally three parameters appearing: $d$, the dimension of a singular Hilbert space, $r$, the number of such Hilbert spaces affected by a singular Clifford unitary, and $n$, the number of copies of such Hilbert spaces that are transformed by stochastic orthogonal transformations.

\begin{figure}[htbp]
   \centering
       \includegraphics[width=10cm]{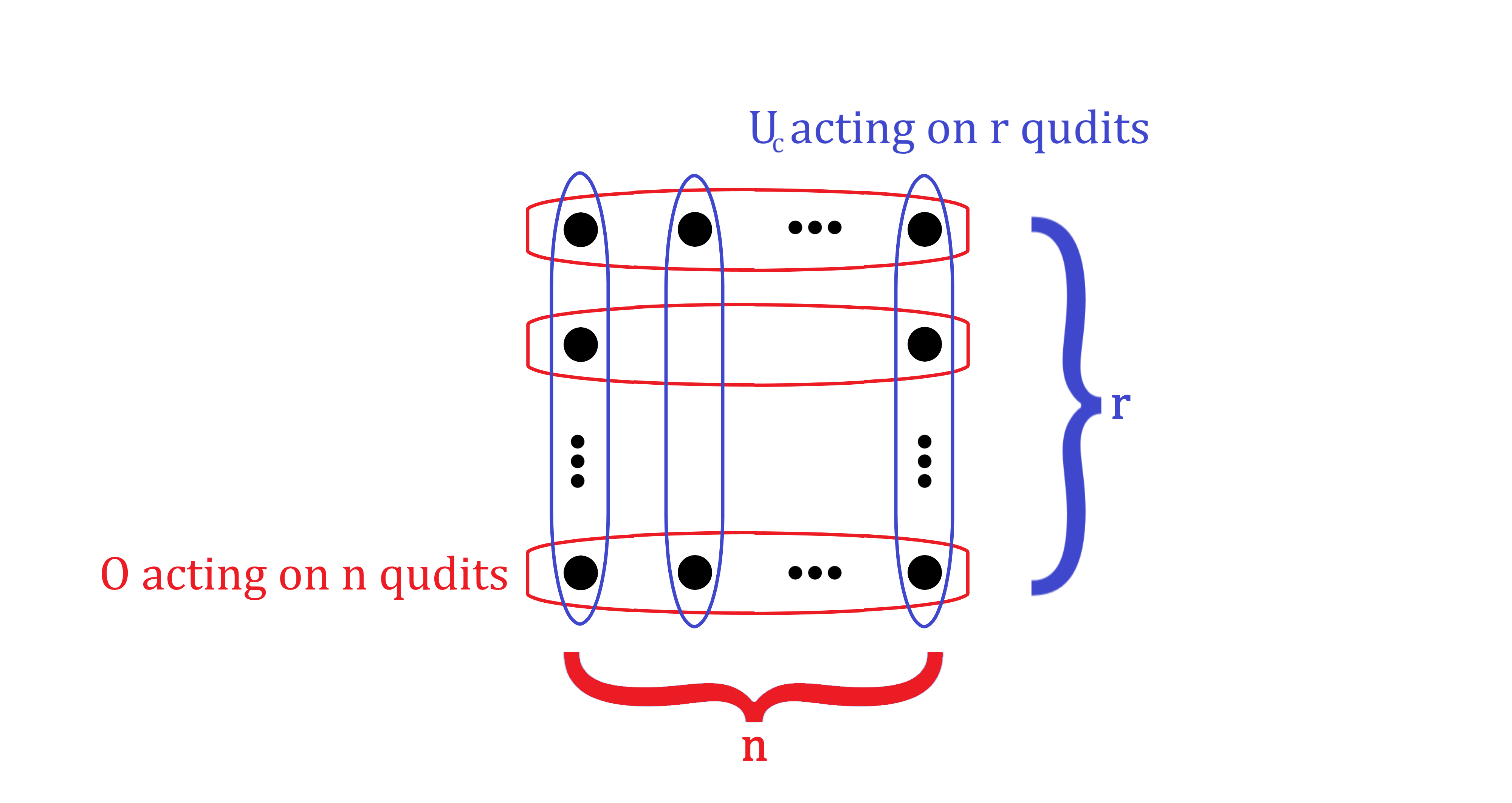}
\caption{Sketch to illustrate how Clifford unitaries $U_C$ and stochastic orthogonal transformations $O$ affect the different subsystems of a state. Each black dot corresponds to a qudit on the Hilbert space $\mathcal{H} = \mathbbm{C}^d$. While the Clifford unitaries act on a subsystem containing $r$ qudits, i.e. on a Hilbert space $\mathcal{H}^{\otimes r}$, the stochastic orthogonal operations act on $n$ such subsystems.
Note the similarities to Figure \ref{transversality-perm}.}
\label{transversality-o}
\end{figure}

The stochastic orthogonal group is defined in the following way:

\begin{definition}[Action of the stochastic orthogonal group] \label{newsym-stochastic}
The action of the stochastic orthogonal group $\mathcal{O}_n(d)$ is defined by the stochastic orthogonal $n\times n$-matrices $O$:%acts on objects in $\mathcal{H}^{\otimes n}$ in the following way:
\begin{itemize}
\item The matrices $O$ are discrete orthogonal: $O^TO=\mathbbm{1} \mod d$
\item The matrices $O$ are stochastic: $O {v_1}={v_1} \mod d$ for the all-ones vector ${v_1}=(1,1,...,1)$ containing $n$ ones.
\end{itemize}
%with an associated quadratic form $q: \mathbbm{Z}_d\mapsto \mathbbm{Z_D}$ defined by $q(x)=x^2$.
\end{definition}

%\begin{remark}
%Note that for $x\in \mathbbm{Z}_d$, which is defined modulo $d$, its square $x^2$ is well defined modulo $D$ (which depends on the quadratic form and d).
%\end{remark}
%., as its minimal projection spans the space of tensor powers of stabilizer states, which we are interested in.

For part of this project (namely, the results in Section \ref{orbitcountingresults}), we consider a slight relaxation of this definition, leading to the \textit{discrete orthogonal group}: %is restricted to the discrete orthogonal group, which is a subgroup of the stochastic orthogonal group:

\begin{definition}[Action of the discrete orthogonal group] 
\label{newsym-ortho}
 The action of the discrete orthogonal group $\tilde{\mathcal{O}}_n(d)$ is defined by the discrete orthogonal $n\times n$-matrices $\tilde{O}$, with
\begin{equation*}
\tilde{O}^T \tilde{O}=\mathbbm{1} \mod d.
\end{equation*}
\end{definition}

This relaxation is justified because we are interested in the particular task of counting orbits, i.e. basis elements which are distinct under the action of the stochastic orthogonal group. The stochastic orthogonal group is a subgroup of the discrete orthogonal group which leaves the all-ones vector invariant. As long as $n$ is not a multiple of the local dimension $d$, there is a direct relation between orbits of the stochastic orthogonal group $\mathcal{O}_n(d)$ and the discrete orthogonal group $\tilde{\mathcal{O}}_{n-1}(d)$. % , therefore counting orbits for the discrete orthogonal group leads to results for the stochastic orthogonal group.
A more detailed justification of this relaxation is given in Section \ref{orbitcountingresults}. 
%permitted by particularities of the method of orbit counting, namely that % justified because the method of orbit counting The justification for this relaxation is given in Section \ref{orbitcountingresults}.
%While this simplifies many things significantly, it is not a problematic restriction because a stochastic orthogonal group is either equal to or can be constructed from a discrete orthogonal group, since they mainly differ by the constraint on the sum of the elements in each row. 
%\textcolor{green}{WHY IS THIS OK - FELIPE?} %In addition, the main implication of such a restriction

The stochastic orthogonal group always contains permutations. In some cases (for example for $(n,d)=(2,d)$, $(3,2)$ or $(3,3)$), the groups are actually equal. 

For qubits, some special group elements can be identified: In addition to the usual permutation matrices, because of the modulo constraint, the binary complement of any permutation will also be part of the group of stochastic orthogonal matrices. These kinds of operations are termed \textit{anti-permutations} in \cite{Gross}. 
For example, the $r$-qubit anti-identity $O_{\tilde{I}} \in \mathcal{O}_6 (2)$ on $n=6$ subsystem copies is the following $n\times n$-matrix:

\begin{equation}
O_{\tilde{I}}=
\begin{pmatrix}
0 & 1 & 1 & 1 & 1 &1  \\
1 & 0 & 1 & 1 & 1 &1  \\
1 & 1 & 0 & 1 & 1 &1  \\
1 & 1 & 1 & 0 & 1 &1  \\
1 & 1 & 1 & 1 & 0 &1  \\
1 & 1 & 1 & 1 & 1 &0  
\end{pmatrix}
\end{equation}

For $n$ copies, the $r$-qubit anti-identity representation (acting on $\big((\mathbbm{C}^d)^{\otimes r}\big)^{\otimes n}$) is given by:

\begin{equation} \label{anti-identity}
R(O_{\tilde{I}})=\frac{1}{2^r} \big( \mathbbm{1}^{\otimes n} + X^{\otimes n} + Y^{\otimes n} + Z^{\otimes n} \big)^{\otimes r}
\end{equation}

There are some particularities connected to stabilizer states which lead to some constraints on $n$ and $d$ for all results in \cite{Gross}, including the de Finetti theorem which is the basis of this project. To counteract this restriction, instead of considering invariance under the stochastic orthogonal group, investigating a subset of non-trivial operations that leave tensor powers of stabilizer states invariant could lead to analogous results for combinations of $n$ and $d$ that were excluded before.
In particular, these restrictions encompass that the theorem does not hold for qubits ($d=2$), which are interesting for numerical studies and easy examples. 
As we will observe in the next section, considering a state's invariance under permutations plus anti-identity leads to a de Finetti statement for qubits.

More details and examples describing the commutant of Clifford unitaries can be found in Chapter 4.3 of \cite{Gross}.

\section{Using Symmetries to Approximate States: de Finetti Theorems}
\label{sec-deFinetti}

Almost anything we could want to know about a quantum system is intimately connected to the system's quantum state. Studying states, and classifying them and their correlations, is therefore a central objective in quantum theory, for example in quantum state tomography \cite{Blume_Kohout_2010} or entanglement certification \cite{Friis_2018}. In quantum information, where one is frequently interested in multiple copies of one state, there is a class of states that is of particular importance: i.i.d.\ states. Given a quantum state $\rho$ on a Hilbert space $\mathcal{H}=\mathbbm{C}^d$, a tensor product of $n$ copies of the state $\rho$, the state $\rho \otimes \rho\otimes \cdots\otimes\rho=\rho^{\otimes n}$ on the Hilbert space $\mathcal{H}^{\otimes n}=\big(\mathbbm{C}^d\big)^{\otimes n}$, would be an example of an i.i.d.\ state, as it is independently and identically distributed over the $n$ subsystems. Any mixture of i.i.d.\ states is also an i.i.d.\ state.

Many quantum information processing tasks assume this structure, and most mathematical framework and results are built around it. This assumption is connected to a very important result and tool in quantum information processing: quantum de Finetti theorems. On its own, a quantum de Finetti theorem describes the closeness of a class of states (usually: permutation invariant states) to an i.i.d.\ state, which can be used to justify an i.i.d. assumption and give an error on it. Thereby, this theorem can also be employed as a mathematical tool for approximating states, for example in terms of a numerical hierarchy.

The classical de Finetti theorem was first introduced in \cite{deFinetti28} and \cite{deFinetti37}, the first of which was translated in \cite{deFinetti_translation}; further details about de Finetti's work on probability theory and statistics can be found in \cite{Cifarelli96}. It is a statement relating symmetric probability distributions to i.i.d.\ probability distributions. More specifically, it states that a marginal distribution (of a potentially small subset of variables) of a symmetric probability distribution is close (with an error $\leq\epsilon$) to a mixture of i.i.d.\ probability distributions. %measure of quantifying how close a probability distribution is to uniform distribution. Quantum analogue. \textcolor{green}{(statement)}

Clearly, a quantum analogue of such a statement, where a marginal of a large symmetric state could be related to an i.i.d.\ state, is of interest for many problems in quantum information processing. First attempts at generalizing the classical de Finetti theorem to a quantum context can be found in \cite{HM76,CFS02}, and subsequently garnered significant interest following \cite{RennerPhD}, where this idea was first explored in the context of its most immediate and obvious application, the security of a QKD protocol with a given symmetry. One important result is the fact that the permutation invariance of a given protocol can be used to prove that its security against general attacks (where an adversary may act on all signals at once and even be entangled with the system) can be inferred from its security against collective attacks (where the adversary acts i.i.d.ly on each signal).

However, this is by far not the only situation where quantum de Finetti theorems have found application. Many quantum information theory problems previously relied on the assumption that the resources are independent and ideally distributed, which can now be scrutinized and often justified via de Finetti type arguments, for example in the study of quantum tomography \cite{ChrRen12}. Furthermore, de Finetti theorems are useful for bounding the diamond norm of a permutation invariant channel \cite{CKR}, and can be employed to provide an alternative proof of quantum Shannon reverse coding theorem \cite{Berta11}. In addition, there is a close connection to the approximation of separable states by a hierarchy of symmetric extensions \cite{DPS04}. Studying the set of separable states is a difficult but ubiquitous problem with application to countless aspects of quantum information theory, and of great importance for improving our understanding of entanglement in general \cite{DPS02}.

Since the earlier versions of the theorem require the number of traced out systems to be rather large (which is especially problematic considering the size of the devices that are currently being developed), an improvement in the form of the exponential de Finetti theorem \cite{Renner07} was proposed. In this version, the state must only be exponentially close to the uniform state. However, the resulting bounds are still largely unattainable in practical implementations.
Several more attempts have been made to generalize and explore the possibilities of this theorem \cite{NOP09,KM09,BH2017,LevCerf09,Lev16,BCS12}.

The most recent and currently best known version of a finite quantum de Finetti theorem is the following:

\begin{theorem}[Quantum de Finetti Theorem, see \cite{CKMR07}]
\label{thm-DeFinetti}    
Let $\rho_{B_1^n}$ be a quantum state on ${(\mathcal{H}_B)}^{\otimes n}={\big({(\mathbbm{C}^d)}\big)}^{\otimes n}$ that commutes with the action of $\mathcal{P}_n$. Let $k\in[1,n]$. Then, there exists a probability distribution $p$ on the set of mixed states on $\mathbbm{C}^d$, such that % of $r$ qudits, such that
\[  \Big\|\rho_{B_1^k} - \int p(\rho_B) \rho_B^{\otimes k}\Big\|_{\tr} \leq 2d^{2}\frac{k}{n}.  \]
\end{theorem}

In \cite{Gross}, a new version with promisingly low error values was proposed, which is the basis of our analysis in this work.
In contrast to first versions, this new de Finetti theorem takes into account a new symmetry beyond permutation symmetry. Instead, it considers protocols that are invariant under stochastic orthogonal symmetry, as introduced in Definition \ref{newsym}. Because of particularities of the stabilizer formalism, this theorem will only hold for some specific cases, namely for odd prime dimensions $d$.

\begin{theorem}[Stabilizer de Finetti Theorem, see \cite{Gross}, Theorem 7.6]
\label{StabDeFinetti}    
Let $\rho_{B_1^n}$ be a quantum state on ${(\mathcal{H}_B)}^{\otimes n}={\big({(\mathbbm{C}^d)}^{\otimes r}\big)}^{\otimes n}$ that commutes with the action of $\mathcal{O}_n(d)$, with $d$ being an odd prime. Let $k\in[1,n]$. Then, there exists a probability distribution $p_S$ on the set of mixed stabilizer states of $r$ qudits, such that
\[ \Big\|\rho_{B_1^k} - \sum_{\sigma} p_S(\sigma_B) \sigma_B^{\otimes k}\Big\|_{\tr} \leq 2d^{2(r+1)^2}d^{-\frac{1}{2} (n-k)}  .\]
\end{theorem}

However, knowledge about the commutant of tensor powers of the Clifford group can also be used to infer a version of the stabilizer de Finetti theorem for a simpler case, dimension $d=2$. In general, a state being invariant under something more than permutations can lead to an alternative version - so a special case to regard is the case of invariance under permutation and one additional group action, the anti-identity introduced in \eqref{anti-identity}. This leads to the following stabilizer de Finetti theorem for qubits:

\begin{theorem}[Stabilizer de Finetti Theorem for Qubits, see \cite{Gross}, Theorem 7.7]
\label{StabDeFinettiQubit}    
Let $\rho_{B_1^n}$ be a quantum state on ${(\mathcal{H}_B)}^{\otimes n}={\big({(\mathbbm{C}^2)}^{\otimes r}\big)}^{\otimes n}$ that commutes with all permutations and the action of the anti-identity on a subsystem consisting of six $r$-qubit blocks. Let $k\in[1,n]$ be a multiple of six. Then, there exists a probability distribution $p_S$ on the set of mixed stabilizer states of $r$ qubits, such that
\[  \Big\|\rho_{B_1^k} - \sum_{\sigma} p_S(\sigma_B) \sigma_B^{\otimes k}\Big\|_{\tr} \leq 6 \sqrt{2}\ 2^r \sqrt{\frac{k}{n}}  .\]
\end{theorem}

%\textcolor{green}{Notably, this de Finetti theorem bound is exponential in the number $k$ of traced out systems, which constitutes a significant improvement over previous versions. Therefore, considering applications of previous quantum de Finetti theorems with this improved bound could also lead to improvements there.
%For example, postselection technique (which evolved from de Finetti theorem), and maximum fidelity....}

%In this part of our project, we will generalize this approach to accommodate general symmetry groups, before applying it to the symmetries studied in \cite{Gross}, combining a promising technique with a promising de Finetti theorem, which leads to significantly improved security bounds.

%Importance? Relevance? Where and why?

%\begin{remark}
For a true comparison between Theorem \ref{thm-DeFinetti} and Theorems \ref{StabDeFinetti} and \ref{StabDeFinettiQubit}, it must be noted that the subspaces which are permuted or orthogonally transformed differ slightly. In the stabilizer de Finetti theorems, each subspace contains $r$ qudits (or qubits) that are transformed by stochastic orthogonal transformations (or permutations and anti-identity). To compare the bounds to the bound of the de Finetti theorem with permutation invariance, one therefore needs to consider a subspace containing $r$ qudits, of which there are $n$ copies, which are permuted. Then, for a state on ${(\mathcal{H}_B)}^{\otimes n}={\big({(\mathbbm{C}^2)}^{\otimes r}\big)}^{\otimes n}$ that commutes with all permutations of the $n$ subsystems, the permutation-based de Finetti theorem in \ref{thm-DeFinetti} holds if the dimension $d$ in the bound is replaced by $d^r$.

Therefore, the bounds which should be compared in the case where $d=2$ and $k\leq n$ is a multiple of 6 are the following:
\[\epsilon_{\text{perm}}=2^{2r+1}\frac{k}{n}\text{ and } \epsilon_{\text{anti-identity}}= 12 \sqrt{2}\ 2^r \sqrt{\frac{k}{n}}\] 

It can be noted that, for qubits, approximating an orthogonally invariant state by a convex combination of stabilizer states is more costly in the limit of large $n$ than an approximation of permutation invariant states by a convex combination of i.i.d.\ states. But while the stabilizer de Finetti theorem for qubits in \ref{StabDeFinettiQubit} leads to no improvement in the convergence (and therefore e.g. error rate for QKD security proofs), it is nonetheless interesting in cases where one is interested in studying stabilizer states specifically (like, for example, in Chapter \ref{chapter-hierarchy}).

For $d$ an odd prime, the following bounds are eligible for comparison:

\[\epsilon_{\text{perm}}=2d^{2r}\frac{k}{n}\text{ and } \epsilon_{\text{ortho}}=2d^{2(r+1)^2}d^{-\frac{1}{2} (n-k)} \] 

As the bound for stochastic orthogonal invariance is exponential in the number of copies $n$, this bound provides a significant improvement in the limit of large $n$. Therefore, using this de Finetti theorem has two advantages, which both motivate this project: On the one hand, it shows improved convergence in the limit of large $n$, which is interesting for QKD error rates (Chapter \ref{chapter-postselec}). On the other hand, it is interesting for problems that could benefit from using stabilizer states as input (like entanglement based QKD, or quantum error correction related problems, Chapter \ref{chapter-hierarchy}).
%\end{remark}

\newpage

\chapter{The Postselection Technique Based on the Stabilizer de Finetti Theorem}

\label{chapter-postselec}
The postselection technique as introduced in \cite{CKR} is a mathematical tool to bound the diamond norm without performing an optimization over the state space, which will be motivated and described in detail in Section \ref{postselec-previous}.
In this chapter, the technique is generalized to accommodate different symmetry groups, which includes developing the necessary mathematical framework in Section \ref{postselectiongeneralized} and investigating its usefulness in QKD settings in Section \ref{horriblecalculation}, before it can be applied to the stochastic orthogonal group via Section \ref{orbitcountingresults} and \ref{StochasticOrbitCounting}.% and used to investigate its usefulness in QKD settings.

\section{Postselection Technique and its Relation to QKD}
\label{postselec-previous}

QKD is the task of generating a string of bits (a \textit{key}) that is only known to two trusted distant parties, Alice and Bob, whilst being completely unknown to an additional party, the adversary Eve. It is assumed that Alice and Bob are linked by an authentic classical communication channel and a potentially insecure quantum channel.

The setting can be described as follows: between them, Alice and Bob ideally share $n$ copies of a quantum state on a $d$-dimensional Hilbert space $\mathcal{H}$, on which they perform measurements to obtain a secret key. In total, they thereby have access to a state on the \textit{trusted} Hilbert space $\mathcal{H}^{\otimes n} \equiv \mathcal{H}_T$, which decomposes into the singular Hilbert spaces $\mathcal{H}$. 
However, the adversary Eve also has access to a Hilbert space of her own, denoted by $\mathcal{H}_E$, and Alice and Bob's input state could be correlated with Eve's state, therefore giving her indirect access to the states that encode Alice and Bob's secret key. Therefore, proving security of a given QKD protocol essentially revolves around bounding Eve's influence on the state on the whole, combined Hilbert space $\mathcal{H}_T\otimes\mathcal{H}_E$, see Figure \ref{Hilbertspaces}.%(The total Hilbert space of Alice, Bob, and Eve is denoted by $\mathcal{H}_T\otimes\mathcal{H}_E$.)

A QKD protocol for two parties is described by a quantum channel, which is a completely positive, trace preserving (CPTP) map $\mathcal{E}$, that transforms Alice and Bob's shared input state into two keys. 
The security of a QKD protocol $\mathcal{E}$ is defined through a comparison between such a quantum channel $\mathcal{E}$ and an ideal version of the same protocol $\mathcal{F}$, which transforms the same input state into two identical keys, with no information about those keys leaking to the eavesdropper Eve. The closer the actual protocol $\mathcal{E}$ is to the ideal protocol $\mathcal{F}$, the more secure it is. In other words, if the distance between the two CPTP maps $\mathcal{E}$ and $\mathcal{F}$ is very small, while taking Eve's access into account, the protocols are approximately equal and $\mathcal{E}$ is approximately secure.

Mathematically, a natural measure of security is therefore given by a difference between two CPTP maps in terms of the diamond distance.

\begin{figure}[H]
   \centering
       \includegraphics[width=7.5cm]{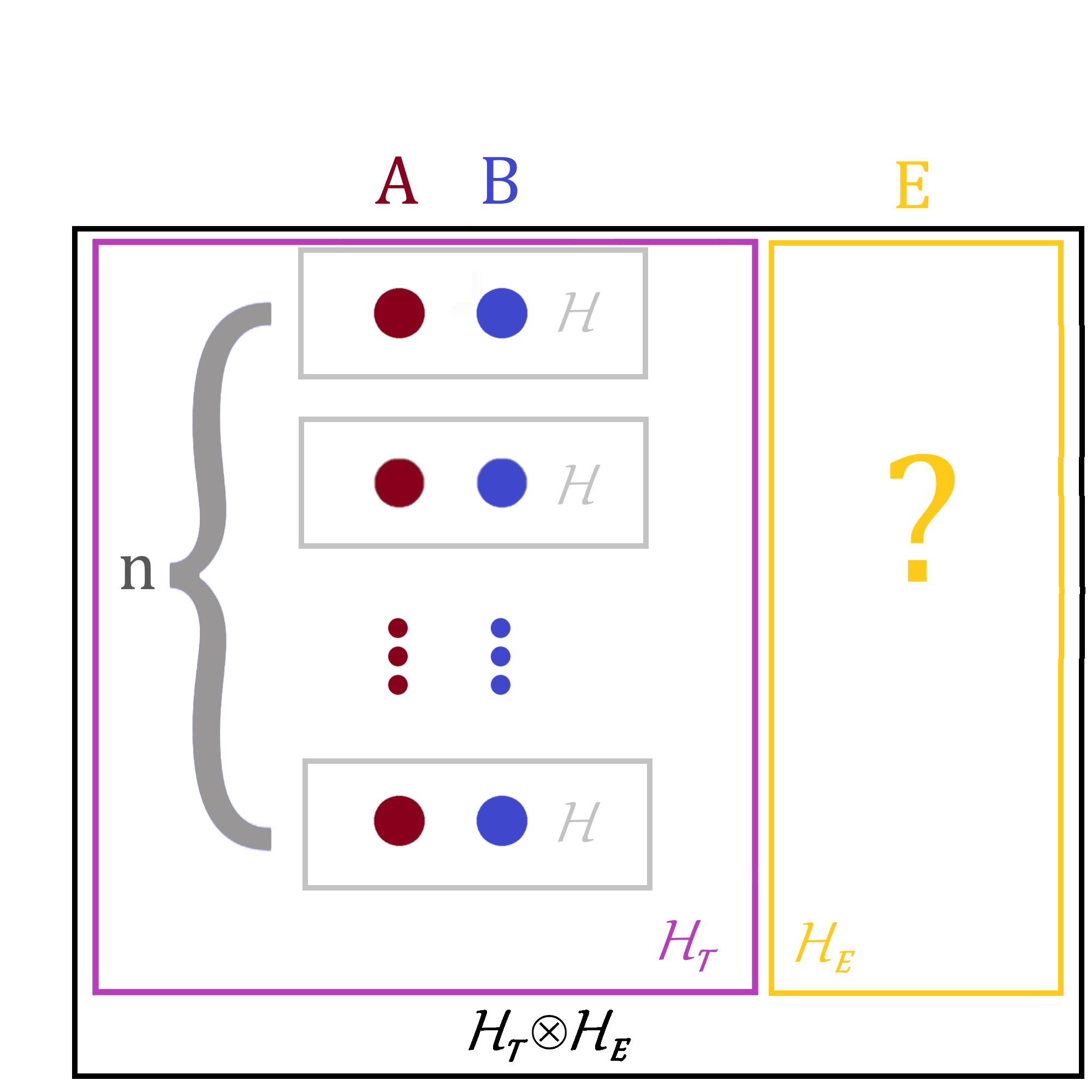}
\caption{Sketch of the parts making up the total Hilbert space of Alice, Bob and Eve in a QKD scheme. The pair of Alice and Bob shares $n$ copies of a quantum state on the $d$-dimensional Hilbert space $\mathcal{H}$, thereby having access to the entire trusted Hilbert space $\mathcal{H}_T=\mathcal{H}^{\otimes n}$. Eve, the untrusted party and potential eavesdropper, also has access to some Hilbert space $\mathcal{H}_E$, which the trusted parties know nothing about. In total, the entire Hilbert space therefore consists of the combination of Alice and Bob's and Eve's space: $\mathcal{H}_T\otimes \mathcal{H}_E$. Usually, it is assumed that the worst case scenario applies, where Eve has access to a complete copy of Alice and Bob's Hilbert space, i.e. $\mathcal{H}_E=\mathcal{H}_T=\mathcal{H}^{\otimes n}$. (With less, she could not properly entangle herself with each of the trusted qudits.)}
\label{Hilbertspaces}
\end{figure}

\begin{definition}[Diamond distance between two CPTP maps]
\label{diamondnorm}
Let $\Delta=\mathcal{E}-\mathcal{F}$ be a difference between CPTP maps $\mathcal{E}$ and $\mathcal{F}$ acting on the Hilbert space $\mathcal{H}_T$, let $\mathcal{H}_E$ be a Hilbert space, and $\rho_{TE} \in \mathfrak{S}(\mathcal{H}_T \otimes \mathcal{H}_E)$ be a quantum state. Then, the diamond distance between the maps, i.e. the diamond norm of $\Delta$, is given by
\[\big\| \Delta\big\|_{\diamond}= \sup_{\rho_{TE} \in \mathfrak{S}(\mathcal{H}_T \otimes \mathcal{H}_E)} \big\| (\Delta \otimes \mathbbm{1}_E) \rho_{TE}\big\|_{\tr} \]
\end{definition}

The trace norm in the above definition is defined by: $ \big\| \rho \big\|_{\tr}= \frac{1}{2} \big\| \rho \big\|_{1}= \frac{1}{2}\tr \big( \sqrt{\rho^{\dagger} \rho}\big)$.

 %\textcolor{green}{MAYBE WHY IS THIS GOOD MEASURE}
In principle, the diamond norm constitutes taking two suprema, one over the input state, and one over the dimension of the space $\mathcal{H}_E$ (Eve's space) that the identity acts on; however, for positive quantum states, the suprema are reached for $\mathcal{H}_E$ having equal dimension to $\mathcal{H}_T$ \cite{Kitaev97}, which we suppose for the QKD analysis.
%As mentioned before, while Eve's space is in principle a mystery to the trusted parties, it is supposed 
Using this distance measure, security of a protocol is then defined by comparing the diamond distance of a protocol $\mathcal{E}$ and a perfect version of the protocol $\mathcal{F}$ to some small parameter $\epsilon$ that bounds the probability of not obtaining perfectly identical, secret keys for Alice and Bob.
\begin{definition}[$\epsilon$-security] \label{security}
A protocol $\mathcal{E}$ is $\epsilon$-secure if \[\big\|\mathcal{E}-\mathcal{F}\big\|_{\diamond} \leq \epsilon.\]
\end{definition}

Clearly, this kind of comparison between an actual protocol and a perfect version (where Alice and Bob share i.i.d.\ states that are decoupled from an adversary's state) is closely related to quantum de Finetti theorems, where a marginal of a large symmetric state can be approximated by an i.i.d.\ state (more precisely, a mixture of i.i.d.\ states). In fact, this has been a key motivation for studying quantum de Finetti theorems in the first place \cite{RennerPhD}. Thereby, it can be proven that the security of a protocol against a collective attack implies its security against a much more powerful general attack, at the cost of an overhead factor. In most quantum de Finetti theorems, tracing out a small number of systems leads to unattainable security parameters (large $\epsilon$ and impossible key lengths for practical purposes)  \cite{Hayashi06,SR08}. %In other words, a large number of subsystems (almost all) must be traced out to get any useful result.
Subsequent improvements on the bound of the theorem resulted in improvements of the corresponding security parameters, but are in general still far from useful for current applications.% Improving the bound of the theorem is 

Another alternative way of improving security bounds, in particular the additional factor between collective and general attacks, emerged in the form of the postselection technique \cite{CKR} discovered by Christandl, König and Renner.
%Instead of focussing on finding de Finetti theorems with improved convergence, security bounds can also be analyzed by using other techniques that rely on de Finetti type considerations. Most notably, the postselection technique \cite{CKR} discovered by Christandl, König and Renner constitutes a tremendous improvement in terms of 
By definition, computing the diamond norm in principle entails an optimization over a large number of states. However, this can be circumvented by the postselection technique, which showed that it is sufficient to consider a single input state $\tau_{TR}$. Namely, if the map $\Delta$ is invariant under the group of permutations $\mathcal{P}_n$, the diamond norm of $\Delta$ is bound in the following way:

\begin{equation}
\big\| \Delta\big\|_{\diamond} \leq g_{n,d} \big\| (\Delta \otimes \mathbbm{1}_R) \tau_{T R} \big\|_{\tr}
\end{equation}

with

\begin{equation} \label{permutationsgtd}
g_{n,d}= \left( \begin{array}{c} d^2-1+n \\\ n \end{array}\right) \leq (n+1)^{d^2-1}.
\end{equation}

The state $\tau_{T {R}}$ is the purification of a particular input state, which is called the {\textit{de Finetti state}}, with a specific form:
\begin{equation}\label{Tau-T-Perm}
\tau_{T}= \int \rho^{\otimes n} \mu(\rho)%_{\mathcal{H}})
\end{equation}
with $\rho \in \mathfrak{S}({\mathcal{H}})$, where $\mathcal{H}$ is a $d$-dimensional Hilbert space.  $\mu$ is the measure induced by the Hilbert-Schmidt metric on a single subsystem $\operatorname{End}(\mathcal{H})$.

Apart from its application to QKD, this is an interesting mathematical bound on the diamond norm which is useful in any scenario where the distinguishability of quantum operators is of interest \cite{7464342,Cao_2015}.
%\textcolor{green}{In itself, this is an interesting mathematical bound.  should i write more here?}
%with $\rho_t \in \mathfrak{S}({\mathcal{H}_t})=\mathfrak{S}({\mathcal{H}})$ with a $d$-dimensional Hilbert space $\mathcal{H}$ and $\mu$ being the measure induced by the Hilbert-Schmidt metric on a single subsystem $\operatorname{End}(\mathcal{H})$. %For the $n$-fold tensor power of the single subsystem Hilbert space $\mathcal{H}$, we introduce the following short-hand notation: $\mathcal{H}^{\otimes n} \equiv \mathcal{H}^{\otimes n}$. 
%The proof of this statement has two steps: Firstly, due to the symmetry of $\Delta$, it is sufficient to consider only states with support on the invariant subspace 
%$(\mahcal{H}^{t})^s \subset \mathcal{H}^{\otimes n}$. Secondly, any state with support on $(\mathcal{H}^{\otimes n})^s $ can be obtained from a purification of $\tau_{T}$.
Note that an alternative version and proof of this bound can be found in \cite{Fawzi15}.
The technique was generalized and adapted to continuous variable schemes by Leverrier \cite{Lev17}. In the case of continuous variables, there is an additional step to replace the total Hilbert space with a finite dimensional Hilbert space (``energy test").

In the next section, we will show that the technique can also be generalized to accommodate symmetry groups beyond permutation invariance, leading to an analogous bound on the diamond norm that depends on the dimension of an invariant subspace.

%START

\section{Generalized Postselection Technique for a Symmetry Group $\mathcal{S}$}
\label{postselectiongeneralized}

Generalizing de Finetti type arguments for symmetries beyond permutation is not a novel concept \cite{Lev17,BCS12}. In fact, there is a short comment in the outlook and appendix of \cite{CKR} itself which outlines how the postselection technique can be generalized to arbitrary symmetries. Nonetheless, it can be considered interesting to analyze this in more detail and scrutinize the necessary steps and assumptions. There is one main assumption, the resolution of identity, which we will investigate at the beginning of this section and in Appendix \ref{appendix-resolutionofidentity}. All necessary lemmata for the proof of the bound on the diamond norm for arbitrary symmetry groups $\mathcal{S}$ can be found in  Appendix \ref{appendix-PostselecLemma}.

As a first step of extending the postselection technique to a symmetry group $\mathcal{S}$, a generalization of the de Finetti state must be considered. In analogy to \eqref{Tau-T-Perm}, such a (symmetry-dependent) \textit{de Finetti state of $\mathcal{S}$} is given by $\tau_T\in\mathfrak{S}(\mathcal{H}_T)=\mathfrak{S}(\mathcal{H}^{\otimes n})$ :
\begin{equation}
\label{deFinettistate}
\tau_{T}= \int \rho^{\otimes n} D_{\mathcal{S}} (\rho)
\end{equation}

with states $\rho \in \mathfrak{S}({\mathcal{H}} )$ and a symmetry-dependent integration measure $D_{\mathcal{S}}$. %To bound the diamond norm with this state, it has to be assumed that an integration measure exists that fulfills the folllowing constraints.

The assumption that justifies applying the postselection technique is not tied to $\tau_T$ directly, but to a purification of it, where each of the $n$ subsystems has been purified separately. This defines the state $\tau_{TE}\in\mathfrak{S}(\mathcal{H}_T\otimes \mathcal{H}_E)=\mathfrak{S}(\mathcal{H}^{\otimes n} \otimes \mathcal{H}^{\otimes n})$:
\begin{equation}\label{tau-TE}
%\tau_{TE}= \int \rho_{TE}^{\otimes n} {d}_{\mathcal{S}} (\rho_{TE})
\tau_{TE}= \int \tilde{\rho}^{\otimes n} {d}_{\mathcal{S}} (\tilde{\rho})
\end{equation}
with the pure states $\tilde{\rho}\in\mathfrak{S}(\mathcal{H}\otimes \mathcal{H})$.

Then, for the postselection technique to be applicable, this state has to fulfill the following relation, called {\textit{resolution of identity}}:
\begin{equation}
\label{resolutionofid-eq}
\tau_{TE}= \frac{1}{g_{n,d}} \mathbbm{1}_{  (\mathcal{H}^{\otimes n} \otimes \mathcal{H}^{\otimes n})^{(s\otimes\overline{s})} } =\frac{1}{\dim(\mathcal{N})} \mathbbm{1}_{{N}}
\end{equation}
with
\begin{equation}
g_{n,d}= dim ( (\mathcal{H}^{\otimes n} \otimes \mathcal{H}^{\otimes n})^{(s\otimes\overline{s})} ) =\dim(\mathcal{N}).
\end{equation}
%SATZZEICHEN

In other words, the state $\tau_{TE}$ must be maximally mixed on the $\mathcal{S}$-invariant subspace $(\mathcal{H}_T \otimes \mathcal{H}_E)^{(s\otimes\overline{s})}\equiv \mathcal{N}$. This implies that the symmetry-dependent integration measure must be invariant under the symmetry $\mathcal{S}$.

Since the resolution of identity is a key ingredient to proving a bound on the diamond norm using the postselection technique, it must therefore be assumed that an integration measure $d_{\mathcal{S}}(\cdot)$ exists such that (\ref{resolutionofid-eq}) holds and $D_{\mathcal{S}}(\cdot)$ exists.
Since all states $\tilde{\rho}$ are pure, $d_{\mathcal{S}}(\cdot)$ is a measure on pure states. Then, the existence of an integration measure $D_{\mathcal{S}}(\cdot)$ can be inferred from the existence of $d_{\mathcal{S}}(\cdot)$. %\textcolor{green}{(How? $D(tr(\rho_{HH}))=d(\rho_{HH})$? measure on pure states $\rightarrow$ measure on mixed states)}
To allow for a resolution of identity, the integration measure $d_S(\cdot)$ must be invariant under the symmetry $\mathcal{S}$.

Rephrased in mathematical terms, the postselection technique can only be applied for a symmetry group $\mathcal{S}$ if the following condition holds: 

 \begin{equation} \label{Cond0}
\operatorname{Cond0}=\Big\{ \exists d_{\mathcal{S}} (\cdot)\  \operatorname{s.t.} \ \int \tilde{\rho}^{\otimes n} {d}_{\mathcal{S}} (\tilde{\rho})= \frac{1}{g_{n,d}} \mathbbm{1}_{  {N} } \Big\} 
\end{equation}

For the symmetric group $\mathcal{P}_n$ with permutations as its representation, 
 the required integration measures are the one induced by the Hilbert-Schmidt metric as $D_{\mathcal{P}_n}(\cdot)$, and the one induced by the Haar measure on the unitary group acting on $\mathcal{H} \otimes \mathcal{H}$ as ${d}_{\mathcal{P}_n}(\cdot)$. Then, since the measure for $\tau_{TE}$ is invariant under permutations and under unitary group (its dual under Schur-Weyl-duality), a resolution of identity holds because of Schur's lemma \cite{CKR}.

 % for the symmetric group $S(t)$ \cite{CKR}, where $\rho_t$ can be any valid density operator - in this case, $\tau_{T}$ is an indepently and ideally distributed state with prototype $\rho_t$. In the case of continuous variables, we require a resolution of identity

In the case of continuous variables \cite{Lev17}, another integration measure is needed. Instead of independent and ideally distributed states, the states of interest are the general coherent states $\rho^{\otimes n}=\ket{\Lambda,n} \bra{\Lambda,n}=(\ket{\Lambda,1}\bra{\Lambda,1})^{\otimes n}$.
For such states, an invariant measure on the corresponding space 
%For such states, the invariance of a measure $d\mu$
is established in \cite{Per86}, and their resolution of identity relies on a version of Schur's lemma for general unimodular groups with a square-integrable representation (such as $SU(p,q)$) \cite{Lev16}. In addition, since a truncation of the Hilbert space is performed, it has to be shown that the finite dimensional truncated space also incorporates an (approximate) resolution of identity \cite{Lev17} to make bounding the diamond norm possible.

%\textcolor{green}{MAYBE WRITE ABOUT DISCRETE ORTHO GROUP HERE}

However, when extending the postselection technique to other symmetry groups, this condition must also be met. Therefore, it may be instructive and helpful to rephrase the assumption using conditions in linear algebra. In Appendix \ref{appendix-resolutionofidentity}, two necessary, but not sufficient conditions are given.

Thus, if the resolution of identity holds for a given symmetry, it can be shown that it is sufficient to consider the particular (symmetry-dependent) state $\tau_{T}$ when computing the diamond norm of an $\mathcal{S}$-invariant map, instead of performing an optimization over a large number of states. For the purpose of this section, we assume that the symmetry group $\mathcal{S}$ fulfills the condition; later, when we apply the postselection technique to the stochastic orthogonal group, we will find that resolution of identity holds for this case.

Given the following preliminary lemmata, which are generalizations of lemmata found in \cite{RennerPhD} and \cite{CKR}, the proof of the bound on the diamond norm becomes rather concise. The first lemma shows that it is sufficient to consider states with support on the invariant subspace instead of arbitrary states for the diamond norm of an invariant map.

%jump
\begin{lemma}
\label{lemmasym}

 %For any finite-dimensional space $\mathcal{R}_1$ and any (arbitrary) density operator $\sigma_{\mathcal{H}^{\otimes n} \mathcal{R}_1}$, the following holds:
Let $\Delta$ be a linear map from $\operatorname{End}(\mathcal{H}^{\otimes n})$ to $\operatorname{End}(\mathcal{H}')$ that is invariant under the symmetry $\mathcal{S}$.
For any finite-dimensional space $\mathcal{M}$ and any (arbitrary) density operator $\sigma_{TM}$, the following holds:

\begin{equation*}
\big\| \  (\Delta \otimes \mathbbm{1})  \sigma_{TM} \big\|_{\tr}
 \leq
 \big\| (\Delta \otimes \mathbbm{1}) \rho_{TE} \big\|_{\tr}
\end{equation*}

where $ \rho_{TE}$  is a state with support on $\mathcal{N}=(\mathcal{H}^{\otimes n} \otimes \mathcal{H}^{\otimes n})^{(s\otimes \overline{s})}$.
\end{lemma}

Then, the second lemma establishes a connection between states with support on the invariant subspace and the de Finetti state in \eqref{tau-TE}.

\begin{lemma}
\label{lemma2}
Suppose we have a state $\rho_{TE}$ with support on the subspace $\mathcal{N}=  (\mathcal{H}^{\otimes n} \otimes \mathcal{H}^{\otimes n})^{(s\otimes\overline{s})}\subseteq \mathcal{H}^{\otimes n} \otimes \mathcal{H}^{\otimes n}$. For any such state, there exists a linear completely positive trace-nonincreasing map $\mathcal{C}: \operatorname{End}(\mathcal{N}) \to \mathbbm{C}$ such that
\begin{equation*}
\rho_{TE}= g_{n,d} (\mathbbm{1}_{TE} \otimes \mathcal{C}) (\tau_{TEN})
\end{equation*}
with $\operatorname{tr}_{\mathcal{N}} \tau_{TEN} = \tau_{TE}=\frac{1}{g_{n,d}} \mathbbm{1}_{ N}$  \eqref{tau-TE}, and $g_{n,d}=\dim(\mathcal{N}) $.
\end{lemma}

A more precise description and the proofs of these lemmata can be found in Appendix \ref{appendix-PostselecLemma}.
Using these two lemmata, the main theorem can be proven, which constitutes a mathematical bound on the diamond norm of an invariant map:

\begin{theorem}[Bound on the Diamond Norm]
\label{ImportantThm}
For a linear map $\Delta: \operatorname{End}( \mathcal{H}_T) \to \operatorname{End}(\mathcal{H}')$ that is invariant under symmetry group $\mathcal{S}$ for which \eqref{Cond0} holds, and a purification $\tau_{TR}$ of $\tau_{T}$ as given in \eqref{deFinettistate},
\begin{equation}
\big\| \Delta\big\|_{\diamond} \leq g_{n,d} \big\| (\Delta \otimes \mathbbm{1}_{{R}}) \tau_{TR} \big\|_{\tr}
\end{equation}
with $g_{n,d}=\dim (\mathcal{N})$ being the dimension of the invariant subspace $ \mathcal{N} \equiv (\mathcal{H}^{\otimes n}\otimes \mathcal{H}^{\otimes n})^{(s\otimes \overline{s})} \subseteq  \mathcal{H}^{\otimes n}\otimes \mathcal{H}^{\otimes n}$.%, where $ Sym_{\mathcal{S}} (\mathcal{H}^{\otimes n}) \subset (\mathcal{H})^{\otimes n}$  is the subset that is symmetric under the symmetry group $\mathcal{S}$.
\end{theorem}

%Using the lemmata from Appendix \ref{appendix-PostselecLemma}, the proof of the above becomes short and concise.

\begin{proof}[Proof of Theorem \ref{ImportantThm}]

We refer to the definition of the diamond distance (\ref{diamondnorm}) of a linear map $\Delta$.
Let $\sigma_{TM}\in\mathfrak{S}(\mathcal{H}_T \otimes \mathcal{M})$ be an arbitrary state, associated to a finite space $\mathcal{M}$. Let $\rho_{TE}$ denote a state with support on the $\mathcal{S}$-invariant subspace $\mathcal{N}$, $\mathcal{C}$ be a CPTP map, and $g_{n,d}=\dim(\mathcal{N})$. Using the fact that $\Delta$ is invariant under $\mathcal{S}$, we find:
\begin{equation*}
\big\| \  (\Delta \otimes \mathbbm{1}_{{M}})  \sigma_{TM} \big\|
% \leq
\overset{\mathrm{Lemma\ \ref{lemmasym}}}{=} \big\| (\Delta \otimes \mathbbm{1}_{E}) \rho_{TE} \big\|
\end{equation*}
\begin{equation*}
\overset{\mathrm{Lemma\ \ref{lemma2}}}{=} \big\| (\Delta \otimes \mathbbm{1}_{E} \otimes \mathbbm{1}_{{N}}) (g_{n,d} (\mathbbm{1}_{TE} \otimes \mathcal{C}) (\tau_{TEN})) \big\|
\end{equation*}
\begin{equation*}
\overset{\mathrm{g_{n,d} \geq 0}}{=} g_{n,d} \big\|  (\Delta \otimes \mathbbm{1}_{E}\otimes \mathcal{C}) (\tau_{TEN})) \big\|
\end{equation*}
\begin{equation*}
\leq g_{n,d} \big\|  (\Delta \otimes  \mathbbm{1}_{EN}) (\tau_{TEN})) \big\|
\end{equation*}
\begin{equation*}
\overset{\mathrm{\mathcal{H}^{\otimes n} \otimes \mathcal{N} \equiv \mathcal{R}}}{=}g_{n,d} \big\|  (\Delta \otimes  \mathbbm{1}_{{R}}) (\tau_{AR})) \big\|
\end{equation*}

where the second-to-last step is possible because $\mathcal{C}$ is trace-nonincreasing (by construction).
\end{proof}

%\begin{remark}
In summary, as long as the assumption of the existence of a resolution of identity holds, the bound on the diamond norm for an arbitrary symmetry $\mathcal{S}$ is directly given by the dimension $g_{n,d}=\dim (\mathcal{H}^{\otimes n} \otimes \mathcal{H}^{\otimes n})^{(s\otimes\overline{s})}=\dim(\mathcal{N})$ of an invariant subspace, and the symmetry dependent state state $\tau_T$.
%\textcolor{green}{write sth more?}

%START

\section{Generalized Postselection Technique Applied to QKD}
\label{horriblecalculation}

As mentioned in Section \ref{postselec-previous}, there have been various attempts in quantum cryptography to use the security of a protocol against collective attacks to infer security against general, more powerful attacks via quantum de Finetti theorems by exploiting the permutation symmetry of the protocol \cite{RennerPhD,Renner07}. The postselection technique can be applied to the same problem, and significantly improved previously known security bounds for permutation-invariant protocols. Here, it is shown that the generalized postselection theorem (Theorem \ref{ImportantThm}) may be used in the same manner, to derive new (and, as we will see in Section \ref{orbitcountingresults} for discrete/stochastic orthogonal symmetry, tighter) bounds for a protocol with a different symmetry.

First, in Section \ref{qkdprotocols}, the general steps of a QKD protocol (for two parties and $N$ parties) will be described, with special focus on the step of privacy amplification and its relation to the length of the final, secret key shared by Alice and Bob. Then, in Section \ref{cohtocoll}, we will show how the symmetry of a protocol can be used to infer security against general attacks from security against collective attacks via the postselection theorem.

\subsection{QKD Protocols, Min-entropy and Privacy Amplification}
\label{qkdprotocols} 

A QKD protocol refers to the task of establishing a common secret key between two parties, Alice and Bob, while ensuring that a potential adversary Eve has no knowledge about it.
A typical two-party QKD protocol consists of the following steps:

\begin{itemize}
\item[1)] \textbf{Key exchange.} The two parties exchange qubits and perform measurements to generate the \textit{raw keys}.
\item[2)] \textbf{Key sifting.} Only certain cases out of the raw key are selected and kept (e.g. the cases where some particular measurements were made), the rest is discarded of. The resulting bit sequence is called \textit{sifted keys}.
\item[3)] \textbf{Key distillation.}
\begin{itemize}
\item[a)] \textbf{Parameter Estimation.} Some random bits are selected and announced publicly via the classical communication channel to estimate the error rate. If the error rate exceeds some limit, the protocol will abort.
\item[b)] \textbf{Information Reconciliation.} {Error correction} is performed to transform the keys into identical bitstrings. An adversary might still have some information about the resulting key.
\item[c)] \textbf{Privacy Amplification.} The two parties use two-universal hashing to get a shorter, but secret key.
\end{itemize}
\end{itemize}

In an $N$ party protocol (for example $N$-six-state protocol \cite{EKMB17} and $N$-BB84 \cite{Grasseli18}), the goal is to establish a secret key known to all $N$ trusted parties, but unknown to Eve. The $N$ parties consist of one Alice, and $N-1$ Bobs. Here, genuinly multipartite entangled states are shared between the parties, and all parties perform local measurements to collect a raw key. Similarly to the two party protocol, the parties reveal some random bits to estimate the error rate of their channel. Then, Alice performs an information reconciliation procedure with each Bob to ensure that they have identical bit sequences, before each party applies the same randomly chosen hashing function during the privacy amplification step to ensure their key's secrecy.

Each of the steps can contribute to the overall error rate of the protocol, and can separately be bounded, but because of composability \cite{RennerPhD} the different bounds can be added together. For the application treated here, the most notable and important step is \textit{privacy amplification}, introduced in Chapter 5 of \cite{RennerPhD}. After obtaining two identical, perfectly correlated bit strings in the information reconciliation step, the two parties perform a series of operations which produces two shorter, perfectly correlated and perfectly secret keys. With this step, it is ensured that an eavesdropper will have no information about the key shared between the two trusted parties if it was shortened by some set amount (or more). This amount is determined by the Leftover Hashing Lemma (originally introduced as Theorem 5.5.1. in \cite{RennerPhD}, and stated using more up-to-date definitions in \cite{TomLev17}).
%Privacy amplification theorem, and stated more formally and up-to-date as Leftover Hashing Lemma in \cite{TomLev17}. 

To arrive at such a bound, the amount of information about the key that is available to an eavesdropper has to be determined, typically in terms of conditional entropy. In general, the adversary Eve may have gained access to correlated side information during previous steps of the QKD protocol. This information may include access to a quantum state holding memory of interfering with the quantum communication in previous steps. Then, conditional entropy measures the amount of uncertainty Eve perceives in the key, while taking into account the side information available to her.

Different assumptions about Eve's knowledge can entail different measures of entropy, which can lead to different bounds on the security of a protocol. The most commonly used measure of entropy is von Neumann entropy - however, this would relate to Eve attempting to guess the key from taking a classical average, which is not optimal.
The first attempt to eliminate Eve's knowledge by reducing key size used Renyi collision entropy \cite{RennerPhD}, where Eve could use the quantum state of her subsystem to construct good measurements. Although carrying out these measurements will provide her with a good guess, it is not optimal.
A natural generalization of conditional Renyi entropy is conditional min-entropy, first proposed in \cite{RennerPhD} and expanded upon in \cite{KR09}, where Eve makes use of the best possible measurements, giving her the best possible probability of guessing the key. This is optimal for her, and the ``worst case scenario" for us.

It is important to note that this bound on Eve's information is relevant at a certain point during the protocol, namely directly before privacy amplification is carried out. Therefore, while we take $\mathcal{H}_T$ to be the space of the input state to the overall protocol, we use $\mathcal{H}_X$ to refer to the space that the state is on after all previous steps and before the privacy amplification step.

\begin{definition}[Min-entropy] \label{minentropy}
Let $\rho_{XE}$ be a state on $\mathcal{H}_X \otimes \mathcal{H}_E$. The min-entropy of $\rho_{XE}$ given $E$ is
\[H_{min}(X|E)_{\rho}= \sup_{\sigma_E \in \mathfrak{S}(\mathcal{H}_E)} \big( \min_{\lambda \in \mathbbm{R}} \{ -\log(\lambda) | {\rho}_{XE} \leq \lambda (\mathbbm{1}_{X}\otimes \sigma_E)\}  \big).\]
\end{definition}

Notably, the computation of the min-entropy can be transformed into a SDP, which can be solved efficiently numerically \cite{Tom16}. Firstly, note that the above definition can be rewritten to read:

\[ H_{min}(X|E)_{\rho}= \min_{\sigma_E \in \mathfrak{S}(\mathcal{H}_E)} \{ -\log(\operatorname{tr} \sigma_E ) | \rho_{XE} \leq \mathbbm{1}_X \otimes \sigma_E \}\]
 
which can directly be translated to the following primal problem:

\begin{optimize} \label{SDP}
\begin{equation*} \begin{split}
\text{minimize } & \operatorname{tr} \sigma_E   \\
\text{subject to }& \rho_{XE}\leq \mathbbm{1}_X \otimes\sigma_E \\
& \sigma_E\geq 0\\
\end{split}
\end{equation*}
\end{optimize}

and the corresponding dual problem:

\begin{optimize}
\begin{equation*} \begin{split}
\text{maximize } &  \operatorname{tr} (\Lambda_{XE}  \rho_{XE})   \\
\text{subject to }& \operatorname{tr}_{X} \Lambda_{XE}\leq \mathbbm{1}_E \\
& \Lambda_{XE} \geq 0\\
\end{split}
\end{equation*}
\end{optimize}

Both of the objective functions of these SDP characterizations evaluate to $2^{-H_{min}(X|E)_{\rho}}$ and can thus be used to compute conditional min-entropy (which we will make use of later).
However, min-entropy is very sensitive to changes of the system's state. Since most applications entail some small error probability, this makes min-entropy a less desireable measure. Instead, one introduces another generalized entropy measure, smooth min-entropy. This entropy measure takes into account a ball of states that are close to $\rho_{ABE}$ in terms of purifying distance, thereby accounting for some deviations in the system's state.

\begin{definition}[Smooth min-entropy] \label{smoothminentropy}
Let $\rho_{XE}$ be a state on $\mathcal{H}_X \otimes \mathcal{H}_E$ and let $\mathcal{B}^{\tilde{\epsilon}}(\rho_{XE})\subseteq \mathfrak{S}(\mathcal{H}_X  \otimes \mathcal{H}_E)$ be a ball of states with $d(\rho_{XE}, \tilde{\rho}_{XE})\leq \tilde{\epsilon}\ \forall \tilde{\rho}_{XE} \in \mathcal{B}^{\tilde{\epsilon}}(\rho_{XE})$. Then, smooth min-entropy of $\rho_{XE}$ given $E$ is
\[H_{min}^{\tilde{\epsilon}}(X|E)_{\rho}= \sup_{\sigma_E \in \mathfrak{S}(\mathcal{H}_E)} \Big( \sup_{\tilde{\rho}_{XE} \in \mathcal{B}^{\tilde{\epsilon}}(\rho_{XE})} \big(\min_{\lambda \in \mathbbm{R}} \{ -\log(\lambda) | \tilde{\rho}_{XE} \leq \lambda (\mathbbm{1}_{X}\otimes \sigma_E) \}\big)\Big) .\]
\end{definition}

Although a similar lemma for privacy amplification can be stated with other entropy measures, the version using smooth min-entropy is most widely used and most appropriate for potential applications. In the Leftover Hashing Lemma, security of a protocol $\mathcal{E}$ (expressed by its distance from a perfectly secure protocol $\mathcal{F}$) is related to the achievable key length $l$ and smooth min-entropy.%can also be used to state a version of the Privacy Amplification Theorem (5.6.1 in \cite{RennerPhD}), also called Leftover Hashing Lemma \cite{TomLev17}.

\begin{theorem}[Leftover Hashing Lemma]
\label{privacyamplification}
%Either H2,Hmin or smooth Hmin. Maybe both Hmin and smooth Hmin?\textcolor{green}{SOMETHING IS MISSING HERE}
For an input state $\rho_{XE}$ under privacy amplification with hashing output of length $l$, the following holds: %the difference between a given protocol $\mathcal{E}$ and an ideal protocol $\mathcal{F}$ with hashing output of length $l$ is defined by the following relation:
\[ \big\|(\mathcal{PA} \otimes \mathbbm{1})( \rho_{XE}) \big\|_{\tr}\leq\frac{1}{2} 2^{-\frac{1}{2}(H_{min}^{\tilde{\epsilon}} ({X}|E)_{\rho} - l)} +2\tilde{\epsilon}.\]
\end{theorem}
%\textcolor{green}{SOMETHING IS MISSING HERE}
The proof can be found in \cite{RennerPhD, TomLev17}. %\textcolor{green}{We could sketch, but that would require hashing definition; probably ppl would be better off just reading Tomamichel's paper.}

Importantly, this directly links the key length $l$ for which the protocol becomes secure to the smooth min-entropy via the following security criterion (first found in \cite{RennerPhD}):

\begin{equation}
l \leq H_{min}^{\tilde{\epsilon}}(X|E)_{\rho}
\end{equation}

For a protocol to be at least $\epsilon$-secure, the key length must be chosen to be at least $l=H_{min}^{\tilde{\epsilon}}(X|E)_{\rho} - 2\log \frac{1}{2(\epsilon-2\tilde{\epsilon})}$. Clearly, there is a tradeoff between key length and security: increasing the key length leads to a decrease in the error $\tilde{\epsilon}$, which leads an increasingly secure protocol. % The more secure the protocol should be, the shorter the key can be. The longer the required key, the less secure the related protocol.

\subsection{From Collective to General Attacks}
\label{cohtocoll}

As found in \cite{CKR}, the postselection technique can be applied to prove that security of a protocol against collective attacks implies security against general attacks, with an improved multiplicative factor security bound in comparison to previous studies. Similarly, the generalized postselection theorem (Theorem \ref{ImportantThm}) can be employed to derive such a relation for a protocol with a more general symmetry $\mathcal{S}$. In one sentence, the result can be summarized as follows:
\\
\\
%\begin{statement}[Result on the Security of General Attacks]
\textit{If a protocol $\mathcal{E}$ is $\epsilon$-secure against collective attacks, performing an additional privacy amplification $\mathcal{PA}$ whereby the key is shortened by $2\log(\dim(\mathcal{N}))=2\log(g_{n,d})$, then the protocol $\mathcal{E}'=\mathcal{PA} \circ \mathcal{E}$ is $\epsilon'$-secure against general attacks with $\epsilon'=g_{n,d}\epsilon$.}
%\end{statement}
\\
\\
\textcolor{red}{Update [13 March 2024]: This statement is not known to be true with $\epsilon'=g_{n,d}\epsilon$, but rather with $\epsilon'=4g_{n,d}\sqrt{2\epsilon}$. For details, we refer to \cite{NTZLT}.}
\\
\\
The remainder of this section will be spent justifying the above statement. No actual changes need to be made to the original argument in \cite{CKR} to accommodate general symmetry groups. However, the original argument is given mostly in terms of intuition rather than mathematical terms, whereas we will attempt to describe and justify it in more detail.
Thereafter, applying this result to a protocol with the symmetry of tensor powers of stabilizer states from \cite{Gross}, an even better overhead for general attacks can be found.
\\
\\
For collective attacks, the adversary Eve would act on each signal independently and identically - the input of the protocol would be a pure state $\sigma%_{\mathcal{H}\mathcal{H}}
 \in \mathfrak{S}(\mathcal{H}\mathcal{H})$ taken to the $n$-fold tensor power: $\sigma_{TE}=\sigma%_{\mathcal{H}\mathcal{H}}
^{\otimes n} \in \mathfrak{S}(\mathcal{H}_{T}\otimes\mathcal{H}_{E})$ %^{\otimes n} \mathcal{H}^{\otimes n})$
 with the Hilbert space $\mathcal{H}_{T}=\mathcal{H}^{\otimes n}$ associated to the trusted pair of Alice and Bob and the Hilbert space $\mathcal{H}_{E}=\mathcal{H}^{\otimes n}$ associated to Eve (see Figure \ref{Hilbertspaces}). A protocol $\mathcal{E}$ would therefore be called \textit{$\epsilon$-secure against collective attacks} if, for any such tensor product state $\sigma_{TE}=\sigma^{\otimes n}$,%\in\mathfrak{S}({\mathcal{H} \mathcal{H}}^{\otimes n})$,
\begin{equation}
 \big\| \mathcal{E}-\mathcal{F} \big\|_{\diamond, coll}= \big\| ((\mathcal{E}-\mathcal{F})\otimes \mathbbm{1}_{E}) \sigma^{\otimes n} \big\|_{\tr} \leq \epsilon.
\end{equation}

Because of the structure of $\tau_{TE}\in\mathfrak{S}(\mathcal{H}_T \otimes \mathcal{H}_E)$ in \eqref{tau-TE}, which is a mixture of such tensor products of pure states, this implies the following statement:
\begin{equation}
 \big\| ((\mathcal{E}-\mathcal{F})\otimes \mathbbm{1}_{E}) \tau_{TE}\big\|_{\tr} \leq \sup_{\sigma \in \mathfrak{S}(\mathcal{H}\otimes  \mathcal{H})}  \big\| ((\mathcal{E}-\mathcal{F})\otimes \mathbbm{1}_{E}) \sigma^{\otimes n} \big\|_{\tr} \leq \epsilon.
\end{equation}

From this, we now want to infer a statement about the security against general attacks, which is related to bounding another protocol $\mathcal{E}'$ with a bigger input state, giving Eve access to an additional (quantum) system $N$ associated to the Hilbert space $\mathcal{N}=\mathcal{N}=(\mathcal{H}^{\otimes n} \otimes \mathcal{H}^{\otimes n})^{s\otimes\overline{s}}$. By the previously stated definition of security in \eqref{security}, a protocol $\mathcal{E}'$ is said to be \textit{$\epsilon'$-secure against general attacks} if

 \begin{equation}
 \big\| \mathcal{E}'-\mathcal{F}' \big\|_{\diamond}= \sup_{\sigma_{TEN}\in\mathfrak{S}(\mathcal{H}_T\otimes \mathcal{H}_E\otimes  \mathcal{N})} \big\| ((\mathcal{E}'-\mathcal{F}')\otimes \mathbbm{1}_{EN}) \sigma_{TEN} \big\|_{\tr} \leq \epsilon' 
\end{equation}

Here, the postselection theorem (\ref{ImportantThm}) that was established in Section \ref{postselectiongeneralized} can be applied, which implies:

\begin{equation}
\big\| \mathcal{E}'-\mathcal{F}' \big\|_{\diamond} \leq g_{n,d} \big\| ((\mathcal{E}'-\mathcal{F}')\otimes \mathbbm{1}_{ EN}) \tau_{T EN}\big\|_{\tr} \leq \epsilon' 
\end{equation}

This relates the notion of collective $\epsilon$-security with the state $\tau_{TE}$, while the notion of general $\epsilon$-security is connected to the state $\tau_{TEN}$, which is a purification of $\tau_{TE}$. Now, we want to exploit this relation, in combination with the Leftover Hashing Lemma \cite{RennerPhD, TomLev17}, to relate the protocol $\mathcal{E}'$ to $\mathcal{E}$ via an additional privacy amplification $\mathcal{PA}$:

\begin{equation}
\mathcal{E}'=\mathcal{PA} \circ \mathcal{E}
\end{equation}

Let $\omega_{X EN}=\mathcal{E}(\tau_{TEN})$ be the state that the input $\tau_{T EN}$ was transformed into in the previous steps of the protocol. The system $X$ denotes the system that the trusted system $T$ of the input state $\tau_{TEN}$ was transformed to during the protocol so far. Of course, the same additional privacy amplification step is performed in $\mathcal{F}'$. It must be noted that the dimension of the spaces did not change; in particular, $X$ now contains Alice and Bob's perfectly correlated but not yet perfectly secret keys. According to the Privacy Amplification Theorem (\ref{privacyamplification}), the following holds for the additional privacy amplification step:
\begin{equation} \label{privamp}
 \big\| ((\mathcal{E'}-\mathcal{F'})\otimes \mathbbm{1}_{EN}) \tau_{T EN} \big\|_{\tr}=\big\| (\mathcal{PA}\otimes \mathbbm{1}_{EN}) \omega_{X EN} \big\|_{\tr} \leq \frac{1}{2} 2^{-\frac{1}{2}(H_{min}^{\tilde{\epsilon } }(X| EN)_{\omega}-l')} +2\tilde{\epsilon }
\end{equation}
where $l'$ is the length of the key and $H_{min}^{\tilde{\epsilon } }(X| EN)_{\omega}$ is the smooth min-entropy of the state $\omega_{XEN}$ with side information $EN$.

To establish a relation between the key length $l$ of $\mathcal{E}$ and $l'$ of $\mathcal{E}'$, we therefore need to find the relation between their entropies. Note that these entropies are related to a state on different spaces, because protocol $\mathcal{E}'$ acts on a purification of the space of $\mathcal{E}$, with an additional system ${N}$. In fact, the act of shortening the key can be interpreted as a compensation of the additional information available to the adversary in the general scenario in form of the space $\mathcal{N}$.

The process of deriving how much the key need to be shortened has three steps: firstly, a relation between the quantum state $\omega_{XEN}$ and its partial trace $\omega_{XE}=\operatorname{tr}_N\omega_{XEN}$ is established (Lemma \ref{twirling}). Secondly, it is investigated how min-entropy changes under this purification (Lemma \ref{damnthiswassohard}), before `'smoothing" the relation (using \eqref{smoothing}) to insert smooth min-entropy in the Leftover Hashing Lemma in \eqref{finally-leftoverhashing}.
\\
\\
One key technique that this process relies on is \textit{twirling}, which is the operation of taking the average of a channel under unitary operations \cite{Wilde12,Magesan08}.
For a given set $\{U_k\}_{k=1,...,K}$ of unitary operators on the Hilbert space $\mathcal{H}$, applying the associated twirling channel $\mathcal{T}$ to a quantum state $\rho\in\mathfrak{S}(\mathcal{H})$ yields
\begin{equation*}
\mathcal{T}(\rho) = \frac{1}{\dim(\mathcal{H})^2} \sum_{k=1}^{K} U_k \rho U_k^{\dagger}.
\end{equation*}
%\textcolor{green}{refernces???} \cite{Magesan08}

This kind of map can be defined for different sets of unitary operators, e.g. Clifford group or Pauli group. This technique is useful in the context of various quantum information processing tasks, such as entanglement purification \cite{BDSW96}, randomized benchmarking \cite{KLRBBLOSW08}, and the simulation of noise in quantum error correction codes \cite{CB19}.
%, which give different meanings to this %Different sets of unitary operators give different meaning to this operation, leading to the notion of \textit{unitary n-designs}$ \cite{}.

%f the unitaries are chosen to be Heisenberg-Weyl operators, 
If the set of unitaries is chosen to be the set of Heisenberg-Weyl operators $\{W_k\}_{k=1,...,\dim(\mathcal{H})^2} =\{X_i Z_j\}_{i,j=0,...,\dim(\mathcal{H})-1}$ (with $X$ and $Z$ being Pauli operators, and the index indicating which space they are applied to), it can be shown that applying them with uniform probability to any qudit density operator yields the maximally mixed state \cite{Wilde12}:
\begin{equation}
\label{twirlingeq}
\mathcal{T}_{HW}(\rho)=\frac{1}{\dim(\mathcal{H})^2} \sum_{k=1}^{\dim(\mathcal{H})^2} W_k \rho W_k^{\dagger}=  \frac{\mathbbm{1}}{\dim(\mathcal{H})}
\end{equation}
For the qubit case, i.e. for Pauli operators, the above relation can easily be checked by writing the density operator as a Bloch vector and using commutation relations \cite{Wilde12}.%, which can be extended to $2n$ qubit systems. \cite{Wilde12}

Now, we introduce the following preliminary lemma about the relation of $\omega_{XEN}$ and its partial trace $\omega_{XE}=\operatorname{tr}_N\omega_{XEN}$, which uses twirling in its proof.
\begin{lemma} \label{twirling}
For a quantum state $\omega_{XEN}\in \mathfrak{S}(\mathcal{H}_X \otimes \mathcal{H}_E \otimes \mathcal{N})$, the following bound holds:
\[\omega_{X EN} \leq \dim(\mathcal{N}) (\omega_{XE} \otimes \mathbbm{1}_{N}) \]
where $\omega_{XE}=\operatorname{tr}_N \omega_{XEN}$.
\end{lemma}

\begin{proof}[Proof of Lemma \ref{twirling}]
Inserting the state of interest $\omega_{X EN}$ and applying the twirling channel $T_{HW}$ associated to such Heisenberg-Weyl operators \eqref{twirlingeq} partially, i.e. only to the subsystem $N$, yields

\begin{equation} \begin{split}
\label{twirling-eq}
\mathcal{T}_{HW}^N (\omega_{X EN})& = \frac{1}{\dim(\mathcal{N})^2} \sum_k (\frac{\mathbbm{1}}{\dim(\mathcal{H}_{XE})} \otimes W_k)\omega_{X EN} (\frac{\mathbbm{1}_{N}}{\dim(\mathcal{H}_{XE})} \otimes W_k)^{\dagger}\\& = \omega_{XE} \otimes \frac{\mathbbm{1}}{\dim(\mathcal{N})} .
\end{split}
\end{equation}

Some of the Heisenberg Weyl operators are tensor powers of the identity matrix, and the sum can be separated:

\begin{equation}
\label{identitypart-eq}
\mathcal{T}_{HW}^N (\omega_{X EN})=\frac{1}{\dim(\mathcal{N})^2}( \omega_{X EN}+\delta_{X EN})
\end{equation}
with some positive $\delta_{X EN} \geq 0$.

Combining \eqref{identitypart-eq} and \eqref{twirling-eq} yields
\begin{equation*}\begin{split}
\omega_{X EN} &=\dim(\mathcal{N})^2 \mathcal{C}_{HW}^N (\omega_{XEN}) -\delta_{XEN} \\ &= \dim(\mathcal{N})^2 \omega_{XE} \otimes \frac{\mathbbm{1}_{N}}{\dim(\mathcal{N})} - \delta_{X EN}\\& \leq \dim(\mathcal{N}) (\omega_{XE} \otimes \mathbbm{1}_{N}).
\end{split}
\end{equation*}
\end{proof}

\begin{remark}
Note that this bound is tight for the maximally entangled state. %\textcolor{green}{do I want to keep that in here?}
\end{remark}

Knowing how these pre-privacy amplification states appearing in the protocol for a collective and a general attack are related to one another directly leads to a relation between the states' min-entropies: 

\begin{lemma}[Min-entropy and purification] \label{damnthiswassohard}
The following relation holds for the min-entropy $H_{min} (X| E)_{\operatorname{tr}_{N} \omega}$ of a state $\omega_{XE}=\operatorname{tr}_N \omega_{XEN}$  and the min-entropy $H_{min} (X| EN)_{\omega}$ of the state's Stinespring dilation $\omega_{XEN}\in \mathfrak{S}(\mathcal{H}_X \otimes \mathcal{H}_E \otimes \mathcal{N})$:
\[ H_{min} (X| EN)_{\omega} \geq H_{min} (X| E)_{\operatorname{tr}_{N} \omega} - 2\log \dim (\mathcal{N}).\]
\end{lemma}

%\textcolor{red}{Here lies the problem}

\begin{proof}[Proof of Lemma \ref{damnthiswassohard}]
Recall the definition of the min-entropy $H_{min} (X| EN)_{\omega}$ in terms of a semi-definite optimization problem (\ref{SDP}). Suppose we have found a feasible solution with a certain $\Lambda_{XEN}$ - then the min-entropy is defined by
\begin{equation} \label{sohard-eq1}
2^{H_{min} (X|EN)_{\omega}} =\operatorname{tr} \Lambda_{XEN}\omega_{XEN}.
\end{equation}

Then, Lemma \ref{twirling} ensures the following bound:

\begin{equation}  \label{sohard-eq2}
\operatorname{tr} \Lambda_{XEN}\omega_{XEN} \leq \dim{(\mathcal{N})} \operatorname{tr} \Lambda_{XEN}\omega_{XE} \otimes \mathbbm{1} =(\dim \mathcal{N})^2 \operatorname{tr} \Gamma_{XE}\omega_{XE}
\end{equation}

%Since the trace of $N$ is only applied to $\Lambda_{XEN}$, it is replaced by $\Gamma_{XE} = \frac{1}{\dim (\mathcal{N})} \operatorname{tr}_{N} \Lambda_{XEN}$ in the last step above.
Here, a renaming of $\Gamma_{XE} = \frac{1}{\dim (\mathcal{N})} \operatorname{tr}_{N} \Lambda_{XEN}$ has taken place in the last step.
Now it remains to be shown that $\Gamma_{XE}$ is in itself a solution (if not the best solution) to the semi-definite optimization problem defining $H_{min} (X|E)_{\operatorname{tr}_{N} \omega}$.

To show this, we check the feasibility criteria. The second criterion, $\Gamma_{XE}\geq 0$, can be directly inferred from the feasibility of $\Lambda_{XEN}$, which entails $\Lambda_{XEN}\geq 0$.
Similarly, the feasibility of $\Lambda_{XEN}$ implies $\operatorname{tr}_{X} \Lambda_{XEN}=\Lambda_{EN}\leq\mathbbm{1}_{EN}$, which can be used to show that $\Gamma_{XE}$ fulfills the first criterion:
\[ \Gamma_{E}=\operatorname{tr}_{X} \Gamma_{XE}=\dim(\mathcal{N}) \operatorname{tr}_{X} \operatorname{tr}_{N} \Lambda_{XEN}=\operatorname{tr}_{N} \Lambda_{EN}\leq \operatorname{tr}_{N} \mathbbm{1}_{EN} \leq \mathbbm{1}_{E} \]

In conclusion, $\Gamma_{XE}$ fullfills the criteria, and is thus a feasible solution for the SDP defining $H_{min} (X|E)_{\operatorname{tr}_{N} \omega}$. Since the SDP still entails a maximization, this implies 
\begin{equation}\label{sohard-eq3}
\operatorname{tr} \Gamma_{XE} \omega_{XE} \leq 2^{-H_{min}(X|E)_{\operatorname{tr}_{N} \omega}}
\end{equation}
where equality would hold if $\Gamma_{XE}$ was optimal.
Inserting \eqref{sohard-eq1} and \eqref{sohard-eq3} into the bound \eqref{sohard-eq2}, we obtain:
\[2^{-H_{min} (X|EN)_{\omega}} \leq \dim(\mathcal{N})^2 2^{-H_{min}(X|E)_{\operatorname{tr}_{N} \omega}} \]
which directly implies
\[H_{min} (X|EN)_{\omega} \geq H_{min} (X|E)_{\operatorname{tr}_{N} \omega} - 2\log \dim (\mathcal{N}).\]
\end{proof}

%SATZZEICHEN
%\subsubsection{From min-entropy to smooth min-entropy}

\textcolor{red}{Update [13 March 2024]: The corresponding statement is not known to hold for smooth entropy wrt to trace distance with the same smoothing parameter $\tilde{\epsilon}$ on both smoothed entropies.}

With this established, the next step is to extend this statement to smooth min-entropy. Smooth min-entropy is related to min-entropy via the following smoothing relation:
\begin{equation} \label{smoothing}
\sup_{\tilde{\omega}_{XEN} \in B^{\tilde{\epsilon}}(\omega_{XEN})}  H_{min} (X|EN)_{\tilde{\omega}}= H_{min}^{\tilde{\epsilon}} (X|EN)_{{\omega}}.
\end{equation}
Since the smoothing parameters $\epsilon$ and $\tilde{\epsilon}$ are defined via the trace distance, which is trace preserving, $\tilde{\omega}_{XEN}\in B^{\tilde{\epsilon}}(\omega_{XEN})$ directly implies $\operatorname{tr}_{N} \tilde{\omega}_{XEN}\in B^{\tilde{\epsilon}}(\operatorname{tr}_{N}\omega_{XEN})$, which we can use to write:
\[\sup_{\tilde{\operatorname{tr}_{N} \omega}_{XEN} \in B^{\tilde{\epsilon}}(\operatorname{tr}_{N} (\omega_{XEN}))}  H_{min} (X|E)_{\operatorname{tr}_{N} \tilde{\omega}}= H_{min}^{\tilde{\epsilon}} (X|E)_{{\operatorname{tr}_{N} \omega}}\]

\textcolor{red}{Update [13 March 2024]: For a given state $\operatorname{tr}_{N} \tilde{\omega}_{XEN}\in B^{\tilde{\epsilon}}(\operatorname{tr}_{N}\omega_{XEN})$, does there always exist a state $\tilde{\omega}_{XEN}\in B^{\tilde{\epsilon}}(\omega_{XEN})$? Using $\tilde{\epsilon}$-balls with respect to purified distance like in \cite{TomLev17}, this is indeed true. 
For the trace distance, the argument can be fixed at the cost of a worse smoothing parameter (going via purified distance, in fact). 
%, giving $\tilde{\omega}_{XEN}\in B^{\sqrt{\tilde{\epsilon}}}(\omega_{XEN})$  (a bigger ball).
For details, we refer to \cite{NTZLT}.
% implies $\operatorname{tr}_{N} \tilde{\omega}_{XEN}\in B^{\tilde{\epsilon}}(\operatorname{tr}_{N}\omega_{XEN})$. This claim can only be made when using purified distance.
}

Assuming that we have found a state $\tilde{\omega}_{XEN}$ such that $H_{min} (X|E)_{\operatorname{tr}_{N} \tilde{\omega}}= H_{min}^{\tilde{\epsilon}} (X|E)_{{\operatorname{tr}_{N} \omega}}$, this immediately implies

\begin{equation}\begin{split} \label{smoothminentropy-purification}
 H_{min}^{\tilde{\epsilon}} (X|E)_{{\operatorname{tr}_{N} \omega}} -2\log(\dim(\mathcal{N})) &= H_{min} (X|E)_{\operatorname{tr}_{N} \tilde{\omega}} -2\log(\dim(\mathcal{N})) \\ &
\leq H_{min} (X|EN)_{\tilde{\omega}} \leq H_{min}^{\tilde{\epsilon}} (X|EN)_{\omega}
\end{split} \end{equation}

where the last inequality follows from the definition of smooth min-entropy as a supremum. This directly implies that a smoothed version of Lemma \ref{damnthiswassohard} is true, and thus tells us how smooth min-entropy transforms under purification.
\\
\\
%\subsubsection{Relation between the key lengths and security parameters}
Now, the relation between the two smooth min-entropies can be used to connect the security against general attacks of the protocol $\mathcal{E}'$ with the security parameter $\epsilon$ of the protocol $\mathcal{E}$ against collective attacks because of privacy amplification, by inserting the relation between the smooth min-entropies \eqref{smoothminentropy-purification} into the Leftover Hashing Lemma associated to the additional privacy amplification step \eqref{privamp}:

\begin{equation}\begin{split} \label{finally-leftoverhashing}
 \big\| \big((\mathcal{E}-\mathcal{F})\otimes \mathbbm{1}_{EN}\big) \tau_{TEN} \big\|_{\tr}& \leq \frac{1}{2} 2^{-\frac{1}{2}(H_{min}^{\tilde{\epsilon } }(X|EN)_{\omega}-l')} +2\tilde{\epsilon }\\& \leq \frac{1}{2} 2^{-\frac{1}{2}(H_{min}^{\tilde{\epsilon } }(X|E)_{\operatorname{tr}_{N} \omega}- 2\log (\dim(\mathcal{N})) -l')} +2\tilde{\epsilon } \\&
\leq \frac{1}{2} 2^{-\frac{1}{2}(l-2\log \frac{1}{2(\epsilon-2\tilde{\epsilon})} - 2\log \dim(\mathcal{N})-l')} +2\tilde{\epsilon } := \epsilon
\end{split} \end{equation}

Thus it is found that Eve gaining more knowledge from a bigger system can be counter-acted by reducing the length of the key by exactly the system's dimension. For the system $\mathcal{N}=(\mathcal{H}_T\otimes \mathcal{H}_E)^{s\otimes \overline{s}}=(\mathcal{H}^{\otimes n} \otimes \mathcal{H}^{\otimes n})^{s\otimes\overline{s}}$, 
this dimension is $g_{n,d}$, as appears in Section \ref{postselectiongeneralized}. This implies the following relation between the key lengths:

\begin{equation} \label{thekeylengths}
l'\geq l-2 \log \dim (\mathcal{N})=l-2 \log (g_{n,d})
\end{equation}

Now, equation \eqref{thekeylengths} can be related back to a statement about the general security of $\mathcal{E}'$ in the following way: If $\mathcal{E}'$ is obtained from $\mathcal{E}$ by shortening the output of the hashing function by $2\log(g_{n,d})$, i.e. shortening the key by this amount, then
\begin{equation}
\big\| \mathcal{E}'-\mathcal{F}' \big\|_{\diamond} \leq g_{n,d} \big\| ((\mathcal{E}'-F')\otimes \mathbbm{1}_{EN}) \tau_{TEN}\big\|_{\tr} \leq {g_{n,d}} \epsilon \coloneqq \epsilon' .
\end{equation}

Therefore, if a protocol $\mathcal{E}$ is $\epsilon$-secure against collective attacks, and we obtain $\mathcal{E}'$ from $\mathcal{E}$ by shortening the hashing function's output by $2\log (g_{n,d})$, then we can infer that $\mathcal{E}'$ is $\epsilon'$-secure against general attacks with $\epsilon'=g_{n,d} \epsilon$. %This is directly analogous to the relation between security parameters found in \cite{CKR}.

%However, since a bigger symmetry group is considered, the dimension of its associated invariant subspace is smaller. As found by orbit counting \ref{orbitcountingresults}, the dimension is upperbounded by $D_N$ for $N$ party protocols.

In summary (similarly to the mathematical bound on the diamond norm), the comparison between security against general and collective attacks of a protocol with invariance under an arbitrary symmetry $\mathcal{S}$ is directly related to the dimension $g_{n,d}=\dim (\mathcal{H}^{\otimes n} \otimes \mathcal{H}^{\otimes n})^{(s\otimes\overline{s})}=\dim(\mathcal{N})$ of an invariant subspace. For this reason, the next section is occupied with computing this number for the symmetry group of interest in this project, the stochastic orthogonal group.

%STARTT

\section{Orbit Counting for Discrete Orthogonal Matrices}
\label{orbitcountingresults}

As established in Section \ref{postselectiongeneralized}, the bound on the diamond norm is related to the coefficients $g_{n,d}=\dim (\mathcal{H}^{\otimes n} \otimes \mathcal{H}^{\otimes n})^{(s\otimes\overline{s})}=\dim(\mathcal{N})$, which are equal to the dimension of an invariant subspace associated to the general symmetry group $\mathcal{S}$. Subsequently, as described in Section \ref{horriblecalculation}, these coefficients directly determine the difference in security parameter between collective and general attacks. In this section, we will introduce Witt's Lemma (described in Section \ref{orbitcounting}) and use it to compute this dimension by counting orbits for the symmetry group of discrete orthogonal matrices, as introduced in \cite{Gross} and Definition \ref{newsym-ortho}. Firstly, we will describe how the dimension was computed in the context of two party QKD protocols in Section \ref{orbitcountingresults1}, before showing how this can be generalized to the $N$ party case in Section \ref{orbitcountingresultsN}. Note that this does not directly have implications relevant to QKD - instead, the dimension for discrete orthogonal group is a stepping stone for getting the dimension for stochastic orthogonal group introduced in Definition \ref{newsym-stochastic}, which is the group that is relevant to the postselection theorem and subsequent QKD results, as will be explained in more detail in Section \ref{StochasticOrbitCounting}.

%\subsection{Orbit counting results}
\subsection{Orbit Counting Using Witt's Lemma}
\label{orbitcounting}

In a system with a given symmetry, some elements of the phase space may become equivalent under transformations with that symmetry (called: being in the same orbit). Many problems can then be simplified because the behaviour of elements in the same orbit can be inferred from one another. As established in previous sections, we are highly interested in the dimension of a certain invariant subspace, $\mathcal{N}$. Computing the dimension of the space of elements that are invariant under a given symmetry group is achieved by counting the dimension of the space that the symmetry group projects to. In case of the stochastic orthogonal group and discrete orthgonal group (and also permutation group), the computational basis is preserved under their action. In such cases, the invariant subspace is spanned by mixtures of states within the same orbit, and each distinct orbit constitutes one basis vector of the invariant subspace. Therefore, the dimension of the invariant subspace is given by the number of distinct orbits of the group, and computing the dimension is achieved by counting the orbits of the symmetry group. %Here, we mainly deal with computing the dimension of the space of elements that are invariant under a given symmetry, which means computing the dimension of the space that the symmetry group projects to. This could, for example, involve applying the symmetric transformations to each basis element and then checking how many of the resulting vectors are linearly independent. Equivalently, however, one can also count the orbits of the symmetry group.

\begin{definition}[Orbit of a group]
Let $G$ be a group acting on a set $X$. For each element $x\in X$ of the set, let %$\operatorname{orb}
$\operatorname{orb} G(x) =\{g x|g \in G\}$. The set %\operatorname
$\operatorname{orb}  G(x)$ is a subset of $X$ that is called the orbit of $x$ under $G$.
\end{definition}

The orbit of an element $x\in X$ is given by all the elements $y\in X$ that are connected to $x$ via the group elements $g\in G$. Therefore, objects in a given orbit can be considered isomorphic in the sense that they will map to the same object under application of some $g\in G$. The dimension of the space that the $g\in G$ map to is thus equal to the number of distinct orbits of $G$.

One important tool for orbit counting is Witt's Lemma \cite{Witt} (sometimes called Witt's Theorem, stated here \cite{BoltRoomWall,MontealegreThesis} for symplectic groups). This lemma relates two sets of vectors, $\{v_i\}$ and $\{w_i\}$, with the same linear product relations to a mapping $M$ (with one particular property) between these two vectors.

\begin{theorem}[Witt's Lemma] \label{WittsLemma}
Let $V$ be a vector space with a non-degenerate bilinear product $\beta(\cdot,\cdot)$. Let $\{v_i\}$ and $\{w_i\}$ be two sets of $k$ linearly independent elements (vectors) of $V$ satisfying
\begin{equation*}
\beta(v_i,v_j)=\beta(w_i,w_j)\ \forall i,j=1,...,k.
\end{equation*}
Then, there exists a map $M: V\mapsto V$ satisfying
\begin{equation*}
\beta(Mv,Mw)=\beta(v,w)\ \forall v,w\in V
\end{equation*}
for which
\begin{equation*}
Mv_i =w_i\ \forall i=1,...,k.
\end{equation*}
\end{theorem}

The converse of this statement is also true and comparatively easy to see: If such a mapping exists (that conserves linear products and maps $v_i$ to $w_i$ $\forall i$), then the two vector sets $\{v_i\}$ and $\{w_i\}$ obey the same linear product relations.
%\textcolor{green}{SKETCH PROOF?}

This theorem directly relates to orbits, since by definition, phase space elements that can be mapped to one another are in the same orbit.
Therefore, to compute the number of distinguishable orbits in a discrete space (which we aim to do in Sections \ref{orbitcountingresults1} and \ref{orbitcountingresultsN}), one needs to count how many distinct values there are for the linear product $\beta(\cdot,\cdot)$ associated to the space. % If the dimension of $V$ is small, the %For small $n$, the linear dependency between the vectors can determine their respective linear products

It can be noted that orbit counting can easily be employed to reproduce the known dimension for the permutation group found in \cite{CKR}. To compute the dimension of the space that is invariant unter permutations, i.e. the orbits of the permutation group, the number of distinct basis elements has to be counted. Firstly, consider a single permutation matrix applied to a Hilbert space vector in the computational basis $\ket{ x_1, x_2, ...x_n}$, where each $x_i$ is some discrete number between $0$ and $d-1$. Applying a permutation will switch the numbers $x_i$ with each other while conserving the number of times each number between $0$ and $d-1$ appears. For example, applying a permutation to the vector $\ket{1,0,...,0}$ will change the position of the $1$, but the number ``0" will always appear $n-1$ times, and the number ``1" will always appear one single time. In total, the basis elements that are distinct with respect to permutations are thus characterized by occupational numbers $n_j$ for $j=1,...,d$ with $\sum_j n_j=n$. This is a simple and well known combinatorics problem (``Stars and Bars" \cite{starsandbars}, boson statistics): How many distinct ways are there to put $n$ stars into $d$ boxes? The solution is exactly $\binom{n+d-1}{n}= (n+1)^{d-1}$, which would be the dimension of the Hilbert space $(\mathcal{H}^{\otimes n})^{\pi}$.
However, the postselection theorem is related to the space $(\mathcal{H}^{\otimes n} \otimes \mathcal{H}^{\otimes n} )^{\pi\otimes\pi}$. In this case, there is the same permutation acting on a vector $\ket{ x_1, x_2, ...x_n}$ in the first Hilbert space and a vector $\ket{ y_1, y_2, ...y_n}$ in the second Hilbert space, which can be understood as a permutation of the rows of the following matrix: \[
\left( \begin{array}{r r}
x_1 & y_1  \\                                              
 \vdots & \vdots \\
x_n & y_n  \\                                              
\end{array}\right).\]
Since the same permutation acts on each column, elements $x_i$ and $y_i$ always stay together and can thus be considered as a pair. Applying the same logic as before, the occupation numbers of such pairs $(x_i,y_i)$ (of which there exist $d^2$ possibilities) now define discrete orbits, so there exist $\binom{n+d^2-1}{n}= (n+1)^{d^2-1}$ orbits, which is exactly the number in equation (\ref{permutationsgtd}), found in \cite{CKR}.

Using the same argument, this problem can also easily be extended to the $N$ party case: then, the number of orbits is $\binom{n+d^{2N}-1}{n}= (n+1)^{d^{2N}-1}$, as it appears in \cite{Grasseli18}.

For the symmetry group of discrete orthogonal matrices, Witt's Lemma will be employed to recast the task of orbit counting as a combinatorics problem.

%It is also important to note that the resolution of identity is known to hold for this group \cite{Gross}. Since we eventually want to use this counting in the context of the postselection theorem, it is crucial to establish that before delving any deeper. \textcolor{green}{IS THIS GOOD PLACEMENT?}

\subsection{Orbits of Discrete Orthogonal Group for Two Parties}
\label{orbitcountingresults1}

%\textcolor{green}{this needs justification of stochastic vs orthogonal group - maybe it can also go to the later section.... but it should be somewhere!!}

To apply the postselection theorem using the new symmetry group discovered in \cite{Gross}, the dimension of the subspace that is invariant under its action has to be computed. This dimension for the stochastic orthogonal group $\mathcal{O}_n$ can be computed via a relaxation to the discrete orthogonal group $\tilde{\mathcal{O}}_n$, which is described by discrete orthogonal matrices $\tilde{O}\in\tilde{\mathcal{O}}_n$, which are introduced in Definition \ref{newsym-ortho} (see Section \ref{newsym}).
%These operators form a projection to the symmetric subspace. 
Therefore, our aim is to compute the dimension of the space $(\mathcal{H}^{\otimes n} \otimes \mathcal{H}^{\otimes n})^{(\tilde{O}\otimes \tilde{O})}$ (since $\tilde{O}$ is real, note that $\overline{\tilde{O}}=\tilde{O}$) which is invariant under the action $\tilde{O}\otimes \tilde{O}$. As established in Section \ref{orbitcounting}, the dimension of that space is equal to the number of distinct orbits of that group.

As a precursor, we will consider the group of discrete orthogonal matrices acting on $\mathcal{H}^{\otimes n}$, thus calculating the dimension ${D_1}$ of the invariant subspace $(\mathcal{H}^{\otimes n})^{\tilde{O}} \subseteq \mathcal{H}^{\otimes n}$. Note that this space is only a potentially helpful construct with no direct relation to the relevant subspace. However, on the one side, this will be helpful to describe the counting strategy - and on the other side, this will become important when treating the $N$ party generalization, where a recursion relation will be proposed.

The representation of the discrete orthogonal transformations act on computational basis elements of the Hilbert space in the following way:
\[R(\tilde{O}) \ket{x} =\ket{\tilde{O}x}\]
where $\ket{x}$ is a vector on the Hilbert space $\mathcal{H}^{\otimes n}$, and the symbols  $x$ live on the discrete Hilbert space $\mathbb{F}_d^n$, on which the $n\times n$ matrix $\tilde{O}$ acts.
Thereby, $R(\tilde{O})$ acts by preserving a finite set of basis elements, which justifies that the number of orbits is equal to the dimension of the invariant subspace.

%Then,
%\[\operatorname{orb}\tilde{\mathcal{O}}(\ket{x})=\{\tilde{O}\ket{x}|\tilde{O}\in\tilde{\mathcal{O}}\}=\{\ket{R(\tilde{O})x}|\tilde{O}\in\tilde{\mathcal{O}}\}\]
%Note that we will change notation to $\ket(x)\equiv \textbf{x}$.

%Let $x,x'\in \mathcal{H}^{\otimes n}$.
According to Witt's Lemma (\ref{WittsLemma}) \cite{Witt}, orbits of the symmetry group $\tilde{\mathcal{O}}_n$ are defined in the following way: 
\begin{equation}
\ket{x'} \in \operatorname{orb}\tilde{\mathcal{O}}(\ket{x}) \Leftrightarrow \beta({x},{x})=\beta({x}',{x}') \mod d 
\end{equation}
with a linear product $\beta(\cdot,\cdot)$ associated to the discrete vector space $\mathbb{F}_d^{\otimes n}$, for $\ket{x},\ket{x'}$ be vectors on $ \mathcal{H}^{\otimes n}$ and $x,x'$ on $\mathbb{F}_d^n$.

The number of distinct orbits is thus given by the number of distinct numbers $\beta(x,x) \mod d$ that define each orbit. Thereby, the task of counting orbits becomes a combinatorics problem where two distinct cases have to be considered. On the one hand, let $x\neq 0$. Because of the modulo on the condition, there are $d$ distinct possibilities to choose from: $\beta(x,x)=0,1,...,d-1$, and alll of these possibilities are realized if $n$ is large enough. On the other hand, let ${x}=0$; then, it is immediate that $\beta(x,x)=0$ by linearity. This case cannot be transformed into the case where $x\neq 0$, $\beta(x,x)=0$ and must thus be considered a separate orbit.

In summary, there are $d$ ways to choose distinct $\beta(x,x)$ for $x\neq 0$, and only one way to choose $\beta(x,x)$ for $x=0$. In conclusion, we obtain the following number of orbits, and thus dimension of $(\mathcal{H}^{\otimes n})^{\tilde{O}}$:

\begin{equation}
{D_1}=d+1
\label{onepartyorbit}
\end{equation}

Now, this can be expanded to the case of interest, where the dimension of interest is $D_2$, the dimension of the invariant subspace $(\mathcal{H}^{\otimes n} \otimes \mathcal{H}^{\otimes n})^{(\tilde{O}\otimes \tilde{O})} \subseteq \mathcal{H}^{\otimes n}\otimes \mathcal{H}^{\otimes n}$.
In this case, the space that $R(\tilde{O})\otimes R(\tilde{O})$ maps to has to be considered:

\[R(\tilde{O})\otimes R(\tilde{O}) \ket{x_1}\otimes\ket{x_2}=\ket{\tilde{O} x_1}\otimes\ket{\tilde{O}x_2}\]

Using Witt's Lemma (\ref{WittsLemma}), orbits of the group acting by $(\tilde{O}\otimes\tilde{O})\ket{x_1,x_2}$ can then be defined by the following set of equations for $x_1$ and $x_2$ being linearly independent:

\[  \ket{x_1',x_2'} \in \operatorname{orb} (\tilde{\mathcal{O}}\otimes\tilde{\mathcal{O}})  (\ket{x_1,x_2}) \Leftrightarrow \begin{cases} \beta(x_1,x_1)=\beta(x_1',x_1') \mod d \\\beta(x_2,x_2)=\beta(x_2',x_2') \mod d\\\beta(x_1,x_2)=\beta(x_1',x_2') \mod d \end{cases} \]
where $\ket{x_i}$ are a vector on the Hilbert space $\mathcal{H}^{\otimes n}$, and the symbol $x_i$ are a discrete vector on $\mathbb{F}_d^n$ $\forall i$.

Now, combinatorics can be employed to upper bound how many different discrete values can be chosen for $\beta(x_1,x_1), \ \beta(x_2,x_2)$ and $\beta(x_1,x_2)$.

\begin{figure}[H]
   \centering
       \includegraphics[width=14cm]{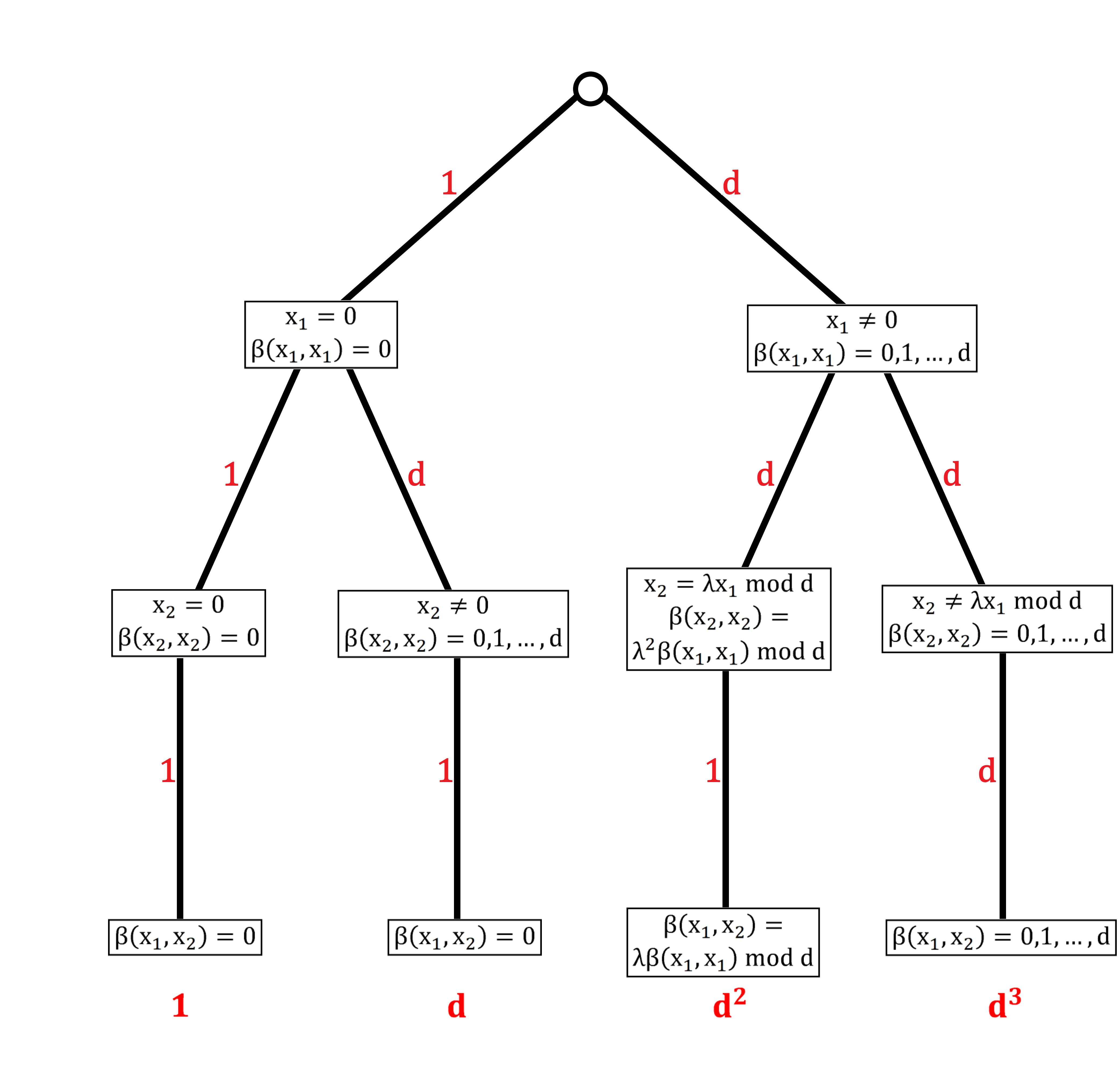}
\caption{Tree diagram to illustrate the different cases and how the choices influence each other. In red, the number of paths between knots have been indicated along each branch of the tree diagram. (1 path connected after d paths is short-hand notation for d paths each having a knot, and 1 path pertruding from this knot.) At the bottom, the total number of possible choices for each of the cases is indicated in bold and red. These numbers have to be added up to sum to the total number of possibilites, which is equal to the dimension $D_2$ we are looking for.}
\label{baumdiagramm}
\end{figure}

Added difficulty stems from the fact that $\beta(x_1,x_2)$ can depend on the previous choices, and that the two vectors $x_1$ and $x_2$ could be linearly dependent on one another. The latter could in principle prohibit the application of Witt's Lemma. However, we are actually not using Witt's Lemma directly, but employing it to derive an upper bound which can be improved upon by considering the different cases, and treating the possibility of linearly dependent vectors separately. In fact, the vectors \textit{must} be linearly dependent when the number of subsystems $n$ (corresponding to the entries of the computational basis vector which are transformed via discrete orthgonal matrices in $\tilde{O}x_i$) is small compared to the number of vectors (here: 2, namely $x_1$ and $x_2$). Therefore, for $n<2$, there will always be linear dependency between the $x_i$ for any computational basis vector. For $n\geq 2$, i.e. $n$ large enough, there may be linearly dependent and linearly independent vectors, and both cases must be taken into account. In an $N$ party scenario, there are $N$ different vectors $x_i$ with $i=1,...,N$ - then, $n<N$ will lead to definitive linear dependency, and both cases are possible for $n\geq N$. Since $n$ is usually assumed to be large (since larger $n$ means larger key length, or smaller error in the de Finetti bound), it can usually be assumed that both cases occur. But even for small $n$, linear dependency would fix some $\beta(x_i,x_j)$, and the assumption that these linear products can be chosen freely can therefore lead to an overcounting, making our result an upper bound, which is sufficient for our purposes.

In total, linear dependency between $x$-vectors and the differentiation between $x_i=0$ and $x_i\neq 0$ for $i=1,2$ all have to be taken into account. 
We will now list all the different cases and how they influence the available choices for the linear products $\beta(x_1,x_1), \ \beta(x_2,x_2)$ and $\beta(x_1,x_2)$. All of the cases and the number of choices for each are sketched in Figure \ref{baumdiagramm} in a tree diagram.

The simplest case that must be differentiated is the case where both vectors are $x_1=x_2=0$. This fully determines $\beta(x_1,x_1)=\beta(x_2,x_2)=\beta(x_1,x_2)=0$ by linearity, and thus only has one path associated to it.

If the first vector $x_1=0$, and $x_2\neq 0$, $\beta(x_2,x_2)$ can still be chosen out of the $d$ members of the set $\{0,1,...d-1\}$ - which constitutes to $d$ choices, while $\beta(x_1,x_1)=\beta(x_1,x_2)=0$ by linearity (1 choice each).

If $x_1\neq 0$, there are $d$ possibilities to fix $\beta(x_1,x_1)=0,1,...,d-1$. Then, there are two possibilites for the vector $x_2$: it can either be linearly independent of $x_1$, or not. If it is \textit{linearly independent}, i.e. $x_1\neq \lambda x_2$ for any $\lambda\in\{0,1,...,d-1\}$, $\beta(x_2,x_2)$ can be chosen from the set $\{0,1,...d-1\}$. Furthermore, there are $d$ possible choices for $\beta(x_1,x_2)$. Overall, this contributes $d^3$ possible paths to the number of orbits. For small $n$, these choices might be restricted, but as this would only lead to a smalller number of orbits, it can be assumed that it can be chosen freely, leading to an upper bound.
However, if $x_2$ is {linearly dependent} on $x_1$, i.e. $x_2=\lambda x_1$, then $\beta(x_2,x_2)=\lambda^2 \beta(x_1,x_1) \mod d$ and $\beta(x_1,x_2)=\lambda \beta(x_1,x_1) \mod d$ are fixed by the choice of $\beta(x_1,x_1)$, contributing one path for each possibility of $\lambda\in\{0,1,...,d-1\}$. There are thus $d$ distinct possibilities of choosing $\beta(x_1,x_1)$ and $d$ distinct choices of $\lambda$ - constituting to an overall addition of $d^2$ possible paths.

In total, summing up all the possible choices for each case results in
\begin{equation} \label{twopartyorbit}
D_2= d^3+d^2+d+1
\end{equation}
as an upper bound to the number of orbits of $\tilde{O}\otimes \tilde{O}$, where $\tilde{O}\in\tilde{\mathcal{O}}_n$, and thus the dimension of $(\mathcal{H}^{\otimes n} \otimes \mathcal{H}^{\otimes n})^{(\tilde{O}\otimes \tilde{O})}$.
%Burnside lemma?

%The bound on the diamond norm of a map with this symmetry is determined by $D_2$. Note that the number $D_2$ only depends on $d$, the dimension of the single Hilbert space $\mathcal{H}$, and not on $n$, the amount of copies of this Hilbert space, which is a crucial improvement.

%\textcolor{green}{Write more?}

\subsection{Orbits of Discrete Orthogonal Group for N Parties}
\label{orbitcountingresultsN}

From the computation of the two party case \ref{twopartyorbit}, it can already be suspected that there may be a recursion relation associated to the number of tensor products of $\tilde{O}$: if the first choice $x_1=0$, for which there is only one choice, the situation is mapped back to the situation where there is only one vector (\ref{onepartyorbit}), which corresponds to having $D_1=d+1$ choices for $\beta(x_2,x_2)$. Similarly, when $x_1$ and $x_2$ are non-zero and linearly independent, it is clear that there will be an additional factor of $d$ for the choice of each new party ($\beta(x_1,x_1)$) and for each cross term ($\beta(x_1,x_2)$). For linear dependency, each possibility of having two vectors linearly dependent contributes a factor of $d$, while fixing cross terms.
Here, this argument will be generalized to find a recursion relation for the dimension $D_K$ that is associated to the subspace that is invariant under $\tilde{O}^{\otimes K}$.

We thus assume that the number of orbits of $K-1$ applications of $\tilde{O}$, $D_{K-1}$, is known, and an additional $K$-th $\tilde{O}$ is brought in. %For convenience, imagine that the additional $O$ comes in at the top of the tree diagram and is chosen first; this will make the argument easier to formulate.

For a more instructive argument, imagine adding the additional $x_K$ at the top of the tree diagram, choosing it first (since the order of choosing should not matter). If $x_K=0$, this branch will reproduce the tree diagram for the $K-1$ case. If $x_K\neq 0$, there are $d$ choices for $\beta(x_K,x_K)$. Then, all subsequent choices for $x_i$ will either linearly depend on $x_K$ (introducing a factor $d$ for $\lambda_K$ each time) or be linearly independent (introducing a factor of $d$ because of the cross term $\beta(x_K,x_i)$. In total, each $x_i$ introduces a factor of $d$, making it a total of $d^{K-1}$ for all subsequent choices. In total, introducing $x_K\neq 0$ will contribute $dd^{K-1}=d^K$. By assuming that all other choices are the same as in the $K-1$ case, we may be overcounting again, because new linear dependencies could emerge, but our result is an upper bound in any case. Thus, there must be a factor of $d^K D_{K-1}$.

In combining the $x_K=0$ case and the $x_K\neq 0$ case, we obtain the following recursion relation:
\begin{equation}
D_K = (d^K +1)D_{K-1} = \prod_{k=0}^K (d^k +1)
\end{equation}

This reproduces all the countings made by hand for $K=1$ (see equation (\ref{onepartyorbit})), $K=2$ (see equation (\ref{twopartyorbit})), and $K=3$ and $K=4$.% (which is the result for 3 parties, see Appendix \ref{appendix-Nparty}).

%The order of this number is given by
%\begin{equation}
%\mathcal{O}(D_K)=\prod_{k=0}^K d^{K-k}= d^{\sum_{k=0}^K (K-k)}= d^{K (K+1)-\sum_{k=0}^K (k)}= d^{K(K+1)- \frac{K(K+1)}{2}}= d^{\frac{K}{2} (K+1)}
%\end{equation}

This does not correspond to $K$ party protocols. It is important to keep in mind that one copy of the space $\mathcal{H}_E=\mathcal{H}^{\otimes n}$ belongs to Eve, while Alice and Bob share the other one, $\mathcal{H}_{T}=\mathcal{H}^{\otimes n}$. In an $N$ party protocol, Alice shares states with $N$ Bobs, which means that there are $N$ Hilbert spaces $\mathcal{H}_{T_i} =\mathcal{H}^{\otimes n}$ for $i=1,...N$. However, when aiming to apply the postselection theorem, Eve's space must be taken into consideration. In an $N$ party protocol, Eve has access to a copy of each of Alice and Bob's shared Hilbert spaces $\mathcal{H}_{T_i}$, meaning that Eve has access to $N$ Hilbert spaces $\mathcal{H}_{E_i}=\mathcal{H}^{\otimes n}$. In total, this means that there are $2N$ spaces $\mathcal{H}^{\otimes n}$. Therefore, the dimension that will be relevant for the postselection technique in Section \ref{StochasticOrbitCounting} is the dimension of the subspace 
\[\big( (\mathcal{H}^{\otimes n})^{\otimes 2N}\big)^{\tilde{O}^{\otimes 2N}}.\]
In conclusion, from the result for $D_K$, the $N$ party-related dimension can easily be obtained by setting $K=2N$.

\section{Results for QKD with Stochastic Orthogonal Symmetry}
\label{StochasticOrbitCounting}

In this section, the postselection technique is applied to the stochastic orthogonal symmetry group $\mathcal{O}_n$ introduced in Definition \ref{newsym-stochastic}, which leads to a mathematical bound on the diamond norm of $\mathcal{O}_n$-invariant maps (as shown in Section \ref{postselectiongeneralized}), and subsequently to a result on the security of $\mathcal{O}_n$-invariant QKD protocols (as shown in Section \ref{horriblecalculation}).
Both of these results depend on one number, the dimension of the $\mathcal{O}_n$-invariant subspace $g_{n,d}=\dim(\mathcal{N})=\dim((\mathcal{H}^{\otimes n} \otimes \mathcal{H}^{\otimes n} )^{O\otimes O})$. This number can be computed via the counting of orbits of the discrete orthogonal group $\tilde{\mathcal{O}}_{n-1}$, which has been achieved in Section \ref{orbitcountingresults}.

Postselection technique can only be applied to the stochastic orthogonal group because it fulfills the central assumption: resolution of identity. The resolution of identity of the stochastic orthogonal group is ensured because an appropriate integration measure has been shown to exist in \cite{Gross}; for discrete orthogonal group, it is not guaranteed.

As alluded to in Section \ref{newsym}, the discrete orthogonal group (for $n-1$ copies) and stochastic orthogonal group (for $n$ copies) are closely related to each other for $n$ not a multiple of $d$. This is a result of the fact that the difference between the two groups is the fact that the stochastic orthogonal group preserves the all-ones-vector, while the discrete orthogonal group does not. If the all-ones vector is not self-orthogonal (i.e. $n$ is not a multiple of $d$), the vector space can be decomposed into a direct sum of a part spanned by the all-ones vector and its orthocomplement, and the stochastic orthonal matrices are block-diagonalized in the corresponding basis, where one block acts on the all-ones vector's span, and one block acts on its orthocomplement. The block acting on the orthocomplement corresponds to a discrete orthogonal matrix in $\tilde{\mathcal{O}}_{n-1}$, preserving the linear product on this part of the vector space. Then, any vector ${v}$ on the whole vector space can be written as a sum ${v}=k{v_1}+{v_2}$ where ${v_1}$ denotes the all-ones vector, the factor $k$ ranges from $0,1,...,d-1$ and ${v_2}$ is a vector in the orthocomplement space of the all-ones vector, which means ${v_2}$ is a vector in the vector space where the discrete orthogonal group acts. For each orbit of the discrete orthogonal group, the stochastic orthogonal group has $d$ orbits, corresponding to the choices of $k$. In total, therefore, the number of orbits of the stochastic orthogonal group for $n$ copies corresponds to the number of discrete orthogonal group orbits for $n-1$ copies multiplied by $d$. %\textcolor{green}{why?}

Because of this, the results obtained from orbit counting for $\tilde{\mathcal{O}_{n-1}}$ have direct implications for the postselection technique for $\mathcal{O}_{n}$, for $n$ not a multiple of $d$:

\begin{equation} \label{g_nd}
g_{n,d}=\dim(\dim(\mathcal{H}^{\otimes n} \otimes \mathcal{H}^{\otimes n} )^{O\otimes O}) =d \dim(\mathcal{H}^{\otimes (n-1)} \otimes \mathcal{H}^{\otimes (n-1)} )^{\tilde{O}\otimes \tilde{O}}= d^4+d^3+d^2+d
\end{equation}

For application to QKD, as established in Section \ref{horriblecalculation}, the $\epsilon'=g_{n,d}\epsilon=\dim(\mathcal{H}_T\otimes \mathcal{H}_E)^{s\otimes\overline{s}} \epsilon$-security of a protocol $\mathcal{E}'$ against general attacks can be inferred from $\epsilon$-security against collective attacks of a protocol $\mathcal{E}$, if $\mathcal{E}'$ is obtained from $E$ with an additional privacy amplification shortening the key by $2\log g_{n,d}= 2\log\dim((\mathcal{H}_T\otimes \mathcal{H}_E)^{s\otimes\overline{s}})$ bits.
Therefore, for a two party QKD protocol, the security of general attacks can be inferred from the security of collective attacks at the cost of $g_{n,d}$ in the error, and $2\log g_{n,d}$ in the key length, with $g_{n,d}$ as given in equation \eqref{g_nd}.

This result can easily be extended to accommodate $N$ party QKD protocols, with the following result:

\begin{equation}
\dim\Big(\big((\mathcal{H}^{\otimes n})^{\otimes 2N} \big)^{O^{\otimes 2N}}\Big)=g_{n,d,N}=\prod_{k=0}^{2N}d (d^k +1)
\end{equation}

as the multiplicative factor between collective and general attacks.

In comparison to previous results for permutation invariant maps, this number constitutes a significant improvement, as it is very small and does not depend on the number of copies $n$. Permutation-based postselection technique leads to a polynomial in $n$: \[g_{n,d}^{(\mathcal{P}_n)} = (n+1)^{d^2-1}\] for two parties, and \[g_{n,d,N}^{(\mathcal{P}_n)}= (n+1)^{d^{2N}-1}\] for $N$ parties (see \eqref{permutationsgtd}, and explained briefly in Section \ref{orbitcounting}).

This means that the multiplicative factor between the diamond distance of CPTP maps and the trace distance with one particular input state is smaller, meaning there is a smaller gap between the upper bound and true diamond norm. Furthermore, since the stochastic orthogonal symmetry preserves tensor powers of stabilizer states, the relevant input state is a purification of a de Finetti state constructed with tensor powers of stabilizer states. In a given setting, there are significantly fewer stabilizer states than arbitrary states, which is also an improvement.

To use this in the context of an actual QKD protocol, the number of copies $n$ should preferably be rather high in order to get a raw key that is as long as possible. After the key has been sifted and error corrected, privacy amplification is performed to transform it into a shorter, but completely secret key. During this step, it is preferable to shorten the key as little as possible, to ensure that the final key that Alice and Bob share is as long as possible. In case of a general attack, there is an additional privacy amplification step, where the same holds true: we want to shorten the key as little as possible, while increasing the error as little as possible.
During this step, the shortening of the key and the change in the error is determined by $g_{n,d}$. For our result, this means that the multiplicative change in the error is smaller, and the key has to be shortened less in comparison with previous schemes.

However, for the postselection technique to be applicable with stochastic orthogonal group, a map with this symmetry is needed. In the context of QKD, this would require a protocol with stochastic orthogonal symmetry. While such a protocol could likely be constructed, investigating existing protocols also has some potential, for example in form of 6-state protocol \cite{Mertz13}, $N$-party 6-state protocol \cite{Grasseli18} or protocols with orthogonal symmetry in continuous variable schemes \cite{LKGC09}.

%\textcolor{blue}{a) why is this so good?}

%Why is it so good? - the thing about it being pretty big that Joe said that one time?

%STARTT

%SATZZEICHEN

%START

\chapter{An SDP Hierarchy for Maximum Channel Fidelity Based on Stabilizer de Finetti Theorem}
\label{chapter-hierarchy}

While QKD security was initially the chief motivation for studying quantum de Finetti theorems, the focus has long shifted to other applications as well. One such example is the approximation of separable states using a hierarchy of SDPs \cite{BBFS18}, based on DPS hierarchy \cite{DPS02,DPS04}.
Under symmetric extension of one party's system (e.g. Bob's), a bipartite state's marginal becomes close to separable in the cut between the two parties (Alice and Bob), with closeness determined by a de Finetti theorem. In other words, a state $\rho_{AB_1B_2\cdots B_n}\equiv \rho_{AB_1^n}$ which is invariant under permutations of the subsystems $B_i$, can be approximated by a state that is separable in the cut between $A$ and $B_1$ via a convergent hierarchy of SDPs.
%However, as shown in \cite{BBFS18}, a particular de Finetti theorem with additional linear constraints is needed to ensure that the above hierarchy will converge towards a separable quantum state.
%Instead of permutation-based de Finetti theorems, the same hierarchy could be constructed for states with invariance under the stochastic orthogonal group. This in itself is interesting for some problems, e.g. optimizing over the convex hull of stabilizers, which could be relaxed to a hierarchy of stochastic orthogonal invariant states. In previous approaches to this problem, this kind of optimization was found to be a linear problem \cite{Heinrich19} in principle (though containing some tedious enumeration), and it is therefore not clear at all that an SDP relaxation would provide an improvement. However, particularly because of the exponential convergence of the Stabilizer de Finetti Theorem, it might.
In this Chapter, an analogue will be shown for stochastic orthogonal transformations instead of permutations, leading to an approximation by separable and (partly) stabilizer states. 

In Section \ref{sec-maxchannelfidelity}, a motivation for considering this kind of problem will be given in the form of a QEC application. Section \ref{section-stabdefinettithm} states and proves the underlying stabilizer de Finetti statement with linear constraints before giving the associated SDP hierarchy and using it to establish convergence in \ref{section-hierarchy}. While this chapter will focus on the stabilizer de Finetti theorem with linear constraints obtained from Theorem \ref{StabDeFinetti}, an analogous version for qubits using Theorem \ref{StabDeFinettiQubit} will be mentioned, and an extended version with stochastic orthogonal invariance on both Alice's and Bob's side is given in Appendix \ref{appendix-bothsides}. Finally, in Section \ref{section-numerics}, some thoughts pertaining to numerical results will be briefly stated.

\section{Approximating Maximum Channel Fidelity}
\label{sec-maxchannelfidelity}

When classical communication channels are unreliable, the successful transfer of a message is by no means certain.
Common noise models include errors occuring randomly but with fixed probability, or dynamic models where errors may occur in bursts. When trying to counteract such noise with error correction, it is integral to know if errors occur, and with what probability they occur, which depends on the nature and amount of the imposed noise and the length of the message. Therefore, a chief quantity of interest is the maximum success probability for transmitting a uniform $d_M$-dimensional message over a channel with a noise model $N_{X\rightarrow Y}$, given by $p(N,d_M)$. Finding this probability for different error correction procedures then gives a hint as to which error correction is most successful. % The maximum probability of success of transmission depends on the nature and amount of the imposed noise, and the length of the message. Then, the best-case success probability of sending a message of length $M$ through a classical communication channel with noise model $N$ is given by $p(N,M)$.

Determining this success probability is a bilinear maximization problem, for which the solution is in general NP-hard to approximate. However, there are methods to approximate the solution from below as well as above, where the latter is achieved by a linear programming relaxation of the problem \cite{Hayashi09,PPV10}. This linear programming relaxation $lp(N,d_M)$ is efficiently computable and has many useful analytic properties.

In similar spirit, the task of understanding data transfer in a quantum setting, i.e. transferring a quantum state over a noisy quantum channel, gains relevance as devices evolve and improve.
Instead of maximum success probability, one most commonly considers the maximum channel fidelity $F_c(N,d_M)$ for transmitting one part of a maximally entangled state over a noisy channel, which can then be used to analyze and compare existing QEC procedures.

For a channel $\mathcal{C}$ that transforms some input state $\rho_{in}$ into an output state $\rho_{out}$, its fidelity is related to the overlap of the transformed input $C(\rho_{in})$ and the desired output $\rho_{out}$, which quantifies how much they differ. The fidelity of two states is therefore defined in the following way:

\[F_s(\rho_{out}, C(\rho_{in})
) = \big\|  \sqrt{\rho_{out}}\sqrt{C(\rho_{in})} \big\|_{\tr}^2\]

When defining channel fidelity for an error correcting procedure, the input state is given by the maximally entangled state $\Phi_{AR}$, and the goal is to transfer one part of it from system $A$ to system $\tilde{B}$ via a channel $ \mathcal{D}_{B\rightarrow\tilde{B}} \circ \mathcal{N}_{\tilde{A}\rightarrow {B}} \circ \mathcal{E}_{A\rightarrow\tilde{A}}$ (without affecting the part on system $R$). The sytems $A$, $R$ and $\tilde{B}$ are all of dimension $d_M$. In applying this channel, the quantum state passes an encoder channel $\mathcal{E}_{A\rightarrow\tilde{A}}$, a noisy transmission channel $\mathcal{N}_{\tilde{A}\rightarrow {B}}$, and a decoder channel $\mathcal{D}_{B\rightarrow\tilde{B}}$. To quantify how well this channel transfers a part of the maximally entangled state, we want to compare the transformed input $\big( ( \mathcal{D}_{B\rightarrow\tilde{B}} \circ \mathcal{N}_{\tilde{A}\rightarrow {B}} \circ \mathcal{E}_{A\rightarrow\tilde{A}})\otimes \mathbbm{1}_R \big) (\Phi_{AR})$ and the desired output $\Phi_{\tilde{B}R}$.

To determine maximum channel fidelity, we want to use the best possible encoder and decoder. This translates to the following maximization problem for determining $F_c(N,d_M)$:

\begin{optimize}\label{Fidelity-original}
\begin{equation*} \begin{split} 
F(N,d_M) =\text{maximize} \ & F_s\Big(\Phi_{\tilde{B}R}, \big( ( \mathcal{D}_{B\rightarrow\tilde{B}} \circ \mathcal{N}_{\tilde{A}\rightarrow {B}} \circ \mathcal{E}_{A\rightarrow\tilde{A}})\otimes \mathbbm{1}_R \big) (\Phi_{AR})\Big)
\\
\text{subject to }& \mathcal{E}_{A\rightarrow\tilde{A}}, \mathcal{D}_{B\rightarrow\tilde{B}} \text{ are quantum channels}
\end{split}
\end{equation*}
\end{optimize}

Using the Choi-Jamio{\l}kowski isomorphism, this can be rewritten as the following bilinear optimization problem (note the similarity to the classical case) with matrix-valued variables (see \cite{BBFS18}, Lemma 5.2): %Appendix \ref{appendix-Choi}): %Finding $F(N,M)$ is a bilinear optimization problem with matrix-valued variables:

\begin{optimize} \label{ChannelFidelity}
\begin{equation*} \begin{split}
F(N,d_M) =\text{maximize} \ & d_{\tilde{A}} d_B \Tr \Big( \big(J_{\tilde{A}B}^N \otimes \Phi_{A\tilde{B}} \big) \big(  E_{A\tilde{A}} \otimes D_{B\tilde{B}} \big) \Big)
\\
\text{subject to }&  E_{A\tilde{A}} \geq 0, \ D_{B\tilde{B}} \geq 0 \\
  & \Tr_{\tilde{A}}( E_{A\tilde{A}}) =\frac{\mathbbm{1}_A}{d_A}, \ \Tr_{\tilde{B}}( D_{B\tilde{B}})  = \frac{\mathbbm{1}_B}{d_B}\\
\end{split}
\end{equation*}
\end{optimize}

where $\Phi_{A\tilde{B}}$ is a maximally entangled state of dimension $d_M$ and $J_{\tilde{A}B}^N$ is the normalized Choi state corresponding to the noisy channel transmitting a state from $\tilde{A}$ to $B$.

The maximum fidelity can be bound from below by seesaw methods \cite{RW05}, and there exists an SDP relaxation to bound it from above \cite{LM15}. (It is also worth mentioning that there exist converse bounds which bound the message length $d_M$ for a given fidelity \cite{TBR16,WD16,WFD19,KDWW19}.) While this SDP relaxation is efficiently computable, the gap between the SDP relaxation's solution and the actual maximum fidelity is not well understood. In \cite{BBFS18}, a converging hierarchy of SDP relaxations on the maximum fidelity is proposed, enabling us to study $F(N,d_M)$ directly. (In fact, the first level of the hierarchy reproduces the bounds of \cite{LM15}.) This relies on the idea that separable states (like $E_{A\tilde{A}} \otimes D_{B\tilde{B}}$) can be approximated by a hierarchy \cite{Lasserre00,Parrilo03,DPS02,DPS04}.

This hierarchy relies on two key concepts: Firstly, that separable states $\rho_{AB}$ are $n$-extendible, which means that $n-1$ systems can be added on Bob's side, extending the state to $\rho_{AB_1 \cdots B_n}\equiv \rho_{AB_1^n}$, such that $ \rho_{AB_1^n}$ is invariant under permutations of the subsystems of $B_1^n$. For example, a separable quantum state $\rho_{AB}=\omega_A\otimes \tau_B$ can be $n$-extended to $\rho_{AB_1^n}=\omega_A\otimes\tau_{B_1} \otimes \cdots \otimes \tau_{B_n}$, which is obviously invariant under permutation of $B_1^n$. However, given a state $\rho_{AB_1^n}$ that is invariant under permutations of the subsystems of $B_1^n$, it is not immediately implied that it originates in an $n$-extension of a separable state; but it is close, and this closeness can be quantified in terms of a de Finetti theorem, which is the second key concept. This  approximation improves with increasing $n$, as the de Finetti error decreases.

Using the extendibility property, the following approximation of maximum channel fidelity in \eqref{ChannelFidelity} can be proposed:

\begin{optimize} \label{hhhhhhh}
\begin{equation*} \begin{split}
F^{(n)}(N,d_M)=\text{maximize} \ & d_{\tilde{A}} d_B \Tr \Big( \big(J_{\tilde{A}B}^N \otimes \Phi_{A\tilde{B}} \big) \big(  \rho_{A\tilde{A}B\tilde{B}} \big) \Big)\\
\text{subject to }&  \rho_{A\tilde{A}(B\tilde{B})_1^n}  \geq 0, \ \Tr(\rho_{A\tilde{A}(B\tilde{B})_1^n})=1 \\
& \rho_{A\tilde{A}(B\tilde{B})_1^n} \text{ is invariant under all permutations} \\
  & \Tr_{\tilde{A}}( \rho_{A\tilde{A}(B\tilde{B})_1^n} ) =\frac{\mathbbm{1}_A}{d_A}\otimes \rho_{(B\tilde{B})_1^n}\\&
\Tr_{\tilde{B}_n}( \rho_{A\tilde{A}(B\tilde{B})_1^n})  = \rho_{A\tilde{A}(B\tilde{B})_1^{n-1}} \otimes \frac{\mathbbm{1}_{B_{n}}}{d_B} \\
\end{split}
\end{equation*}
\end{optimize}

This corresponds to level $n$ of a hierarchy approximating maximum channel fidelity. Then, the convergence of such a hiearchy relies on a de Finetti theorem.

It is important to note that a specific type of de Finetti theorem is needed here; in particular, the standard version with best known convergence in \ref{thm-DeFinetti} is not applicable in this scenario.
While any extendability property leads to separability on the side where permutation invariance is imposed, i.e. between the systems $B_i$ in the above example and additional separability in the cut between Alice and Bob (or encoder and decoder), the additional constraint on the encoder's and decoder's marginal is not guaranteed in general. %For this reason, convergence of the hierarchy towards this separable decomposition is integral. % However, this separability is in fact crucial to the problem of determining maximum channel fidelity. For this reason, it is integral %Having separability between Alice ($A\tilde{A}$) and Bob ($B\tilde{B}$) 

The convergence towards a separable state with desired constraints on the marginal ($Tr_{\tilde{A}}( E_{A\tilde{A}}) =\frac{\mathbbm{1}_A}{d_A}$, $\Tr_{\tilde{B}}( D_{B\tilde{B}})  = \frac{\mathbbm{1}_B}{d_B}$ in \ref{ChannelFidelity}) can be ensured by imposing additional linear constraints of a particular form on the state during the hierarchy (namely, $ \Tr_{\tilde{A}}( \rho_{A\tilde{A}(B\tilde{B})_1^n} ) =\frac{\mathbbm{1}_A}{d_A}\otimes \rho_{(B\tilde{B})_1^n}$ and $
\Tr_{\tilde{B}_n}( \rho_{A\tilde{A}(B\tilde{B})_1^n})  = \rho_{A\tilde{A}(B\tilde{B})_1^{n-1}} \otimes \frac{\mathbbm{1}_{B_{n}}}{d_B}$ at level $n$ in \ref{hhhhhhh}), leading to the de Finetti theorem found in \cite{BBFS18} and stated below. Because a de Finetti theorem incorporating these constraints can be found, the hierarchy can be shown to converge towards the desired form, i.e. towards a solution of \ref{ChannelFidelity}.

One could propose a simplified hierarchy where these constraints are modified to be local constraints, i.e. $\Tr_{\tilde{A}}( \rho_{A\tilde{A}}) =\frac{\mathbbm{1}_A}{d_A}$ and $\Tr_{\tilde{B}_n}( \rho_{B_n\tilde{B}_n})= \frac{\mathbbm{1}_{B_{n}}}{d_B}$ instead of constraints on the larger state as in \eqref{hhhhhhh}, which allows for communication between Alice and Bob. However, this would not lead to a convergent hierarchy approximating the desired form, as it would not lead to a mixture of normalized channels. In other words, not having these linear constraints would lead to a convergence towards a mixture of maps which are completely positive, but not trace preserving (in fact, maybe not even trace non-increasing). Thereby, the desired constraint does not hold for each summand in the mixture separately. For our applications, each summand must obey the desired linear constraints separately, so that each part of the mixture corresponds to a CPTP map via Choi-Jamio{\l}kowski isomorphism. In summary, what we want is separability with states that correspond to CPTP maps instead of just separability.
%this mixture must correspond to shared randomness assistance, i.e. a mixture of CPTP maps, but not allow for communication. Having no constraints, or local constraints, however, allows for classical communication.
For details and counter-examples of theorems with local (or no) linear constraints, we refer to Examples 3.7 in \cite{BBFS18}.

%which is crucial to ensure the convergence of the hierarchy.%However, this is crucial for the convergence of the hierarchy, where encoder and decoder are written as %for any valid solution to optimization problem \ref{Fidelity-original}, they must be quantum channels.
The theorem which ensures convergence of \eqref{hhhhhhh} towards maximum channel fidelity is the following:

\begin{theorem}[De Finetti Theorem with Linear Constraints, see \cite{BBFS18}, Theorem 3.4] \label{LinDeFinetti}
Let $\rho_{AB_1^n}$ be a quantum state that is permutation invariant with respect to permutations of the $n$ subsystems $B_1^n$. Let $\Lambda_{A\rightarrow C_A}$ and $\Gamma_{B\rightarrow C_B}$ be linear maps, and $X_{C_A}$ and $Y_{C_B}$ be operators such that the following two linear constraints hold:
\[ \Lambda_{A\rightarrow C_A}(\rho_{AB_1^n})=X_{C_A}\otimes \rho_{B_1^n} ,\]
%\[ \Gamma_{B_{m+1} \rightarrow C_B}(\rho_{B_1^n})=\rho_{B_1^{m}} \otimes Y_{C_B} \otimes \rho_{B_{m+2}^n}\]
\[ \Gamma_{B_{n} \rightarrow C_B}(\rho_{B_1^n})=\rho_{B_1^{n-1}} \otimes Y_{C_B}.\]
%Then, there exists a probability distribution $\{p_Z(z)\}_{z\in Z}$, and $0\leq m <n-1$ such that %quantum states $\omega_A^z$ and $\sigma_B^z$, and $0\leq m <n-1$ such that
Then, there exists an $m\in [0,n-1]$ and a probability distribution $\{p_Z(z_1^m)\}_{z_1^m\in Z}$ such that
\[\Big\| \rho_{AB_{m+1}} - \sum_{z_1^m} p_Z(z_1^m) \rho_{A|z_1^m} \otimes \rho_{B_{m+1}|z_1^m} \Big\|_{\tr} \leq  \epsilon(d_B,d_A,n)\]
where
\begin{equation} \label{epsilon}
 \epsilon(d_B,d_A,n):=\min\Big\{d_B^2(d_B+1), 18\sqrt{d_A d_B}\Big\} \sqrt{\frac{2\ln(2)\ln(d_A)}{n}}
\end{equation}
and \[\Lambda_{A\rightarrow C_A}(\rho_{A|z_1^m})=X_{C_A},\ \Gamma_{B_{m+1} \rightarrow C_B}(\rho_{B_{m+1}|z_1^m})= Y_{C_B}.\]
\end{theorem}

\begin{remark}
Note that Theorem \ref{LinDeFinetti} is not equivalent to the full theorem as it appears in \cite{BBFS18}, but rather appears at an intermediate step in the proof of their main theorem. For completeness, the full proof of this modified statement is given in Appendix \ref{appendix-proofforhierarchy}.
\end{remark}

\begin{remark}
%Let us reiterate that this theorem differs from previous versions (which at first glance have better convergence, see Theorem \ref{thm-DeFinetti}) in two ways: Firstly, by the additional linear constraints imposed on the state, and secondly, by the fact that it results an additional separabilty in the cut between Alice and Bob's parts of the system (in the limit $n\rightarrow \infty$).
%Without the additional linear constraints, separability in this cut would not be ensured. 
%Notably, this theorem differs from previous versions (which at first glance have better convergence, see Theorem \ref{thm-DeFinetti}) by the additional linear constraints. The reason why it is necessary to prove a new theorem instead of imposing additional constraints on an existing theorem is fairly subtle. It is precisely these constraints that ensure that the resulting encoder and decoder Choi matrices will correspond to completely positive, trace-preserving maps. Without these constraints, they are only guaranteed to be completely positive, but may not even be trace-nonincreasing. For details, we refer to Example 3.7 in \cite{BBFS18}. It is integral that the imposed constraints contain some restriction on where correlations can appear, which guarantees the separability in the cut between Alice's and Bob's part of the system.
%
As noted before, the linear constraints are of a particular form, requiring $\Lambda_{A\rightarrow C_A}(\rho_{AB_1^n})=X_{C_A}\otimes \rho_{B_1^n}$ instead of a localized version $\Lambda_{A\rightarrow C_A}(\rho_{AB_1^n})=X_{C_A}$. While these two conditions are equivalent under the trace, it is important that they are not equivalent in general. Phrasing the constraint like this ensures that any correlations may only be within Alice's or within Bob's side, not between them, and it is crucial for the convergence of the hierarchy towards the desired form in \ref{ChannelFidelity}.
\end{remark}

As found in \cite{BBFS18}, using this theorem
while renaming $A\rightarrow A\tilde{A}$, $B\rightarrow B\tilde{B}$, $C_A={A}$, $\Lambda_{A\tilde{A} \rightarrow {A}} =\Tr_{\tilde{A}}$ and $X_{A}=\frac{\mathbbm{1}_A}{d_A}$, as well as $C_B={B}$, $\Gamma_{B\tilde{B} \rightarrow {B}} =\Tr_{\tilde{B}}$ and $Y_{B}=\frac{\mathbbm{1}_B}{d_B}$, it can be shown that the optimal values of the SDP relaxation in \eqref{hhhhhhh} converge to the optimal value of \eqref{ChannelFidelity} for $n\rightarrow \infty$.
The difference between the state $\rho_{A\tilde{A}B\tilde{B}}$ obeying the constraints in \eqref{hhhhhhh}, and the state $E_{A\tilde{A}} \otimes D_{B\tilde{B}}$ obeying the constraints listed in \eqref{ChannelFidelity} is essentially given by the error bound in Theorem \ref{LinDeFinetti}, which means that the proof and speed of this convergence rely exclusively on Theorem \ref{LinDeFinetti}.
\\
\\
In this project, we want to investigate the problem of studying maximum channel fidelity using stabilizer de Finetti theorems instead of the traditional de Finetti theorems with permutation invariance. Because Theorems \ref{StabDeFinetti} and \ref{StabDeFinettiQubit} describe closeness to tensor powers of stabilizer states rather than arbitrary states, this would lead to a new hierarchy containing additional constraints which ensure that the encoder, decoder or encoder \textit{and} decoder are related to Choi-matrixes of stabilizer states, which %greatly reduces the amount of states one has to optimize over.
are Clifford operations. This does not only greatly reduce the amount of states one has to optimize over, but also makes it interesting for studying Clifford-related problems in general, in particular for studying Clifford encoders and/or decoders rather than arbitrary encoders and decoders. % \textcolor{green}{Is this correct, Joe? And also: should we move everything over to the A side? Because clifford encoder is probably better than decoder?}

In short, we want to approximate (for example) the following fidelity:
\begin{optimize}\label{Fidelity-CliffordDecoder}
\begin{equation*} \begin{split} 
F_C(N,d_M)=\text{maximize} \ & F_s\Big(\Phi_{\tilde{B}R}, \big( ( \mathcal{D}_{B\rightarrow\tilde{B}} \circ \mathcal{N}_{\tilde{A}\rightarrow {B}} \circ \mathcal{E}_{A\rightarrow\tilde{A}})\otimes \mathbbm{1}_R \big) (\Phi_{AR})\Big)
\\
\text{subject to }& \mathcal{E}_{A\rightarrow\tilde{A}} \text{ is a quantum channel}\\
&  \mathcal{D}_{B\rightarrow\tilde{B}} \text{ is a Clifford channel}
\end{split}
\end{equation*}
\end{optimize}

We find that this is indeed possible for stochastic orthogonal symmetry at a small cost to precision by finding a stabilizer de Finetti theorem with additional linear constraints. This cost can be neglected in comparison to the overall error. From this theorem, an analogous converging hierarchy for states with stochastic orthogonal symmetry instead of permutation symmetry can be proposed for odd prime dimensions $d_B$. Using this hierarchy, we can approximate the maximum channel fidelity of an optimal (arbitrary) encoder and optimal Clifford decoder instead of the maximum channel fidelity of an optimal arbitrary encoder and decoder. In addition, a hierarchy for qubits can also be found for states with invariance under permutations plus anti-identity, following from stabilizer de Finetti theorem for qubits \ref{StabDeFinettiQubit}.
However, since both of these results are obtained very similarly, details will only be given for the case of stochastic orthogonal symmetry.%, before giving the qubit version which can be found and proven completely analogously.

Because our results will be symmetric under exchange of Alice and Bob, note that our theorem and subsequent hierarchy can also be employed to approximate maximum channel fidelity of an optimal Clifford encoder and an optimal (arbitrary) decoder. Furthermore, the theorem can be extended to lead to a hierarchy for finding maximum channel fidelity of an optimal Clifford encoder and an optimal Clifford decoder, see Appendix \ref{appendix-bothsides}.

\begin{remark}
In Theorem \ref{LinDeFinetti}, when permuting the systems $B_1^n$, each subsystem $B_i$ has local Hilbert space dimension $d_B$, containing $d_B$-dimensional qudits. In the case of stochastic orthogonal transformations of $B_1^n$, each subsystem consists of $r$ qudits, which means that each subsystem $B_i$ has a Hilbert space dimension $d_B^r$. To effectively compare results for permutation invariance and stochastic orthogonal invariance, this difference needs to be taken into account, i.e. one has to look at $\epsilon(d_B^r,d_A,n)$. However, comparing results and bounds is not the main objective of replacing permutation invariance by stochastic orthogonal invariance - instead, this replacement is motivated by the prospect of studying Clifford decoders and encoders instead of arbitrary encoders and decoders.
\end{remark}

\section{Stabilizer de Finetti Theorem with Linear Constraints}
\label{section-stabdefinettithm}

%At first glance, it does not seem like incorporating linear constraints like the ones in \ref{ChannelFidelity} into a standard DPS hierarchy requires a new de Finetti theorem. Naively, one would think that local, linear constraints could be added on top of the original de Finetti theorem. However, this is not true, and the reason for that is fairly subtle.
%It is because the linear constraints are not actually local, but instead impose something.
%At first glance, one might consider that the original de Finetti theorem already suffices to obtain a de Finetti

%START
%NEEDEDHERE

%\textcolor{green}{INTROOOO -why do we need a new de Fin? Why lin constraints?Why constraints: correlations.it is not just a local constraint, but more.Why new de Finetti: incorporating the constraints. (authors said it is not possible, but also it is not possible)}

For proposing an analogue hierarchy to approximate maximum channel fidelity with Clifford encoders or decoders, an appropriate de Finetti theorem is needed to ensure its convergence. Here, we combine Theorem \ref{LinDeFinetti} with the stabilizer de Finetti theorems \ref{StabDeFinetti} and \ref{StabDeFinettiQubit}, leading to an approximation of a state where there is separability in the cut between $A$ and $B$, and $B$ is approximately given by a convex combination of tensor powers of stabilizer states (instead of arbitrary states).

\begin{theorem}[Stabilizer de Finetti Theorem with Linear Constraints]
\label{mainthm}
Let $\rho_{AB_1^n}$ be a quantum state on the Hilbert space $\mathcal{H}_A\otimes\mathcal{H}_B^{\otimes n}$ that commutes with the action of $\mathcal{O}_n$ acting on the systems $B_1^n \equiv B_1 B_2 \cdots B_n$, each of dimension $d_B^r$. Let $d_B$ be an odd prime. Let $\Lambda_{A\rightarrow C_A}$ and $\Gamma_{B\rightarrow C_B}$ be linear maps, and $X_{C_A}$ and $Y_{C_B}$ be operators such that the following two linear constraints hold:
\[ \Lambda_{A\rightarrow C_A}(\rho_{AB_1^n})=X_{C_A}\otimes \rho_{B_1^n}, \]
\[ \Gamma_{B_n \rightarrow C_B}(\rho_{B_1^n})=\rho_{B_1^{n-1}} \otimes Y_{C_B}. \]
Then, there exists a probability distribution $\{p_Z(z)\}_{z\in Z}$ and a probability distribution $\{p_S(\sigma_B)\}$ over the set of mixed stabilizer states on $r$ qudits such that
%\begin{align*}
 %&\Big\|\rho_{AB}-\sum_{z,\sigma_B} p_Z(z) p_{S}(\sigma_B)\rho_{A|z} \otimes  \sigma_B \Big\|_{\tr} \\ & \ \ \ \ \ \ \ \ \ \ \leq \min\Big\{d_B^{2r}(d_B^r+1), 18\sqrt{d_A d_B^r}\Big\} \sqrt{\frac{2\ln(2)\ln(d_A)}{n}}+4d_B^{2(r+1)^2}d_B^{-\frac{1}{2} (n-1)} \\
 %\end{align*}
 \[ \Big\|\rho_{AB}-\sum_{z,\sigma_B} p_Z(z) p_{S}(\sigma_B)\rho_{A|z} \otimes  \sigma_B \Big\|_{\tr}\leq \epsilon(d_B^r,d_A,n)+\bar{\epsilon}(d_B,r,n)\]
 where $\sigma_B$ are the mixed stabilizer states of $r$ qudits on $\mathcal{H}_B={(\mathbbm{C}^{d_B})}^{\otimes r}$, $\epsilon(d_B^r,d_A,n)$ defined in \eqref{epsilon}, 
 \begin{equation}\label{epsilon2}
 \bar{\epsilon}(d_B,r,n):=2d_B^{2(r+1)^2}d_B^{-\frac{1}{2} (n-1)}
 \end{equation}
and

\[\Lambda_{A\rightarrow C_A}(\rho_{A|z})=X_{C_A},\ \Gamma_{B_{n} \rightarrow C_B}(\rho_{B|x})= Y_{C_B},\]

\[
\Big\| \sum_{z} p_Z(z)  \big(\rho_{B|z}- \sum_{\sigma_B} p_S(\sigma_B) \sigma_{B}\big) \Big\|_{\tr} \leq  \bar{\epsilon}(d_B,r,n).
\]
\end{theorem}

This theorem shows that a state that is partially invariant under the stochastic orthogonal group introduced in \eqref{newsym-stochastic} is approximately separable and close to a convex combination of stabilizer states on one side. 
It is a direct combination of two de Finetti like statements: 
On the one hand, it makes use of the de Finetti theorem with additional linear constraints, Theorem \ref{LinDeFinetti}; on the other hand, it relies on the stabilizer de Finetti theorem, Theorem \ref{StabDeFinetti}. The theorem with linear constraints ensures separability in the cut between Alice's and Bob's subsystem, while the stabilizer de Finetti theorem ensures closeness to stabilizer states on Bob's side.

%It is important to note that Bob's side is approximated by stabilizer states, but these stabilizer states do not satisfy the linear constraints directly, but only approximately. 

In Theorem \ref{LinDeFinetti}, the marginal of a permutation invariant state is approximated by a separable state where each part satisfies the linear constraints directly. It is important to note that this is not entirely analogous for stochastic orthogonal invariance: while Bob's side is now approximated by stabilizer states, these stabilizer states do not precisely satisfy the linear constraints. Instead, they satisfy them approximately, with an error corresponding to the bound of stabilizer de Finetti theorem. Notably, this error is much smaller than the overall error for approximating separable and partly stabilizer states.

To relate the two theorems to each other, consider the following definition:
%These two concepts can be combined via the following two corollaries, which will be integral to the proof of Theorem \ref{mainthm}:

\begin{definition}[Post-Measurement Quantum State] \label{rhoAz}
For positive operator valued measures (POVMs) $\{ \Pi_z \}$ mapping from $C$ to classical system $Z$ % projectors $\{ \Pi_z=\ket{z}\bra{z} \}$ mapping from a classical system $Z$ to $Z$,
and a state $\omega_{AC}$ with classical system $Z$, the state after measurement of $z$ is given by:
\[\omega_{A|z}=\dfrac{\tr_Z \Big( (\mathbbm{1}_A\otimes \Pi_z )  \omega_{AC} \Big)}{\Tr_{AZ} \Big( (\mathbbm{1}_A\otimes \Pi_z ) \omega_{AC}\Big)} .\]
\end{definition}

With this definition, the two concepts - Theorem \ref{LinDeFinetti} and Theorem \ref{StabDeFinetti} - can be combined via the following observations, which will be integral to the proof of Theorem \ref{mainthm}:

\begin{observation} \label{CPTPmapLemma}
For all POVMs $\{ \Pi_z \}$ mapping from $C$ to classical system $Z$ and quantum states $\omega_{A|z}$, there exists a CPTP map $\mathcal{M}$ that acts as follows on any quantum state $\rho_{BC}$:
\[\mathcal{M}: \rho_{BC}\mapsto \sum_z \omega_{A|z} \otimes \Tr_C (\Pi_z \rho_{BC}) =\sum_z p_Z(z) \omega_{A|z} \otimes \rho_{B|z} 
\]
with a probability distribution $p_Z(z)$.
\end{observation}

\begin{proof}[Proof of Observation \ref{CPTPmapLemma}]
The map $\mathcal{M}$ can be interpreted as a composition of three maps:
\[\mathcal{M}=\mathcal{M}_3 \circ \mathcal{M}_2 \circ \mathcal{M}_1 \]
The first map $\mathcal{M}_1$ acts as a measurement of the POVMs $\{ \Pi_z \}$.
Applying it to a quantum state $\rho_{BC}$ yields
\[ \mathcal{M}_1: \rho_{BC}\mapsto \sum_z \Tr_C (\Pi_z \rho_{BC})\otimes \ket{z}\bra{z}_Z  .\]
Clearly, this map is trace preserving and completely positive.
The second map recovers a specific purification of the state, introducing system $A$:
\[ \mathcal{M}_2\circ \mathcal{M}_1: \rho_{BC}\mapsto \sum_z \omega_{A|z} \otimes \Tr_C (\mathbbm{1}_B \otimes \Pi_z \rho_{BC})\otimes \ket{z}\bra{z}_Z % =\rho_{ABZ} 
\]
This map is completely positive because $\omega_{A|z}$, as defined in Definition \ref{rhoAz}, is positive semidefinite - and trace preserving because $\omega_{A|z}$ has unit trace.
Lastly, the final map $\mathcal{M}_3$ corresponds to taking the trace over the classical system $Z$:
\[ \mathcal{M}_3\circ\mathcal{M}_2\circ\mathcal{M}_1: \rho_{BC}\mapsto \sum_z \omega_A \otimes \Tr_C (\Pi_z \rho_{BC}) % =\rho_{AB}
 \]
Clearly, taking the trace over $Z$ is a completely positive and trace preserving operation, making the composed map $\mathcal{M}_3 \circ \mathcal{M}_2 \circ \mathcal{M}_1 =\mathcal{M}$ a CPTP map.
\end{proof}

\begin{observation} \label{CPTPmapLemma2}
For all POVMs $\{ \Pi_z \}$ mapping from $C$ to classical system $Z$, there exists a CPTP map ${\mathcal{M}'}$ that acts as follows on any quantum state $\rho_{BC}$:
\[\mathcal{M}': \rho_{BC}\mapsto \sum_z \Tr_C (\Pi_z \rho_{BC}) =\sum_z p_Z(z)  \rho_{B|z} 
\]
with a probability distribution $p_Z(z)$.
\end{observation}

\begin{proof}[Proof of Observation \ref{CPTPmapLemma2}]
The map $\mathcal{M}'$ can be regarded as a composition of two of the maps appearing in the proof of \ref{CPTPmapLemma}. Using the notation from there, it can be rewritten as
\[\mathcal{M}'= \mathcal{M}_3 \circ \mathcal{M}_1 .\]
As shown in the proof of \ref{CPTPmapLemma}, the maps $\mathcal{M}_3$ and $\mathcal{M}_1$ are CPTP, and thus their composition is also CPTP.
\end{proof}

Using Observations \ref{CPTPmapLemma} and \ref{CPTPmapLemma2}, Theorem \ref{mainthm} can be proven. This theorem admits two interpretations: it can be regarded as a stabilizer de Finetti theorem with additional linear constraints - or as a de Finetti theorem with linear constraints with the additional constraint that some parts should be a stabilizer state.

\begin{proof}[Proof of Theorem \ref{mainthm}]
The proof follows the following general outline: Starting from the stabilizer de Finetti theorem in \ref{StabDeFinetti} on Bob's side, we find a measurement and a purification that transform the theorem into a statement about the closeness of separable states and states with stabilizer tensor powers on Bob's side. Then, the de Finetti theorem with linear constraints in \ref{LinDeFinetti} can be used to connect separable states with the full state $\rho_{AB_1^n}$ via triangle inequality, which leads to a theorem bounding the closeness of the full state and states with stabilizer tensor powers on Bob's side.

Since $\rho_{AB_1^n}$ is invariant under the stochastic orthogonal group with representation $O_n$ acting on $B_1^n$, clearly $\tr_A(\rho_{AB_1^n})=\rho_{B_1^n}$ is invariant under $O_n$ acting on $B_1^n$.
Because of this, the stabilizer de Finetti theorem (Theorem \ref{StabDeFinetti}) can be applied. Therefore, for any $k\leq n-1$, we find:
\begin{equation} \label{proof-step1}
\Big\|\rho_{B_1^{k+1}} - \sum_{\sigma_B} p_S(\sigma_B) \sigma_B^{\otimes {k+1}}\Big\|_{\tr} \leq 2d_B^{2(r+1)^2}d_B^{-\frac{1}{2} (n-(k+1))} 
\end{equation}
where $\sigma_B$ is a mixed stabilizer state on $\mathcal{H}_B$. 
Using Observation \ref{CPTPmapLemma}, this can be connected to the de Finetti theorem for permutation invariant states with linear constraints, Theorem \ref{LinDeFinetti}. Choosing
the systems $A\rightarrow A$, $B \rightarrow B_1^k$ and $C \rightarrow B_{k+1}$, a bitstring $z\rightarrow z_1^k$ and \[ \rho_{A|z_1^k}=\frac{\Tr_{B_1^{n}} (\mathbbm{1}_A\otimes \{ \Pi_{z_1^k} \otimes \mathbbm{1}_{B_{k+1}^n} \rho_{AB_1^n})}{\Tr_{AB_1^{n}} (\mathbbm{1}_A\otimes  \{ \Pi_{z_1^k} \otimes \mathbbm{1}_{B_{k+1}^n} \rho_{AB_1^n})}\]
according to Definition \ref{rhoAz} with %projectors $\{ \Pi_{z_1^k}=\ket{z_1^k}\bra{z_1^k}\}$,
POVMs $\{ \Pi_{z_1^k}\}$, we find a CPTP map $\mathcal{M}$ that acts in the following way:
\[\mathcal{M}: X_{B_1^{k+1}} \rightarrow \sum_{z_1^k} \rho_{A|z_1^k} \otimes \Tr_{B_1^k} (\Pi_{z_1^k} \otimes \mathbbm{1}_{B_{k+1}} X_{B_1^{k+1}}) \]
Applying this to the two states of interest in the distance relation above yields
\begin{equation*} \begin{split} \label{proof-step2-1}
\mathcal{M} (\rho_{B_1^{k+1}}) &=\sum_{z_1^k} \rho_{A|z_1^k} \otimes \Tr_{B_1^k} ( \Pi_{z_1^k} \otimes \mathbbm{1}_{B_{k+1}} \rho_{B_1^{k+1}})\\
&= \sum_{z_1^k} p_Z(z_1^k) \rho_{A|z_1^k} \otimes \rho_{B_{k+1}|z_1^k}
\end{split} \end{equation*}
and
\begin{equation*} \begin{split} \label{proof-step2-2}
\mathcal{M}( \sigma^{\otimes (k+1)}) &=\sum_{z_1^k} \rho_{A|z_1^k} \otimes \Tr_{B_1^k} ( \Pi_{z_1^k} \otimes \mathbbm{1}_{B_{k+1}} \sigma^{\otimes (k+1)})\\
&= \sum_{z_1^k} p_Z(z_1^k) \rho_{A|z_1^k} \otimes \sigma_{B_{k+1}}.
\end{split} \end{equation*}

Inserting this into the trace distance relation \eqref{proof-step1}, and using the fact that $\mathcal{M}$ is trace-nonincreasing, yields

\smallskip\noindent
\begin{equation} \begin{split} \label{proof-step-3}
 2d_B^{2(r+1)^2}d_B^{-\frac{1}{2} (n-(k+1))} &\geq \Big\|\rho_{B_1^{k+1}} - \sum_{\sigma_B} p_S(\sigma_B) \sigma_{B_{k+1}} \Big\|_{\tr}\\
& \geq \Big\| \mathcal{M} \big(\rho_{B_1^{k+1}} - \sum_{\sigma_B} p_S(\sigma_B)  \sigma_{B_{k+1}} \big) \Big\|_{\tr}\\
&=\Big\| \sum_{z_1^k} p_Z(z_1^k) \rho_{A|z_1^k} \otimes \big(\rho_{B_{k+1}|z_1^k}- \sum_{\sigma_B} p_S(\sigma_B) \sigma_{B_{k+1}}\big) \Big\|_{\tr}.
 \end{split}
\end{equation}

In total, we find a statement about separable states which are obtained from a state with invariance under $\mathcal{O}_n$ being close to a mixture of stabilizer states:
\begin{equation} \label{proof-step-4}
 \Big\| \sum_{{z}_1^k}p_Z(z_1^k) \rho_{A|{z}_1^k} \otimes  \big(\rho_{B_{k+1}|z_1^k}   -  \sum_{\sigma} p_S(\sigma) \sigma_{B_{k+1}}\big)\Big\|_{\tr}  \leq 2d_B^{2(r+1)^2}d_B^{-\frac{1}{2} (n-(k+1))}
\end{equation}
The de Finetti theorem with linear constraints (Theorem \ref{LinDeFinetti}) relates a state's closeness to such a separable state. It has to be noted that permutations are a subgroup of the stochastic orthogonal group, which means that $\rho_{B_1^n}$ is also invariant with respect to permutations of the $B_1^n$ subsystems (but with dimension $d_B^r$, since each subsystem consists of $r$ qudits). Therefore, given the linear constraints listed in the theorem, there exists an $m\in [0,n-1]$, such that
\[\Big\| \rho_{AB_{m+1}} - \sum_{z_1^m} p_Z(z_1^m) \rho_{A|z_1^m} \otimes \rho_{B_{m+1}|z_1^m} \Big\|_{\tr} \leq \min\Big\{d_B^{2r}(d_B^r+1), 18\sqrt{d_A d_B^r}\Big\} \sqrt{\frac{2\ln(2)\ln(d_A)}{n}} =\epsilon(d_B^r,d_A,n) \]
with \[\Lambda_{A\rightarrow C_A}(\rho_{A|z_1^m})=X_{C_A},\ \Gamma_{B_{m+1} \rightarrow C_B}(\rho_{B_{m+1}|z_1^m})= Y_{C_B}.\]

Choosing $k=m$ for the stabilizer de Finetti theorem and combining the two statements \eqref{proof-step-3} and \eqref{proof-step-4} yields the following relation, which can be bound using the triangle inequality:
\smallskip\noindent
\begin{align*}
&\Big\|  \rho_{AB_{m+1}} - \sum_{z_1^m} p_Z(z_1^m) \rho_{A|z_1^m} \otimes \sum_{\sigma} p_S(\sigma) \sigma_{B_{m+1}}) \Big\|_{\tr} \\&
=
\Big\|  \rho_{AB_{m+1}} - \sum_{z_1^m} p_Z(z_1^m) \rho_{A|z_1^m} \otimes \rho_{B_{m+1}|z_1^m} + \sum_{{z}_1^m}p_Z(z_1^m) \rho_{A|{z}_1^m} \otimes  \big(\rho_{B_{m+1}|z_1^m}   -   \sum_{\sigma} p_S(\sigma) \sigma_{B_{m+1}}\big) \Big\|_{\tr} \\&
 \leq \Big\|  \rho_{AB_{m+1}} - \sum_{z_1^m} p_Z(z_1^m) \rho_{A|z_1^m} \otimes \rho_{B_{m+1}|z_1^m}\Big\|_{\tr}+\Big\| \sum_{{z}_1^m}p_Z(z_1^m) \rho_{A|{z}_1^m} \otimes  \big(\rho_{B_{m+1}|z_1^m}   -  \sum_{\sigma} p_S(\sigma) \sigma_{B_{m+1}}\big)  \Big\|_{\tr} 
\\&
\leq \epsilon(d_B^r,d_A,n) + 2d_B^{2(r+1)^2}d_B^{-\frac{1}{2} (n-{(m+1)})}
\end{align*}
In total, we find that there exists an $m\in[0,n-1]$ such that
\begin{equation*} \begin{split}
&\Big\|  \rho_{AB_{m+1}} - \sum_{z_1^m} p_Z(z_1^m) \rho_{A|z_1^m} \otimes \sum_{\sigma} p_S(\sigma) \sigma_{B_{m+1}}) \Big\|_{\tr}  \\
&\leq \epsilon(d_B^r,d_A,n) + 2d_B^{2(r+1)^2}d_B^{-\frac{1}{2} (n-{(m+1)})} .
\end{split} \end{equation*}

Note that permutations are a subgroup of the stochastic orthogonal group, which means that $\rho_{AB_1^n}$ is also invariant with respect to permutations of the $B_1^n$ subsystems. Because of this, all t subsystems $B_1^n$ are separately equal to some subsystem $B$ - meaning $ B_1^n=B^{\otimes n}$. Therefore, $m$ can be chosen freely out of $[0,n-1]$, including the best possible case in terms of the bound, $m=0$, which leads to the desired statement in Theorem \ref{mainthm}.

The additional resulting statements follow directly from Theorem \ref{LinDeFinetti}, and from combining \eqref{proof-step1} and the fact that there exists a trace preserving map $\mathcal{M}'$ (see Observation \ref{CPTPmapLemma2})
transforming a state according to 
\[ \mathcal{M}': X_{B_1^{k+1}}\mapsto \sum_{z_1^k} \Tr_{Z_1^k} (\Pi_{z_1^k} X_{B_1^{k+1}}) . \]
Applied to the two states in \eqref{proof-step1}, we obtain
\begin{equation*} \begin{split}
 \mathcal{M}'(\rho_{B_1^{k+1}}) &=  \sum_{z_1^k} \Tr_{Z_1^k} (\Pi_{z_1^k} \rho_{B_1^{k+1}})
\\&
= \sum_{z_1^k} p_Z(z_1^k) \rho_{B_{k+1}|z_1^k}
\end{split} \end{equation*}
and
\begin{equation*} \begin{split} 
\mathcal{M}'(\sigma^{\otimes (k+1)}) &=  \sum_{z_1^k} \Tr_{Z_1^k} (\Pi_{z_1^k} \sigma^{\otimes (k+1)})
\\&
= \sum_{z_1^k} p_Z(z_1^k) \sigma_{B_{k+1}}.
\end{split} \end{equation*}
Inserting this into \eqref{proof-step1}, and using the fact that $\mathcal{M}'$ is CPTP, yields

\smallskip\noindent
\begin{equation} \begin{split} \label{proof-steplast}
 2d_B^{2(r+1)^2}d_B^{-\frac{1}{2} (n-(k+1))} &\geq \Big\|\rho_{B_1^{k+1}} - \sum_{\sigma_B} p_S(\sigma_B) \sigma_{B_{k+1}} \Big\|_{\tr}\\
& \geq \Big\| \mathcal{M}' \big(\rho_{B_1^{k+1}} - \sum_{\sigma_B} p_S(\sigma_B)  \sigma_{B_{k+1}} \big) \Big\|_{\tr}\\
&=\Big\| \sum_{z_1^k} p_Z(z_1^k)  \big(\rho_{B_{k+1}|z_1^k}- \sum_{\sigma_B} p_S(\sigma_B) \sigma_{B_{k+1}}\big) \Big\|_{\tr}.
 \end{split}
\end{equation}

\end{proof}

%\begin{remark}
%the other way around?
%\end{remark}

In analogy to the stabilizer de Finetti theorem, this theorem is restricted to the cases where $d_B$ is an odd prime. However, one can also consider states that are not invariant under the whole stochastic orthogonal group, but just permutations and anti-identity, which leads to the de Finetti theorem for qubits (Theorem \ref{StabDeFinettiQubit}). This can be used to obtain the following alternative stabilizer de Finetti theorem with additional linear constraints for qubits:

\begin{theorem}[Stabilizer de Finetti Theorem for Qubits with Linear Constraints]
\label{mainthm-qubit}
Let $\rho_{AB_1^n}$ be a quantum state on the Hilbert space $\mathcal{H}_A\otimes\mathcal{H}_B^{\otimes n}$ that commutes with the action of all permutations and the anti-identity on $B_1^n=B_1 B_2 \cdots B_n$. Let $\Lambda_{A\rightarrow C_A}$ and $\Gamma_{B\rightarrow C_B}$ be linear maps, and $X_{C_A}$ and $Y_{C_B}$ be operators such that the following two linear constraints hold:
\[ \Lambda_{A\rightarrow C_A}(\rho_{AB_1^n})=X_{C_A}\otimes \rho_{B_1^n} , \]
\[ \Gamma_{B_n \rightarrow C_B}(\rho_{B_1^n})=\rho_{B_1^{n-1}} \otimes Y_{C_B}. \]
Then, there exists a probability distribution $\{p_Z(z)\}_{z\in Z}$ and a probability distribution $\{p_S(\sigma_B)\}$ over the set of mixed stabilizer states on $r$ qubits such that
 \[ \Big\|\rho_{AB}-\sum_{z,\sigma_B} p_Z(z) p_{S}(\sigma_B)\rho_{A|z} \otimes  \sigma_B \Big\|_{\tr}\leq \epsilon(d_B^r,d_A,n)+\tilde{\epsilon}(d_B,r,n)\]
 where $\sigma_B$ are the mixed stabilizer states of $r$ qudits on $\mathcal{H}_B={(\mathbbm{C}^{d_B})}^{\otimes r}$, $\epsilon(d_B^r,d_A,n)$ defined in \eqref{epsilon}, 
 \begin{equation}\label{epsilon3}
 \tilde{\epsilon}(d_B,r,n):= 6 \sqrt{2}\ 2^r \sqrt{\frac{1}{n}}
 \end{equation}
and
\[\Lambda_{A\rightarrow C_A}(\rho_{A|z})=X_{C_A},\ \Gamma_{B_{n} \rightarrow C_B}(\rho_{B|x})= Y_{C_B},\]

\[
\Big\| \sum_{z} p_Z(z)  \big(\rho_{B|z}- \sum_{\sigma_B} p_S(\sigma_B) \sigma_{B}\big) \Big\|_{\tr} \leq \tilde{\epsilon}(d_B,r,n).
\]
\end{theorem}
%SATZZEICHEN

%\textcolor{green}{should this be proof or remark?}

\begin{proof}[Proof of Theorem \ref{mainthm-qubit}]
The proof for Theorem \ref{mainthm-qubit} has one additional step, but is otherwise completely analogous to the proof of Theorem \ref{mainthm} with changed error probability in instances where the stabilizer de Finetti theorem appears. The additional step is neccessary because the stabilizer de Finetti theorem for qubits only holds for $k\leq n$ being a multiple of six  (using the notation conventions from the proof of \ref{mainthm}). Therefore, to facilitate chosing $k$ equal to $m$, $m$ must be a multiple of six as well. However, since a state that is invariant under stochastic orthogonal transformations is also invariant under permutations, the de Finetti theorem with linear constraints in \ref{LinDeFinetti} holds for any $m$, so a multiple of six can be chosen. In particular, $k=m=0$ can be chosen.
\end{proof}

Both of the above stabilizer de Finetti theorems with linear constraints are symmetric in the choice of Alice and Bob (see also Remark 5.5 in \cite{BBFS18}), which means that analogous theorems can be stated approximating Alice's side instead of Bob's, which might lead to a better error bound depending on $d_A$ and $d_B$. Subsequently, this will lead to a hierarchy approximating the encoder of a QEC procedure by stabilizer state tensor powers instead of the decoder in the next section.

Furthermore, one could also simultaneously approximate Alice's \textit{and} Bob's side by stabilizer state tensor powers by expanding the above theorems, leading to a slight change in the error. For Theorem \ref{mainthm}, proof of this can be found in Appendix \ref{appendix-bothsides}.

%START
\section{An SDP Hierarchy for Maximum Channel Fidelity with Optimal Clifford Decoder}
\label{section-hierarchy}

As mentioned before, de Finetti theorems with permutation invariance can be used to design an SDP hierarchy for approximating separable states, which is a problem that appears in many quantum information theory settings, but is in general rather hard to solve. When using a de Finetti approximation, instead of enforcing separability as a numerical constraint, permutation invariance can be imposed, which is in general more tractable. As the approximation becomes better with increasing number of systems $n$, this immediately translates to a hierarchy for optimizing polynomial functions. However, although the de Finetti theorem in Theorem \ref{LinDeFinetti} from \cite{BBFS18} has (at first glance) worse convergence than best known standard de Finetti theorems, it additionally and crucially ensures that the right linear constraints are separately satisfied on Alice's and Bob's systems.
%Similarly, the newly established de Finetti theorem with stabilizer states, Theorem \ref{StabDeFinetti}, found in \cite{Gross} can be used to approximate the convex hull of tensor powers of stabilizer states. %gives an SDP-relaxation of the convex hull of symmetric powers of stabilizer states. This suggests an SDP hierarchy for optimizing polynomial functions - including the 2-norm distance, to decide membership - over the convex hull.

Instead of permutation-based de Finetti theorems, a similar hierarchy could be constructed for states with invariance under the stochastic orthogonal group based on the de Finetti theorem with stabilizer states, Theorem \ref{StabDeFinetti}, found in \cite{Gross}. This in itself is interesting for some problems, e.g. optimizing over the convex hull of stabilizers, which could be relaxed to a hierarchy of stochastic orthogonal invariant states. In previous approaches to this problem, this kind of optimization was found to be a linear problem \cite{Heinrich19} in principle (though containing some tedious enumeration), and it is therefore not clear at all that an SDP relaxation would provide an improvement. However, particularly because of the exponential convergence of the Stabilizer de Finetti Theorem, it might.
However, using Theorem \ref{StabDeFinetti} alone would not allow us to impose the important additional constraints.

In Section \ref{section-stabdefinettithm}, we have found that Theorem \ref{mainthm} implies that a state with stochastic orthogonal invariance on one side (here: Bob's) is approximately separable in the cut between $A$ and $B$ \textit{and} approximated by a convex combination of stabilizer states on Bob's side.
Similarly to \cite{BBFS18}, where the de Finetti theorem with linear constraints implies that separable states can be approximated by the hierarchy given by \eqref{hhhhhhh}, this theorem implies that the stabilizer de Finetti theorem with linear constraints can be used to obtain a hierarchy approximating separable and partly stabilizer states.

For now, we are interested in finding the best arbitrary encoder and Clifford decoder, which is given by the optimization problem \ref{Fidelity-CliffordDecoder}.
Using Choi-Jamio{\l}kowski isomorphism, this can be translated to the following optimization problem:

\begin{optimize}\label{L*}
\begin{equation*} \begin{split} 
F_C(N,d_M)=\text{maximize} \ & d_{\tilde{A}} d_B^r \Tr \Big( \big(J_{\tilde{A}B}^N \otimes \Phi_{A\tilde{B}} \big) \big(  E_{A\tilde{A}} \otimes D_{B\tilde{B}} \big) \Big)
\\
\text{subject to }&  E_{A\tilde{A}} \geq 0, \ D_{B\tilde{B}} \geq 0 \\
  & \Tr_{\tilde{A}}( E_{A\tilde{A}}) =\frac{\mathbbm{1}_A}{d_A}, \ \Tr_{\tilde{B}}( D_{B\tilde{B}})  = \frac{\mathbbm{1}_B}{d_B^r}\\
  & D_{B\tilde{B}}=  \sum_{\sigma_{B\tilde{B}}} p_S(\sigma_{B\tilde{B}}) \sigma_{B\tilde{B}} \text{ with stabilizer states $\sigma_{B\tilde{B}}$}\\%\text{ is the Choi matrix of a convex combination of tensor powers of stabilizer states}
\end{split}
\end{equation*}
\end{optimize}

%From now, the optimal value of this problem will be referred to as $F_C(N,d_M)$.

%The goal of the subsequent hierarchy is to approximate the separable state $E_{A\tilde{A}} \otimes D_{B\tilde{B}}$.
To connect this optimization problem with Theorem \ref{StabDeFinetti}, the systems $A$ and $B$ appearing in the stabilizer de Finetti theorem with linear constraints must be renamed to $A\rightarrow A\tilde{A}$ and $B\rightarrow B\tilde{B}$. Then, the linear constraints appearing in the theorem translate to
\[\Tr_{\tilde{A}}( \rho_{A\tilde{A}(B\tilde{B})_1^n} ) =\frac{\mathbbm{1}_A}{d_A}\otimes \rho_{(B\tilde{B})_1^n}\]
and
\[\Tr_{\tilde{B}_n}( \rho_{A\tilde{A}(B\tilde{B})_1^n})  = \rho_{(B\tilde{B})_1^{n-1}} \otimes \frac{\mathbbm{1}_{B_{n}}}{d_B^r}\]
which corresponds to choosing $C_A={A}$, $\Lambda_{A\tilde{A} \rightarrow {A}} =\Tr_{\tilde{A}}$ and $X_{A}=\frac{\mathbbm{1}_A}{d_A}$, as well as $C_B={B}$, $\Gamma_{B\tilde{B} \rightarrow {B}} =\Tr_{\tilde{B}}$ and $Y_{B}=\frac{\mathbbm{1}_B}{d_B^r}$ and inserting these choices in the constraints of Theorem \ref{mainthm}.

Now, Theorem \ref{mainthm} implies that a state $\rho_{A\tilde{A} (B\tilde{B})_1^n}$ that is invariant under the action of $O_n$ on Bob's side can be used to approximate a separable state where Bob's side of the state (in the above case $D_{B\tilde{B}}$) is a mixture of stabilizer states. Recast as an optimization problem, level $n$ of the hierarchy is given by:
\begin{optimize}\label{L_n}
\begin{equation*} \begin{split}
F_C^{(n)}(N,d_M)=\text{maximize} \ & d_{\tilde{A}} d_B^r \Tr \Big( \big(J_{\tilde{A}B}^N \otimes \Phi_{A\tilde{B}} \big) \big(  \rho_{A\tilde{A}B\tilde{B}} \big) \Big)\\
\text{subject to }&  \rho_{A\tilde{A}(B\tilde{B})_1^n}  \geq 0, \ \Tr(\rho_{A\tilde{A}(B\tilde{B})_1^n})=1 \\
& \rho_{A\tilde{A}(B\tilde{B})_1^n} \text{ is invariant under the action of } O_n \\
  & \Tr_{\tilde{A}}( \rho_{A\tilde{A}(B\tilde{B})_1^n} ) =\frac{\mathbbm{1}_A}{d_A}\otimes \rho_{(B\tilde{B})_1^n}\\&
\Tr_{\tilde{B}_n}( \rho_{A\tilde{A}(B\tilde{B})_1^n})  = \rho_{A\tilde{A}(B\tilde{B})_1^{n-1}} \otimes \frac{\mathbbm{1}_{B_{n}}}{d_B^r} \\
\end{split}
\end{equation*}
\end{optimize}

The optimal value $F_C(N,d_M)$ of problem \ref{L*} (which is ultimately what we want to know) and the optimal value $F_C^{(n)}(N,d_M)$ of the SDP \ref{L_n} (which is more easily computable) correspond directly to the states inserted in the trace distance in the stabilizer de Finetti theorem with linear constraints, and thereby, their difference corresponds to the error bound in the theorem. With increasing level $n$, this error decreases, and the values $F_C^{(n)}(N,d_M)$ approach the value $F_C(N,d_M)$. The convergence of the SDP hierarchy, meaning $\lim_{n\rightarrow \infty} F_C^{(n)}(N,d_M) =F_C^{(n)}(N,d_M)$, thus follows from Theorem \ref{mainthm}.

\begin{theorem}[Convergence of the hierarchy] \label{convergence}
Considering the SDPs in \eqref{L_n} with optimal value $F_C^{(n)}(N,d_M)$ for some $n\geq 0$, we find that
\[ F_C^{(n+1)}(N,d_M) \leq F_C^{(n)}(N,d_M) \]
\[\lim_{n\rightarrow \infty} F_C^{(n)}(N,d_M)=F_C(N,d_M)\]
where $F_C(N,d_M)$ is the optimal value to the optimization problem \eqref{L*}.
\end{theorem}

\begin{proof}[Proof of Theorem \ref{convergence}]
Because of Theorem 4.1 in \cite{BBFS18}, we know that convergence holds for the hierarchy with permutation invariance. Replacing permutation invariance by stochastic orthogonal invariance does not change their proof. Because the constraints become more powerful, $ F_C^{(n+1)}(N,d_M) \leq F_C^{(n)}(N,d_M)$. Given a feasible solution $E_{A\tilde{A}}\otimes D_{B\tilde{B}}$ to $F_C(N,d_M)$, a feasible solution for $F_C^{(n)}(N,d_M)$ can easily be constructed via $E_{A\tilde{A}}\otimes D_{B\tilde{B}}^{\otimes n}$. Given a feasible solution to $F_C^{(n)}(N,d_M)$, this is only approximately equal to a convex combination of feasible solutions to $F_C(N,d_M)$, where the approximation error is given in form of the bound of the de Finetti theorem, which improves with increasing $n$, and tends to zero in the limit $n\rightarrow \infty$. This can easily be seen in the following calculation:

\begin{equation*}
\begin{split}
 &F_C^{(n)}(N,d_M)-F_C(N,d_M)\\&= d_{\tilde{A}} d_B^r \Tr \Big( \big(J_{\tilde{A}B}^N \otimes \Phi_{A\tilde{B}} \big) \big(  \rho_{A\tilde{A}B\tilde{B}} \big) \Big)- d_{\tilde{A}} d_B^r \Tr \Big( \big(J_{\tilde{A}B}^N \otimes \Phi_{A\tilde{B}} \big) \big(  E_{A\tilde{A}} \otimes D_{B\tilde{B}} \big) \Big)\\
 &= d_{\tilde{A}} d_B^r \Tr \Big( \big(J_{\tilde{A}B}^N \otimes \Phi_{A\tilde{B}} \big) \big(  \rho_{A\tilde{A}B\tilde{B}} -E_{A\tilde{A}} \otimes D_{B\tilde{B}}  \big) \Big)\\
 &= d_{\tilde{A}} d_B^r \Tr \Big( \big(J_{\tilde{A}B}^N \otimes \Phi_{A\tilde{B}} \big) \big( \rho_{A\tilde{A}B\tilde{B}}- E_{A\tilde{A}} \otimes \sum_{\sigma_{B\tilde{B}}} p_S(\sigma_{B\tilde{B}}) \sigma_{B\tilde{B}}  \big) \Big)\\
 \end{split}
 \end{equation*}

 $F_C(N,d_M)$ can be rewritten in terms of a convex combination, which corresponds to adding classical shared randomness assistance, which does not affect the optimal value of the fidelity \cite{BBFS18}. Then, we find
 
 \begin{equation*}
\begin{split}
  F_C^{(n)}(N,d_M)-F_C(N,d_M)&= d_{\tilde{A}} d_B^r \Tr \Big( \big(J_{\tilde{A}B}^N \otimes \Phi_{A\tilde{B}} \big) \big( \rho_{A\tilde{A}B\tilde{B}}-\sum_{z,\sigma_{B\tilde{B}}} p_Z(z) E_{A\tilde{A}|z} \otimes  p_S(\sigma_{B\tilde{B}}) \sigma_{B\tilde{B}}\big)   \Big)\\
 &= d_{\tilde{A}} d_B^r \big(\epsilon(d_B^r,d_A,n)+\bar{\epsilon}(d_B,r,n)\big)  \Tr \Big( J_{\tilde{A}B}^N \otimes \Phi_{A\tilde{B}}\Big)\\
\end{split}
 \end{equation*}
% \[ \Big\|\rho_{AB}-\sum_{z,\sigma_B} p_Z(z) p_{S}(\sigma_B)\rho_{A|z} \otimes  \sigma_B \Big\|_{\tr}\leq \epsilon(d_B^r,d_A,n)+\tilde{\epsilon}(d_B,r,n)\]
with $\epsilon(d_B^r,d_A,n)$ given in \eqref{epsilon} and $\bar{\epsilon}(d_B,r,n)$ given in \eqref{epsilon2}.

Then, because $\lim_{n\rightarrow \infty}\epsilon(d_B^r,d_A,n)=0$ and $\lim_{n\rightarrow \infty} \tilde{\epsilon}(d_B,r,n)=0$, the above tends to zero, and consequently the objective value of the optimization problem \ref{L_n} converges to the optimal value of the optimization problem \ref{L*}.%However, in contrast to the permutation-based version, the linear constraints are not fulfilled exactly, but approximately. %However, one must pay attention to the additional convergence of the linear constraints. In the limit of $n\rightarrow \infty$, the state converges towards a separable and partly stabilizer state. 
\end{proof}

\begin{remark}
In contrast to the permutation-based version, the linear constraints are not fulfilled exactly, but approximately. If separability is achieved after a low number of steps, stabilizer-ness cannot be guaranteed. However, since the convergence towards stabilizer states is exponential, it is considerably faster than the convergence towards separability.
\end{remark}

%\newpage

\section{Remark on numerical tests}
\label{section-numerics}

%With a convergent hierarchy in place, the next natural step would be 
%With a convergent hierarchy established, this directly leads to a numerical application of it. 
%The ultimate goal of establishing a convergent hierarchy is to use it in numerical computations.
One central goal of establishing a convergent hierarchy like the one in Section \ref{section-hierarchy} is to implement and apply it to some numerical computations. Most naturally, the simila-rities to the convergent hierarchy in \cite{BBFS18} could be exploited to obtain a direct comparison between the maximum channel fidelity for arbitrary encoder and decoder, $F(N,d_M)$ and the maximum channel fidelity with Clifford operations $F_C(N,d_M)$. However, due to the high number of parameters and time constraints, the numerical implementation was not completed within this thesis.
%Unfortunately, the numerical implementation of the SDP hierarchy was not completed due to time constraints and numerical difficulties.% However, some work has already been done.

As a preliminary exercise, we were interested in finding the optimal Clifford decoder for the 3-qubit bitflip code. This means that separability between encoder and decoder is given, and we are not making use of the hierarchy above, but a simplified version of it to reduce the number of parameters. Even for two copies (i.e. no orthogonal symmetry, just swapping the subsystems), a home computer did not have sufficient memory for this task. However, the problem was not yet optimized by using symmetry to reduce the number of variables, which may make the problem more tractable.

For studies making use of the hierarchy, it must be noted that only some particular combinations of the number of subsystems $n$, the local dimension $d$ and the number of qudits $r$ that the Clifford unitaries act on could be of interest. In the case of qubits, i.e. $d=2$, the stabilizer de Finetti theorem for qubits only holds for $n\geq 6$. %, which is already a rather large number of parameters in the optimization.
For qudits and combinations $d=3$ and $n=2,3$, the stochastic orthogonal group is equal to the permutation group. Consequently, the first potentially interesting case may appear for $d=3$, $n=4$.

All of these combinations already imply a large number of parameters. However, the number of parameters can be somewhat reduced because of the symmetry of the problem. Permutation invariant states can be defined in terms of fewer parameters than regular states (using Clebsch-Gordan coefficients). In addition, depending on the problem, the symmetry of the channel can be used to simplify the optimization as well (as noted in \cite{BBFS18}).% So, in the extremely simple case of only bit flip errors and repetition code, the problem can be further simplified by the covariance of the channel and code to phase flips and bit flips. That might make it tractable to do the Clifford decoder question. That's how Mario et al can go to such high blocksizes for the depolarizing channel. 

As with many SDP problems, something may also be gained from looking at the corresponding dual problem. It may reduce the number of variables (which is unlikely in our case), or show some structure that allows for a good guess at a feasible solution. More importantly, it could be argued that the dual is the more natural optimization for the original problem because the dual directly gives an upper bound on the optimal value for any feasible solution, thereby giving an upper bound to the upper bound on the maximum channel fidelity achieved at this level. The primal problem may only give an upper bound of the maximum channel fidelity if it attains the optimal value, as it is a maximization rather than a minimization. In other words, any feasible solution of the dual provides an upper bound on $F_C^{(n)}(N,d_M)$, while $F_C^{(n)}(N,d_M)$ itself is an upper bound on $F_C(N,d_M)$. 

To obtain the dual to optimization problem \ref{L_n}, it is helpful to rephrase the problem in terms of maps:

\begin{optimize}\label{primal-L_n}
\begin{equation*} \begin{split}
F_C^{(n)}(N,d_M)=\text{maximize} \ & d_{\tilde{A}} d_B \Tr \Big( \big(J_{\tilde{A}B}^N \otimes \Phi_{A\tilde{B}} \big) \big(  \rho_{A\tilde{A}B\tilde{B}} \big) \Big)\\
\text{subject to }&  \rho_{A\tilde{A}(B\tilde{B})_1^n}  \geq 0, \ \Tr(\rho_{A\tilde{A}(B\tilde{B})_1^n})=1 \\
& \mathcal{U}_O(\rho_{A\tilde{A}(B\tilde{B})_1^n})=0 \ \forall O \in \mathcal{O}_n \\
  &\mathcal{C}_A(\rho_{A(B\tilde{B})_1^n})=0 \\
& \mathcal{C}_{B_n}(\rho_{A\tilde{A}(B\tilde{B})_1^{n-1} B_n})=0 \\%&
%\Tr_{\tilde{B}_n}( \rho_{A\tilde{A}(B\tilde{B})_1^n})  = \rho_{A\tilde{A}(B\tilde{B})_1^{n-1}} \otimes \frac{\mathbbm{1}_{B_{n}}}{d_B} \\
\end{split}
\end{equation*}
with the following maps:
\[\mathcal{U}_O(x)=OxO^{\dagger}-x\]
\[\mathcal{C}_A(x)=\big( \mathcal{T}_{HW}^{A}\otimes \mathbbm{1}_{(B\tilde{B})_1^n}\big)x \big( (\mathcal{T}_{HW}^{A})^{\dagger} \otimes \mathbbm{1}_{(B\tilde{B})_1^n}\big)-x\]
\[\mathcal{C}_{B_n}(x)=\big(\mathbbm{1}_{A\tilde{A}(B\tilde{B})_1^{n-1}}\otimes \mathcal{T}_{HW}^{B_n}\big) x\big( \mathbbm{1}_{A\tilde{A}(B\tilde{B})_1^{n-1}}\otimes (\mathcal{T}_{HW}^{B_n})^{\dagger}\big) -x\]
with $\mathcal{T}_{HW}^{A}$ being the twirling map from \eqref{twirlingeq} applied to subsystem $A$, and $\mathcal{T}_{HW}^{B_n}$ the twirling of subsystem $B_n$.
\end{optimize}

Then, the dual of optimization problem \ref{primal-L_n} (and thereby \ref{L_n}) is given by the following optimization problem:

\begin{optimize}\label{dual-L_n}
\begin{equation*} \begin{split}
F_C^{(n)}(N,d_M)=\text{minimize} \ & \lambda \\
\text{subject to }&  \big(J_{\tilde{A}B_1}^N \otimes \Phi_{A\tilde{B}_1} \big) \otimes \mathbbm{1}_{(B\tilde{B})_2^n} +\sum_O \mathcal{U}_O(P_{A\tilde{A}(B\tilde{B})_1^n}^O) \\&\ \ \ \ \ +\mathcal{C}_A(Y_{A(B\tilde{B})_1^n})\otimes \mathbbm{1}_{\tilde{A}}+\mathcal{C}_{B_n} (Z_{A\tilde{A}(B\tilde{B})_1^{n-1} B})\otimes \mathbbm{1}_{\tilde{B}_{n}}\leq  \lambda \mathbbm{1}_{A\tilde{A}(B\tilde{B})_1^n} \\
\end{split}
\end{equation*}
\end{optimize}

As to be expected, while the primal problem contains one SDP variable ($\rho_{A\tilde{A}(B\tilde{B})_1^n}$) and many constraints, the dual contains many SDP variables ($\lambda$, $P_{A\tilde{A}(B\tilde{B})_1^n}^O$, $Y_{A(B\tilde{B})_1^n}$, $Z_{A\tilde{A}(B\tilde{B})_1^{n-1} B}$) and only one constraint. If the constraints in the primal are relaxed to inequality constraints instead of equalities, the dual contains additional constraints on the positivity of $P_{A\tilde{A}(B\tilde{B})_1^n}^O$, $Y_{A(B\tilde{B})_1^n}$ and $Z_{A\tilde{A}(B\tilde{B})_1^{n-1} B}$.

%\textcolor{blue}{What conclusion can we get from this? This looks like an infinity norm...? Also, were there other numerical remarks that I am forgetting?}

While numerical tests are possible in principle, and something can be gained from exploiting the symmetry of the problem to reduce the number of parameters, any reasonable test would still require a computer with the ability to keep track of large numbers of parameters. First attempts were already beyond the capabilities of a home computer. None-theless, with additional time and effort, numerical results which allow direct comparison to \cite{BBFS18} are likely attainable.% and some further reduction of the number of parameters by exploiting the symmetry. First attempts were already beyond the capabilities of a home computer. %; and were coded in the programming language Julia, for which the support on ETH's high performance computer Euler is sparse.

\newpage
%START
\chapter{Summary and Outlook}

%Different de Finetti theorems
In accordance with the multitudes of applications relying on the original de Finetti theorem with permutation invariance, the stabilizer de Finetti theorem in \ref{StabDeFinetti} also has various interesting applications. Although the results are similar to some degree, they are nonetheless interesting, for two reasons: on the one hand, the stabilizer de Finetti theorem has exponential convergence, and on the other hand, it leads to an approximation by stabilizer states instead of arbitrary quantum states. %has various applications with similar but interesting results.
\\
\\
The convergence of the theorem makes it particularly interesting for QKD, where a large number of subsystems must usually be traced out to lead to any reasonable error bound on the communication - a notorious problem for real-world applications. Using the postselection technique, it can be shown that the security of general attacks can be inferred from the security of collective attacks, with a multiplicative factor on the error which depends on the dimension of a permutation-invariant subspace. Here, we have shown that the postselection technique can be generalized to arbitrary symmetry groups, and a bound on the diamond norm of the difference of two CPTP maps can be obtained from it, where this bound depends chiefly on the dimension of the invariant subspace. Then, a computation of the dimension of an invariant subspace associated with the symmetry appearing in the stabilizer de Finetti theorem - stochastic orthogonal symmetry - was performed, finding that the resulting multiplicative factor for QKD applications is significantly smaller and independent of the number of subsystems $n$.

The natural next step would be using this postselection theorem to obtain key rates and achievable key lengths for known protocols. The permutation-symmetry based postselection theorem has already been applied to determine key rates for a variety of protocols, including qubits \cite{Sheridan10_1} and qudits \cite{Sheridan10_2} in the BB84 protocol, 6-state protocol \cite{Mertz13} and $N$-party 6-state protocol \cite{Grasseli18}. However, finding a protocol with the desired symmetry is not trivial. For BB84, the technique of using uncertainty relations for QKD \cite{Koashi06} outperforms most de Finetti type considerations, and it is unlikely that our result will constitute an improvement. The 6-state protocol, on the other hand, is a more promising candidate. If the 6-state protocol has stochastic orthogonal invariance, the corresponding key rates can be determined by replacing one summand in the key lengths found in \cite{Mertz13} and \cite{Grasseli18} - precisely the summand corresponding to the difference between collective and general attack key lengths (corresponding to the dimension of the permutation-invariant subspace). Furthermore, as another potential protocol, one might use the protocol appearing in the context of continuous variable de Finetti theorems with orthogonal symmetry \cite{LKGC09} and restrict it to finite dimensions, but the practicality of such an endeavour is unclear.
\\
\\
Because it leads to an approximation by tensor powers of stabilizer states, the stabilizer de Finetti theorem is also interesting for the study of applications tied to stabilizer codes and Clifford operations. Here, we find that it can be employed for benchmarking the success of a QEC procedure in terms of its maximum channel fidelity, where either the encoder or decoder (or both) are Clifford operations. Computing maximum channel fidelity contains a bilinear optimization for finding the optimal encoder and decoder for a given noisy channel, which is in general not directly possible. However, it can be approximated by a converging hierarchy of SDP relaxations, which we also find to be possible for maximum channel fidelity with restriction to Clifford decoders (or encoders, or both).

This directly leads to the next, most logical step: performing numerical tests for some examples. As a preliminary, one could look at a simplified problem of finding the optimal Clifford decoder for 3-qubit repitition code with bitflips, before moving on to more complicated examples which can be compared to the numerical results in \cite{BBFS18}. Although we did take first steps in this direction, full fledged numerical computations are outside the scope of this work.
\\
\\
While this thesis focuses on two aspects of applications emerging from the stabilizer de Finetti theorem, it could also be interesting to investigate many others. In particular, it could be beneficial to look at other results relying on the permutation-based de Finetti theorem and postsselection technique, including applications pertaining to tomography \cite{ChrRen12} and Shannon reverse coding \cite{Berta11}.

%\textcolor{blue}{GOOD LAST WORDS?}

%SATZZEICHEN
\newpage

\addcontentsline{toc}{chapter}{References}
\bibliographystyle{linksen}
\bibliography{byb2}
\newpage

\appendix
%\appendixpage
\addappheadtotoc
%\chapter{Overview of notation conventions}
%\label{notation}

%\newpage
\chapter{Proof of the Bound on the Diamond Norm }

\section{Resolution of Identity}
\label{appendix-resolutionofidentity}

The postselection technique can only be applied for a symmetry group $\mathcal{S}$ if the following condition, termed resolution of identity, holds: 

 \begin{equation}
\operatorname{Cond0}=\Big\{ \exists d_{\mathcal{S}} (\cdot)\  \operatorname{s.t.} \ \int \rho^{\otimes n} {d}_{\mathcal{S}} (\rho)= \frac{1}{g_{n,d}} \mathbbm{1}_{  {N} } \Big\} 
\end{equation}

Here, we state two necessary, but not sufficient conditions for Cond0, and thereby for resolution of identity.
Thus, if these conditions are met, it is sufficient to consider the particular state $\tau_{T}$ when computing the diamond norm, instead of performing an optimization over a large number of states:

\begin{gather*}
\operatorname{Cond1}= \Big\{ \exists \ket{\phi^{(k)}} \in \mathcal{H}\otimes \mathcal{H} \ \operatorname{s.t.} \ 
\mathcal{N} \  \operatorname{is\ spanned\ by}\ \ket{\phi^{(k)}}^{\otimes n} \Big\}
 %s\otimes \overline{s} \ket{\phi}^{\otimes n} =\ket{\phi}^{\otimes n} \ \operatorname{s.t.}\  (\mathcal{H}^{\otimes n} \otimes \mathcal{H}^{\otimes n})^{(s\otimes\overline{s})} = \operatorname{span} (\ket{\phi}^{\otimes n}) \Big\} 
\\
{\Longleftrightarrow} \operatorname{Cond2}=
\Big\{ \exists X^{(k)} \in L(\mathcal{H},\mathcal{H}) \  \operatorname{s.t.\ the\ commutant\ of} \ \mathcal{S}\  \operatorname{is\ spanned\ by} \ (X^{(k)})^{\otimes n}% \ \operatorname{with} \ X^{\otimes n}=s X^{\otimes n} s^{\dagger} 
\Big\} 
\end{gather*}
where $ L(\mathcal{H},\mathcal{H})$ denotes the space of linear maps $X^{(k)}:\mathcal{H}\mapsto\mathcal{H}$, and $\mathcal{N}=(\mathcal{H}^{\otimes n} \otimes \mathcal{H}^{\otimes n})^{(s\otimes\overline{s})}$.

%Instead of a complicated statement about an integration measure, the alternative formulations (especially Prop2) have the advantage of being conditions \textcolor{green}{in linear algebra that can be more easily verified for a specific symmetry group? that can be understood more easily? that are less specific to our particular problem?}%The advantage of the last formulation is %that it is a condition in the context of linear algebra, and therefore more easily understood without theory about integration measures.
These alternative formulations provide an advantage over the original condition Cond0 because they are conditions in linear algebra that are less specific to our particular problem. %\textcolor{green}{MAYBE SOMETHING IS MISSING HERE}

Firstly, we show the equivalence of the two conditions Cond1 and Cond2 by finding a one-to-one map between the two sets they describe.

\begin{proof}[Proof: Cond1 $\Longleftrightarrow$ Cond2]
Firstly, note the definition of the space $\mathcal{N}$:
\begin{equation*}
\mathcal{N}=(\mathcal{H}^{\otimes n} \otimes \mathcal{H}^{\otimes n})^{(s\otimes\overline{s})} =\{
\ket{\Phi}\in(\mathcal{H}^{\otimes n} \otimes \mathcal{H}^{\otimes n}) | \ket{\Phi}=s\otimes\overline{s} \ket{\Phi} \ \forall s\in{\mathcal{S}} 
\}
\end{equation*}
and the definition of the commutant:
\begin{equation*}
S' =\{
X \in L(\mathcal{H}^{\otimes n}, \mathcal{H}^{\otimes n}) | X=sXs^{\dagger} \ \forall s\in{\mathcal{S}} 
\}
\end{equation*}
There is a one-to-one map between these two spaces such that the $\ket{\Phi}$ and the $X$ can be transformed into one another.

Let $\{\ket{i}\}$ be a basis of $\mathcal{H} $. Then, $\{\ket{i} \otimes \ket{j} \}$ is a basis of $\mathcal{H}\otimes \mathcal{H} $. The map we propose transforms vectors on two copies of $\mathcal{H}$ to linear maps between the two spaces: $\mathcal{H} \otimes \mathcal{ H} \rightarrow L(\mathcal{H} ,\mathcal{H} )$, and acts in the following way: $\ket{i}\otimes \ket{j} \mapsto \ket{i} \bra{j}$.

First, let us expand $\ket{\Phi}$ in this basis and perform the mapping.

\begin{equation*}
\ket{\Phi}=
\sum_{\substack{i_1,...,i_n \\  j_1,...,j_n}}
 \Phi_{i_1,...,i_n , j_1,...,j_n} \ket{i_1,...,i_n } \otimes \ket{ j_1,...,j_n }
\mapsto
\sum_{\substack{i_1,...,i_n \\  j_1,...,j_n}}  \Phi_{i_1,...,i_n , j_1,...,j_n} \ket{i_1,...,i_n } \bra{ j_1,...,j_n } \equiv X
\end{equation*}

Now we show that the restrictions for $\ket{\Phi}\in\mathcal{N} $ and $X\in S'$ translate into one another.

\begin{align*}
\ket{\Phi}&=
(s\otimes\overline{s})\ket{\Phi}
=
(s\otimes\overline{s})\Big(\sum_{\substack{i_1,...,i_n \\  j_1,...,j_n}}
 \Phi_{i_1,...,i_n , j_1,...,j_n} \ket{i_1,...,i_n }\otimes \ket{ j_1,...,j_n }\Big)
\\&
=
\sum_{\substack{i_1,...,i_n \\  j_1,...,j_n}} \Phi_{i_1,...,i_n , j_1,...,j_n} \big(s \ket{i_1,...,i_n }\big) \otimes \Big(\overline{s}  \ket{ j_1,...,j_n }\Big)
\\&
=
\sum_{\substack{i_1,...,i_n \\  j_1,...,j_n}} \Phi_{i_1,...,i_n , j_1,...,j_n} \big(s \ket{i_1,...,i_n }\big) \otimes
\\ & \ \ \ \ \ \ \ \ \ \  \ \ \ \ \ \ \ \ \ \ 
 \bigg( \sum_{\substack{m_1,...,m_n \\  n_1,...,n_n}}\Big(\bra{m_1,...,m_n} \overline{s} \ket{ n_1,...,n_n}\Big) \ket{ m_1,...,m_n }\bra{ n_1,...,n_n }\bigg)  \ket{ j_1,...,j_n }
\\&
=
\sum_{\substack{i_1,...,i_n \\  j_1,...,j_n\\m_1,...,m_n\\n_1,...,n_n}} \Phi_{i_1,...,i_n , j_1,...,j_n} \big(s \ket{i_1,...,i_n }\big) \otimes \\
& \ \ \ \ \ \ \ \ \ \  \ \ \ \ \ \ \ \ \ \ 
 \big(\bra{m_1,...,m_n} \overline{s} \ket{ n_1,...,n_n}\big) \ket{ m_1,...,m_n } \big(\bra{ n_1,...,n_n } \ket{ j_1,...,j_n }\big)
\\&
=
\sum_{\substack{i_1,...,i_n \\  j_1,...,j_n\\m_1,...,m_n}} \Phi_{i_1,...,i_n , j_1,...,j_n} \big(s \ket{i_1,...,i_n }\big) \otimes \big(\bra{m_1,...,m_n} \overline{s} \ket{ j_1,...,j_n}\big) \ket{ m_1,...,m_n }
\\&
\mapsto
\sum_{\substack{i_1,...,i_n \\  j_1,...,j_n\\m_1,...,m_n}} \Phi_{i_1,...,i_n , j_1,...,j_n} \big(s \ket{i_1,...,i_n }\big)  \big(\bra{m_1,...,m_n} \overline{s} \ket{ j_1,...,j_n}\big)   \bra{ m_1,...,m_n }
\\&
=
\sum_{\substack{i_1,...,i_n \\  j_1,...,j_n\\m_1,...,m_n}} \Phi_{i_1,...,i_n , j_1,...,j_n} \big(s \ket{i_1,...,i_n }\big) \big(\bra{j_1,...,j_n} (\overline{s})^T \ket{ m_1,...,m_n}\big)   \bra{ m_1,...,m_n }
\\&
=
\sum_{\substack{i_1,...,i_n \\  j_1,...,j_n}} \Phi_{i_1,...,i_n , j_1,...,j_n} \big(s \ket{i_1,...,i_n }\big)  \Big(\bra{j_1,...,j_n} \overline{s}^{\dagger} \big( \sum_{m_1,...,m_n} \ket{ m_1,...,m_n}\big)   \bra{ m_1,...,m_n }\Big)
\\&
=
s\Big(\sum_{\substack{i_1,...,i_n \\  j_1,...,j_n}} \Phi_{i_1,...,i_n , j_1,...,j_n} \ket{i_1,...,i_n } \bra{ j_1,...,j_n }\Big)s^{\dagger} =sXs^{\dagger}=X
\end{align*}
Each $\ket{\Phi}\in \mathcal{N} $ therefore defines an $X\in S'$, and since the proposed map is invertible and one-to-one, each $X\in S'$ defines one $\ket{\Phi}\in  \mathcal{N}$.

Now, Cond1 implies
\begin{equation*}
\mathcal{N} =\{\ket{\Phi}\in(\mathcal{H}^{\otimes n} \otimes \mathcal{H}^{\otimes n}) | \ket{\Phi}=\sum_k c_k \ket{\phi^{(k)}}^{\otimes n} \}.
\end{equation*}

Let the map be applied to each summand of $\ket{\Phi}$ (i.e. for each $k$):

\begin{equation*}
\ket{\phi^{(k)}}^{\otimes n} =... \mapsto .... =(X^{(k)})^{\otimes n}
\end{equation*}

Therefore, for each possible $\ket{\Phi}\in \mathcal{N}$, we obtain \[\ket{\Phi}=\sum_k c_k \ket{\phi^{(k)}}^{\otimes n} \mapsto\sum_k c_k (X^{(k)})^{\otimes n}=X,\]
defining all possible $X\in S' $. In consequence, the commutant is spanned by $(X^{(k)})^{\otimes n}$, proving Cond1 $\Rightarrow$ Cond2.

Since the map is one-to-one and invertible, the other direction (Cond2 $\Rightarrow$ Cond1) is analogous with the inverse map $X\mapsto \ket{\Phi}$.

\end{proof}

Having proven that the two proposed conditions are equivalent, we move on to show that Cond1 (and thus the equivalent Cond2) is a necessary condition for the original statement Cond0. By contraposition, this can be phrased as: Whenever Cond0 is true, Cond1 is true, which is the statement that will be proven in the following.

\begin{proof}[Proof: Cond0 $\Rightarrow$ Cond1]
We assume that a measure $ d_{\mathcal{S}} (\cdot)$ exists,
and $\int \rho^{\otimes n} {d}_{\mathcal{S}} (\rho)= \frac{1}{g_{n,d}} \mathbbm{1}_{N} $ holds.

Since we are only concerned with finite spaces, the integral can be written as a sum: \[\int \rho^{\otimes n} {d}_{\mathcal{S}} (\rho)\rightarrow \sum_{k} ({\rho^{(k)}})^{\otimes n} {d}_{\mathcal{S}} ({\rho^{(k)}})\]
As this is equal to the identity on $\mathcal{N}$ by assumption, this implies that the space $\mathcal{N}$ is spanned by the $({\rho^{(k)}})^{\otimes n}$. Because each ${\rho^{(k)}}$ is pure, there exists a vector $\ket{\phi^{(k)}}$ such that ${\rho^{(k)}}=\ket{\phi^{(k)}} \bra{\phi^{(k)}}$, and therefore $ ({\rho^{(k)}})^{\otimes n}= (\ket{\phi^{(k)}} \bra{\phi^{(k)}})^{\otimes n}$. Thus, $\exists \ket{\phi^{(k)}}$ such that $\mathcal{N}$ is spanned by $ \ket{\phi^{(k)}} ^{\otimes n}$ (Cond1).

\end{proof}

Thus Cond1 (and the equivalent Cond2) are necessary, but not sufficient conditions of the original statement Cond0, which in turn is a necessary and sufficient condition for the postselection technique to be applicable for a symmetry $\mathcal{S}$. 

\section{Preliminary Lemmata}
%jump
\label{appendix-PostselecLemma}

The postselection theorem \ref{ImportantThm} states that the diamond norm of a difference of $\mathcal{S}$-invariant CPTP maps, which entails an optimization problem over all states within a certain space, can be bound using one specific state, which is a purification of the de Finetti state given in \eqref{deFinettistate}.
To simplify the proof of this bound on the diamond norm in Theorem \ref{ImportantThm}, three lemmata are useful.
First, Lemma \ref{purif} establishes a connection between states on $\mathcal{H}^{\otimes n}$ that are invariant under the symmetry $\mathcal{S}$ and a state with support on the invariant subspace $(\mathcal{H}^{\otimes n} \otimes \mathcal{H}^{\otimes n})^{(s\otimes \overline{s})}$. This lemma will appear in the proof of Lemma \ref{lemmasym}, where it is shown to be sufficient to consider states with support on the invariant subspace $\mathcal{N}=(\mathcal{H}^{\otimes n} \otimes \mathcal{H}^{\otimes n})^{(s\otimes \overline{s})}$ for computing the diamond norm of an invariant map. Thirdly, in the context of the diamond norm, we state and prove Lemma \ref{lemma2} which connects states with support on the invariant subspace $\mathcal{N}$ and the de Finetti state from (\ref{deFinettistate}).

\begin{lemma}%[See \cite{RennerPhd}, Lemma ]
\label{purif}
Any state $\rho_{T}$ on $\mathcal{H}_T=\mathcal{H}^{\otimes n}$ that is invariant under a symmetry group $\mathcal{S}$ admits a purification $\rho_{TE}$ with support on $ \mathcal{N}=(\mathcal{H}^{\otimes n} \otimes \mathcal{H}^{\otimes n})^{(s\otimes \overline{s})}\subseteq \mathcal{H}^{\otimes n} \otimes \mathcal{H}^{\otimes n}=\mathcal{H}_T\otimes \mathcal{H}_E$.
\end{lemma}

\begin{proof}[Proof of Lemma \ref{purif}]
The state $\rho_{T}$ is $\mathcal{S}$-invariant, which means $\rho_{T}=s\rho_{T}s^{\dagger}$ $\forall s \in \mathcal{S}$. 
Let $\Lambda$ be the set of eigenvalues of $\rho_{T}$, and let $\{ \ket{x} \}_{x\in\mathcal{X}}$ be its eigenbasis. For each eigenvalue $\lambda\in\Lambda$, consider the associated eigenvectors $\{ \ket{x_{\lambda}} \}_{x_{\lambda}\in\mathcal{X}_{\lambda}}$ (with $\mathcal{X}_{\lambda} \subseteq \mathcal{X}$). Then, for each $\lambda\in \Lambda$, $\rho_{T} \ket{x_{\lambda}} =\lambda \ket{x_{\lambda}}$ $\forall x_{\lambda}\in\mathcal{X}_{\lambda}$. With these eigenvalues and the eigenbasis, we can write the state as follows:
\begin{equation*}
\rho_{T}=\sum_{\lambda\in \Lambda} \sum_{x_{\lambda}\in\mathcal{X}_{\lambda}} \lambda \ket{x_{\lambda}} \bra{x_{\lambda}}
\end{equation*}

Construct the following state for each $\lambda\in \Lambda$:
\begin{equation*}
\ket{\Phi_{\lambda}} = \sum_{x_{\lambda}\in\mathcal{X}_{\lambda}} \ket{x_{\lambda}} \otimes %\overline
{\ket{x_{\lambda}}}.
\end{equation*}
%where $\overline{\ket{x}}$ is the complex conjugate with respect to some basis of $\mathcal{H}^{\otimes n}$.

Then, we claim that $\ket{\Phi}\bra{\Phi} \in\mathfrak{S}( \mathcal{H}_T \otimes \mathcal{H}_E)$ with $\mathcal{H}_E=\mathcal{H}_T$ and with

\begin{equation*}
\ket{\Phi} = \sum_{\lambda\in \Lambda}  \sqrt{\lambda} \ket{\Phi_{\lambda}}
\end{equation*}
%We claim that $\ket{\Phi}\bra{\Phi} \in \mathcal{H}^{\otimes n} \otimes \mathcal{H}^{\otimes n}$
is a purification of $\rho_{T} \in \mathfrak{S}( \mathcal{H}_T)$ with the desired characteristics.

To test this claim, we firstly show that $\ket{\Phi}\bra{\Phi}$ is indeed a purification of $\rho_{T}$, which can be shown by a quick calculation:
\begin{equation*}
\trace_{E} \ket{\Phi}\bra{\Phi} =\trace_E ( \sum_{\lambda,\lambda'\in \Lambda} \sqrt{\lambda} \sqrt{\lambda'} \sum_{\substack{x_{\lambda}\in\mathcal{X}_{\lambda} \\ x_{\lambda'}\in\mathcal{X}_{\lambda'}}} (\ket{x_{\lambda}} \otimes %\overline
{\ket{x_{\lambda}}}) (\bra{x_{\lambda'}} \otimes %\overline
{\bra{x_{\lambda'}}}))
\end{equation*}
\begin{equation*}
 =\sum_{\lambda,\lambda'\in \Lambda} \sum_{\substack{x_{\lambda}\in\mathcal{X}_{\lambda} \\ x_{\lambda'}\in\mathcal{X}_{\lambda'}}}
 \sqrt{\lambda} \sqrt{\lambda'} \trace_{E}(\ket{x_{\lambda}} \bra{x_{\lambda'}} \otimes
%\overline
{\ket{x_{\lambda}}}  %\overline
{\bra{x_{\lambda'}}})=\sum_{\lambda\in \Lambda} \sum_{x_{\lambda}\in\mathcal{X}_{\lambda}}
 \lambda \ket{x_{\lambda}} \bra{x_{\lambda}}=\rho_{T}
\end{equation*}

Secondly, we show that $\ket{\Phi}\bra{\Phi}$ has support on $\mathcal{N}$. Since $\ket{\Phi_{\lambda}}$ has a structure similar to a kind of ``maximally entangled state", the following property holds for any operator $A$:
%Secondly, we show that $\ket{\Phi}\bra{\Phi}$ has support on $(\mathcal{H}^{\otimes n} \otimes \mathcal{H}^{\otimes n})^{(s\otimes \overline{s})}$. For this, we consider a property of states with a structure like $\ket{\Phi_{\lambda}}$ (similar to a kind of ``maximally entangled state").%maximally entangled states (which applies to $\ket{\Phi_{\lambda}} \bra{\Phi_{\lambda}}$).
\begin{equation*} \begin{split}
\mathbbm{1}_{T} \otimes A \ket{\Phi_{\lambda}}& =
\mathbbm{1}_{T} \otimes A \sum_{x_{\lambda}\in\mathcal{X}_{\lambda}} \ket{x_{\lambda}} \otimes %\overline
{\ket{x_{\lambda}}}=
\sum_{x_{\lambda}\in\mathcal{X}_{\lambda}} \ket{x_{\lambda}} \otimes %\overline
{A \ket{x_{\lambda}}}
\\&
=\sum_{x_{\lambda} \in\mathcal{X}_{\lambda}} \ket{x_{\lambda}} \otimes %\overline
{\Big(\sum_{y_{\lambda} \in\mathcal{X}_{\lambda}} \ket{y_{\lambda}} \bra{y_{\lambda}}\Big) A\ket{x_{\lambda}}}
\\&
=\sum_{x_{\lambda},y_{\lambda} \in\mathcal{X}_{\lambda}} \ket{x_{\lambda}} \otimes %\overline
\ket{y_{\lambda}} \bra{y_{\lambda}} A\ket{x_{\lambda}}
=\sum_{x_{\lambda},y_{\lambda} \in\mathcal{X}_{\lambda}} \ket{x_{\lambda}} \otimes %\overline
\ket{y_{\lambda}} \bra{x_{\lambda}} A^T\ket{y_{\lambda}}
\\&
=\sum_{y_{\lambda} \in\mathcal{X}_{\lambda}} (\sum_{x_{\lambda} \in\mathcal{X}_{\lambda}}  \ket{x_{\lambda}}\bra{x_{\lambda}}) A^T\ket{y_{\lambda}} \otimes %\overline
\ket{y_{\lambda}}
=A^T \otimes \mathbbm{1}_{E} \sum_{y_{\lambda} \in\mathcal{X}_{\lambda}} \ket{y_{\lambda}} \otimes %\overline
\ket{y_{\lambda}} 
= A^T \otimes \mathbbm{1}_{E} \ket{\Phi_{\lambda}}
\end{split}
\end{equation*}
With this property, we find for symmetry transformations $s\in\mathcal{S}$:

\begin{equation*}
s \otimes \overline{s} \ket{\Phi_{\lambda}}= s(\overline{s})^T \otimes \mathbbm{1}_{E} \ket{\Phi_{\lambda}} = s s^{\dagger} \otimes \mathbbm{1}_{E} \ket{\Phi_{\lambda}}= \ket{\Phi_{\lambda}}
\end{equation*}
since we can assume $s$ to be a unitary representation with $ ss^{\dagger}=\mathbbm{1}_{T}$. By linearity, the above statement is also true for a linear combination of $\ket{\Phi_{\lambda}}$, and thus for $\ket{\Phi}$. In conclusion, $\ket{\Phi}\bra{\Phi}$ is a purification of $\rho_{T}$, and it has support on  $\mathcal{N}$.
%FOR SYMMETRY ACTIONS $s$ with $s^T s=\mathbbm{1}$.
%(\mathcal{H}^{\otimes n} \otimes \mathcal{H}^{\otimes n})^{(s\otimes \overline{s})

\end{proof}
For the permutation group and previous versions of a quantum de Finetti theorem, this lemma and its proof appear in \cite{RennerPhD, Lev16}.

With this established, we can investigate the diamond norm. To begin, we move from arbitrary states to states with support on the $\mathcal{S}$-invariant subspace $\mathcal{N}$. The following lemma implies that it is sufficient to consider such states when computing the diamond norm of a map that is invariant under $\mathcal{S}$.

\begin{repeatedlemma1}
%\label{lemmasym}

 %For any finite-dimensional space $\mathcal{R}_1$ and any (arbitrary) density operator $\sigma_{\mathcal{H}^{\otimes n} \mathcal{R}_1}$, the following holds:
Let $\Delta$ be a linear map from $\operatorname{End}(\mathcal{H}^{\otimes n})$ to $\operatorname{End}(\mathcal{H}')$ that is invariant under the symmetry $\mathcal{S}$.
For any finite-dimensional space $\mathcal{M}$ and any (arbitrary) density operator $\sigma_{TM}$, the following holds:

\begin{equation*}
\big\|  (\Delta \otimes \mathbbm{1})  \sigma_{TM} \big\|_{\tr}
 \leq
 \big\| (\Delta \otimes \mathbbm{1}) \rho_{TE} \big\|_{\tr}
\end{equation*}

where $ \rho_{TE}$  is a state with support on $\mathcal{N}=(\mathcal{H}^{\otimes n} \otimes \mathcal{H}^{\otimes n})^{(s\otimes \overline{s})}$.
\end{repeatedlemma1}

\begin{proof}[Proof of Lemma \ref{lemmasym}]

We introduce an additional space $\mathcal{L}$ with dimension $|\mathcal{S}|$ equal to the number of elements of the symmetry group $\mathcal{S}$. Then, we can write its orthonormal basis as $\{ \ket{s} \}_{s\in\mathcal{S}}$.

Using this, the following transformation of the state is possible:
\begin{equation*}
\big\| (\Delta \otimes \mathbbm{1}_{{M}})  \sigma_{TM} \big\|_{\tr} = \big\| (\Delta \otimes \mathbbm{1}_{ML})  (\sigma_{TM} \otimes \mathbbm{1}_{{L}})  \big\|_{\tr}
\end{equation*}
\begin{equation*}
= \big\| (\Delta \otimes \mathbbm{1}_{ML})  (\sigma_{TM} \otimes \frac{1}{|\mathcal{S}|} \sum_{s\in\mathcal{S}} \ket{s} \bra{s}_{{L}})  \big\|_{\tr}
= \big\|   \frac{1}{|\mathcal{S}|} \sum_{s\in\mathcal{S}} (\Delta \otimes \mathbbm{1}_{ML}) (\sigma_{T{M}} \otimes \ket{s} \bra{s}_{{L}})  \big\|_{\tr}
\end{equation*}

Since every $s$ is trace non-increasing, it can be inserted in the following way:

\begin{equation*}
\big\|   \frac{1}{|\mathcal{S}|} \sum_{s\in\mathcal{S}} (\Delta \otimes \mathbbm{1}_{ML}) (\sigma_{TM} \otimes \ket{s} \bra{s}_{{L}})  \big\|_{\tr} =\big\|   \frac{1}{|\mathcal{S}|} \sum_{s\in\mathcal{S}} (s \circ \Delta \otimes \mathbbm{1}_{ML}) (\sigma_{TM} \otimes \ket{s} \bra{s}_{L})  \big\|_{\tr}
\end{equation*}

By assumption, the CPTP map $\Delta$ is invariant under the symmetry $\mathcal{S}$, meaning $\Delta=s \circ \Delta \circ s^{\dagger}\  \forall s \in \mathcal{S}$. Thus, equivalently, $s \circ \Delta= \Delta \circ s$. Therefore:
\begin{equation*}
\big\|   \frac{1}{|\mathcal{S}|} \sum_{s\in\mathcal{S}} (s \circ \Delta \otimes \mathbbm{1}_{ML}) (\sigma_{TM} \otimes \ket{s} \bra{s}_{{L}})  \big\|_{\tr}
= \big\|   \frac{1}{|\mathcal{S}|} \sum_{s\in\mathcal{S}} (\Delta \circ s \otimes \mathbbm{1}_{ML}) (\sigma_{TM} \otimes \ket{s} \bra{s}_{{L}})  \big\|_{\tr}
\end{equation*}

%\begin{equation*}
%= \| \   (\Delta \otimes \mathbbm{1}_{ML})  \frac{1}{|\mathcal{S}|} \sum_{s\in\mathcal{S}} (s \otimes \mathbbm{1}_{ML}) (\sigma_{\mathcal{N} \mathcal{M}} \otimes \ket{s} \bra{s}_{\mathcal{L}})  \|_{\tr}
%\end{equation*}
\begin{equation*}
= \big\| (\Delta \otimes \mathbbm{1}_{ML})  \frac{1}{|\mathcal{S}|} \sum_{s\in\mathcal{S}}  (s \otimes \mathbbm{1}_{ML}) (\sigma_{TM} \otimes \ket{s} \bra{s}_{{L}})  \big\|_{\tr} = \big\|  (\Delta \otimes \mathbbm{1}_{ML})  (\sigma_{TML})  \big\|_{\tr}
\end{equation*}
Thus we have constructed a state $\sigma_{TML}$. %$\sigma_{\mathcal{H}^{\otimes n} \mathcal{M}\mathcal{L}}$
For this state, its marginal $\rho_{T}=\trace_{{M} {L} } (\sigma_{TML})$ is invariant under all $s$, thus invariant under the symmetry group $\mathcal{S}$ (by construction).

For any state $\rho_{T}$ that is invariant under the symmetry group $\mathcal{S}$, Lemma \ref{purif} states that there exists a purification $\rho_{TE}\in \mathfrak{S}(\mathcal{H}_T\otimes \mathcal{H}_E)$, $\mathcal{H}_E=\mathcal{H}_T'$ with support on $\mathcal{N}$.

All purifications are equal up to an isometry, and thus there exists a CPTP map $\mathcal{Z}: \operatorname{End}(\mathcal{H}^{\otimes n}) \rightarrow \operatorname{End}(\mathcal{M}\mathcal{L})$ such that $\sigma_{TML}=(\mathbbm{1}_{T} \otimes Z) \rho_{TE}$.

Using this and the fact that $\mathcal{Z}$ is trace-preserving (and thus trace-nonincreasing), we find:

\begin{equation*}
\big\|  (\Delta \otimes \mathbbm{1}_{ML})  (\sigma_{TML})  \big\|_{\tr} =\big\|  (\Delta \otimes \mathbbm{1}_{ML}) (\mathbbm{1}_{T} \otimes \mathcal{Z}) \rho_{TE}  \big\|_{\tr} \leq  \big\| (\Delta \otimes \mathbbm{1}_{ML}) \rho_{TE}  \big\|_{\tr}
\end{equation*}

In summary:
\begin{equation*}
\big\|  (\Delta \otimes \mathbbm{1}_{{M}})  \sigma_{TM} \big\|_{\tr} \leq  \big\| \  (\Delta \otimes \mathbbm{1}_{E}) \rho_{TE}  \big\|_{\tr}
\end{equation*}

%\cite{RennerPhD, Lev16, KM09}

\end{proof}

%Lastly!
%To prove Thm. 
%Resolution of Identity

As a second step in bounding the diamond norm, we establish a connection between states with support on the $\mathcal{S}$-invariant subspace $ (\mathcal{H}^{\otimes n})^s$, and a purification $\tau_{TEN}$ of our specific de Finetti state $\tau_{T}$ (\ref{deFinettistate}).

Here, the specific structure of  $\tau_{T}$ becomes important and  the assumption of resolution of identity for $\tau_{TE}$, as discussed in Section \ref{appendix-resolutionofidentity}, is required.

\begin{repeatedlemma2}
%\label{lemma2}
Suppose we have a state $\rho_{TE}$ with support on the subspace $\mathcal{N}=  (\mathcal{H}^{\otimes n} \otimes \mathcal{H}^{\otimes n})^{(s\otimes\overline{s})}\subseteq \mathcal{H}^{\otimes n} \otimes \mathcal{H}^{\otimes n}$. For any such state, there exists a linear completely positive trace-nonincreasing map $\mathcal{C}: \operatorname{End}(\mathcal{N}) \to \mathbbm{C}$ such that
\begin{equation*}
\rho_{TE}= g_{n,d} (\mathbbm{1}_{TE} \otimes \mathcal{C}) (\tau_{TEN})
\end{equation*}
with $\operatorname{tr}_{\mathcal{N}} \tau_{TEN} = \tau_{TE}=\frac{1}{g_{n,d}} \mathbbm{1}_{ N}$ and $g_{n,d}=\dim(\mathcal{N}) $.
\end{repeatedlemma2}

\begin{proof}[Proof of Lemma \ref{lemma2}]
Let $\{ \ket{i} \}_i$ be an eigenbasis of $\rho_{TE}$.
Since $\tau_{TEN}$ is a purification of $ \tau_{TE} \propto \mathbbm{1}_{ TEN }$, it can be written as $\tau_{TEN}=\ket{\Psi}\bra{\Psi} $ with a pure state $\ket{\Psi}= \frac{1}{\sqrt{g_{n,d}}} \sum_i \ket{i}\otimes\ket{i}$ with 
\begin{equation} 
g_{n,d}=\dim (\mathcal{H}^{\otimes n} \otimes \mathcal{H}^{\otimes n})^{(s\otimes\overline{s})}=\dim(\mathcal{N}).
 \label{dimN}
\end{equation}

Let $\mathcal{C}: \sigma_{\mathcal{N}} \mapsto \operatorname{tr}(\sigma_{\mathcal{N}} \rho^T_\mathcal{N}$). Then we find
\begin{equation*}
g_{n,d} (\mathbbm{1}_{TE} \otimes \mathcal{C}) (\tau_{TEN})= (\mathbbm{1}_{TE} \otimes \mathcal{C}) \sum_{i,j} (\ket{i}\otimes\ket{i}) (\bra{j}\otimes\bra{j})=\sum_{i,j} (\ket{i} \bra{j}\otimes \mathcal{C}(\ket{i} \bra{j})
\end{equation*}
\begin{equation*}
= \sum_{i,j} (\ket{i} \bra{j}\otimes \operatorname{tr}(\ket{i} \bra{j}  \rho^T_{\mathcal{N}})
=\sum_{i,j} \ket{i} \bra{j} ( \bra{j}  \rho^T_{\mathcal{N}}\ket{i})=\sum_{i,j} \ket{i} \bra{j}  (\rho^T_{\mathcal{N}})_{j,i} =\sum_{i,j} \ket{i} \bra{j}\otimes  (\rho_{TE})_{i,j} =\rho_{TE}.
\end{equation*}

Thus it is demonstrated that for any $\rho_{TE}$ on $ (\mathcal{H}^{\otimes n} \otimes \mathcal{H}^{\otimes n})^{(s\otimes\overline{s})}$, a map $\mathcal{C}$ exists such that the Lemma holds. 

\end{proof}

\begin{remark}
This proof can be interpreted as a teleportation. Since $ \tau_{TE} \propto \mathbbm{1}_{ N}$, its purification will be a maximally entangled state in the corresponding space. This maximally entangled state can then be used as a resource to teleport the state $\rho$ from $\mathcal{N}$ to $\mathcal{H}^{\otimes n} \otimes \mathcal{H}^{\otimes n}$.
\end{remark}

%\chapter{Orbit counting for $3$ parties}
%\label{appendix-Nparty}
%\textcolor{green}{tree diagram for K=3, K=4}
%Do I even still need this?

%\chapter{Choi-Jamio{\l}kowski isomorphism}
%\label{appendix-Choi}

%The Choi-Jamio{\l}kowski isomorphism is an isomorphism between linear maps and matrices.

%Given a linear map $\mathcal{M}_{A\rightarrow \tilde{A}}$, a matrix $M_{A\tilde{A}}$ is obtained by applying the map to one part of a maximally entangled state $\Phi_{A\tilde{A}}$.
%\[M_{A\tilde{A}} = (\mathbbm{1}_A \otimes \mathcal{M}_{A\rightarrow \tilde{A}}) (\Phi_{A\tilde{A}}) \]

%Given a matrix $M_{A\tilde{A}}$, the corresponding linear map $\mathcal{M}_{A\rightarrow \tilde{A}}$, acting on some state $\rho_A$, is obtained by the following relation:
%\[\mathcal{M}_{A\rightarrow \tilde{A}}(\rho_{A}) = \Tr_{A} (M_{A\tilde{A}} \mathbbm{1}_{\tilde{A}} \otimes \rho_A^T)\]

%In this section, we will demonstrate how the two definitions of channel fidelity (see \eqref{Fidelity-original}, \eqref{ChannelFidelity}) correspond to one another via the Choi-Jamio{\l}kowski isomorphism, by demonstrating how \eqref{ChannelFidelity} can be obtained from \eqref{Fidelity-original}.

\chapter{Proof of the de Finetti Theorem with Linear Constraints}
 \label{appendix-proofforhierarchy}

In this part of the appendix, the de Finetti theorem with linear constraints from Section \ref{section-stabdefinettithm} will be proven:

\begin{repeatedtheorem}[De Finetti Theorem with Linear Constraints] %\label{LinDeFinetti}
Let $\rho_{AB_1^n}$ be a quantum state that is permutation invariant with respect to permutations of the $n$ subsystems $B_1^n$, let $\Lambda_{A\rightarrow C_A}$ and $\Gamma_{B\rightarrow C_B}$ be linear maps, and $X_{C_A}$ and $Y_{C_B}$ be operators such that the following two linear constraints hold:
\[ \Lambda_{A\rightarrow C_A}(\rho_{AB_1^n})=X_{C_A}\otimes \rho_{B_1^n} ,\]
%\[ \Gamma_{B_{m+1} \rightarrow C_B}(\rho_{B_1^n})=\rho_{B_1^{m}} \otimes Y_{C_B} \otimes \rho_{B_{m+2}^n}\]
\[ \Gamma_{B_{n} \rightarrow C_B}(\rho_{B_1^n})=\rho_{B_1^{n-1}} \otimes Y_{C_B}.\]
Then, there exists an $m\in [0,n-1]$ and
 a probability distribution $\{p_Z(z_1^m)\}_{z_1^m\in Z}$ such that %quantum states $\omega_A^z$ and $\sigma_B^z$, and $0\leq m <n-1$ such that
\[\Big\| \rho_{AB_{m+1}} - \sum_{z_1^m} p_Z(z_1^m) \rho_{A|z_1^m} \otimes \rho_{B_{m+1}|z_1^m} \Big\|_{\tr} \leq\epsilon(d_B,d_A,n)\]
with \[\epsilon(d_B,d_A,n):=\min\Big\{d_B^2(d_B+1), 18\sqrt{d_A d_B}\Big\} \sqrt{\frac{2\ln(2)\ln(d_A)}{n}} \]
and \[\Lambda_{A\rightarrow C_A}(\rho_{A|z_1^m})=X_{C_A},\ \Gamma_{B_{m+1} \rightarrow C_B}(\rho_{B_{m+1}|z_1^m})= Y_{C_B}.\]
\end{repeatedtheorem}

To prove this theorem, there are three helpful preliminary lemmata. All of these lemmata are given and proven in either \cite{BBFS18} or \cite{BH16}. However, the first lemma, Lemma \ref{lemma-expectationvalue}, is stated to be requiring permutation invariance, which we have found to be an unnecessary assumption that can be omitted. We give here a more detailed proof than the original.%, where some key steps are done quickly and without additional explanation of the techniques used.

The first lemma connects the entire system's quantum state $\rho_{AZ_1^n}$ with post-measurement states (conditioned on some measurement outcomes), and provides a bound on the expectation value of their trace distance.

\begin{lemma}[See \cite{BBFS18}, Lemma 3.1]
\label{lemma-expectationvalue}
Let $\rho_{AZ_1^n}$ be a quantum state with the $Z_1^n$ systems classical. Then, there exists $0\leq m<n-1$, such that
\[  \mathbb{E}_{z_1^m} \Big\{ \Big\|\rho_{AZ_{m+1}|z_1^m} - \rho_{A|z_1^m}\otimes \rho_{Z_{m+1}|z_1^m} \Big\|_{\tr}^2\Big\}  \leq \frac{2\ln(2) \ln(d_A)}{n}.\]
\end{lemma}

\begin{proof}[Proof of Lemma \ref{lemma-expectationvalue}]
Using the fact that the quantum relative entropy of two quantum states is bounded by the dimension of a subsystem, we can find a bound on the expectation value of quantum relative entropy of the quantum state $\rho_{AZ_1^n}$ and a separable, conditional state $\rho_{A|z_1^m}\otimes \rho_{Z_{m+1}|z_1^m} $, before employing Pinzker's inequality to relate this to a bound on the distance of these two quantum states.

Quantum relative entropy is defined for arbitrary quantum states $\rho$ and $\sigma$ in the following way:
\[D(\rho|| \sigma) := \Tr\big(\rho(\log(\rho)-\log(\sigma)\big)\]
and is bound by the dimensions of the systems.

We compare the quantum state $\rho_{AZ_1^n}$ and a related separable state $ \Tr_{Z_1^n}(\rho_{AZ_1^n}) \otimes \Tr_A(\rho_{AZ_1^n})= \rho_A \otimes \rho_{Z_1^n}$, to find that the following always holds:
\[ D(\rho_{AZ_1^n} \| \rho_A \otimes \rho_{Z_1^n}) \leq \log(d_A) .\]

The quantum relative entropy above can be recast as a sum over the quantum relative entropy of different subsystems.
\[D(\rho_{AZ_1^n} \| \rho_A \otimes \rho_{Z_1^n})=\sum_{m=0}^{n-1} \Big( D(\rho_{AZ_1^{m+1}} \| \rho_A \otimes \rho_{Z_1^{m+1}})- D(\rho_{AZ_1^m} \| \rho_A \otimes \rho_{Z_1^m}) \Big) \leq \log(d_A)\]
which can be related to a bound on one individual term of the sum: Because the $Z$-systems are classical, each term is positive, and therefore there exists an $m\in[0,n-1]$ such that
\[D(\rho_{AZ_1^{m+1}} \| \rho_A \otimes \rho_{Z_1^{m+1}})- D(\rho_{AZ_1^m} \| \rho_A \otimes \rho_{Z_1^m})\leq \frac{\log(d_A)}{n}.
\]
Using the definition of the quantum relative entropy, we can write this out in terms of traces:
\smallskip\noindent
\begin{align*}
&
D(\rho_{AZ_1^{m+1}} \| \rho_A \otimes \rho_{Z_1^{m+1}})- D(\rho_{AZ_1^m} \| \rho_A \otimes \rho_{Z_1^m}) 
\\
&= \Tr \Big(\rho_{AZ_1^{m+1}} (\log(\rho_{AZ_1^{m+1}}) -\log(\rho_{A} \otimes \rho_{Z_1^{m+1}})\Big) - \Tr \Big(\rho_{AZ_1^{m}} (\log(\rho_{AZ_1^{m}}) -\log(\rho_{A} \otimes \rho_{Z_1^{m}})\Big) 
\\
&=
\Tr \Big(\rho_{AZ_1^{m+1}} \big(\log(\rho_{AZ_1^{m+1}}) -\log(\rho_{A} \otimes \rho_{Z_1^{m+1}} )- \log(\rho_{AZ_1^{m}})\otimes \mathbbm{1}_{Z_{m+1}} -\log(\rho_{A} \otimes \rho_{Z_1^{m}})\otimes \mathbbm{1}_{Z_{m+1}} \big) \Big)
\end{align*}
Now, two characteristics of the matrix logarithm are needed: On the one hand, if two matrices $A$ and $B$ commute, their logarithm is additive: $\log(AB)= \log(A)+\log(B)$. On the other hand, $\log(A\otimes \mathbbm{1})=\log(A)\otimes \mathbbm{1}$. Using these two identities, we obtain 
\smallskip\noindent
\begin{align*}
&D(\rho_{AZ_1^{m+1}} \| \rho_A \otimes \rho_{Z_1^{m+1}})- D(\rho_{AZ_1^m} \| \rho_A \otimes \rho_{Z_1^m})  = \\&
\Tr \Big(\rho_{AZ_1^{m+1}} \big(\log(\rho_{AZ_1^{m+1}}) -\log(\rho_{A} )\otimes \mathbbm{1}_{Z_1^{m+1}} -
\mathbbm{1}_{A} \otimes \log(\rho_{Z_1^{m+1}}) - \log(\rho_{AZ_1^{m}})\otimes \mathbbm{1}_{Z_{m+1}}
 \\& \ \ \ \ \ \ \ \ \ \ \ \
-\log(\rho_{A}) \otimes \mathbbm{1}_{Z_1^{m+1}} -
\mathbbm{1}_{A} \otimes \log(\rho_{Z_1^{m}})\otimes \mathbbm{1}_{Z_{m+1}} 
\big) \Big).
\end{align*}
Some of these terms cancel. In addition, we can make use of the fact that the $Z$ systems are classical, which means the state can be written as
\[\rho_{AZ_1^{m+1}} =  \sum_{z_1^{m+1}} p_Z(z_1^{m+1}) \rho_{A|z_1^{m+1}}\otimes \ket{z_1^m} \bra{z_1^m} = \sum_{z_1^{m}}  p_Z(z_1^{m}) \ket{z_1^m} \bra{z_1^m} \otimes \rho_{AZ_{m+1}|z_1^{m+1}},
\]
and its logarithm is
\[\log(\rho_{AZ_1^{m+1}})  = \sum_{z_1^{m}} \ket{z_1^m} \bra{z_1^m} \otimes \big(\log( p_Z(z_1^{m})) \mathbbm{1}_{AZ_{m+1}} + p_Z(z_1^{m}) \log(\rho_{AZ_{m+1}|z_1^m}) \big).\]
Inserting this, we arrive at the following:
\smallskip\noindent
\begin{align*}
&
D(\rho_{AZ_1^{m+1}} \| \rho_A \otimes \rho_{Z_1^{m+1}})- D(\rho_{AZ_1^m} \| \rho_A \otimes \rho_{Z_1^m}) 
\\
& =\Tr \bigg( \Big(
\sum_{z_1^{m+1}} p_Z(z_1^{m+1}) \ket{z_1^m} \bra{z_1^m} \otimes \rho_{AZ_{m+1}|z_1^{m+1}}\Big) \Big( 
\sum_{\tilde{z}_1^{m+1}} \ket{\tilde{z}_1^m} \bra{\tilde{z}_1^m} \otimes \big(
\log(p_{Z}(\tilde{z}_1^m) \otimes \mathbbm{1}_{AZ_{m+1}} 
\\&  \ \ \ \ \ \ \ \ \ \ \ \
+  p_Z(\tilde{z}_1^{m}) \log(\rho_{AZ_{m+1}|\tilde{z}_1^m}) - \log(p_{Z}(\tilde{z}_1^m) \otimes \mathbbm{1}_{AZ_{m+1}}  - p_Z(\tilde{z}_1^{m}) \mathbbm{1}_A \otimes \log(\rho_{Z_{m+1}|\tilde{z}_1^m})\\&  \ \ \ \ \ \ \ \ \ \ \ \
- \log(p_{Z}(\tilde{z}_1^m) \otimes \mathbbm{1}_{AZ_{m+1}}  - p_Z(\tilde{z}_1^{m})  \log(\rho_{A|\tilde{z}_1^m}) \otimes \mathbbm{1}_{Z_{m+1}}
+\log(p_{Z}(\tilde{z}_1^m) \otimes \mathbbm{1}_{AZ_{m+1}} \\&  \ \ \ \ \ \ \ \ \ \ \ \
+p_Z(\tilde{z}_1^{m}) \mathbbm{1}_A \otimes \log(\mathbbm{1}_{AZ_{m+1}})
\big)
 \Big)
\bigg)
\\&
=\Tr \bigg( 
\sum_{z_1^{m+1}} p(z_{m+1}|z_{1}^{m}) \ket{z_1^m} \bra{z_1^m} \otimes \Big(\rho_{AZ_{m+1}|z_1^{m+1}} \big(
  p_Z(z_1^{m}) \log(\rho_{AZ_{m+1}|z_1^m}) -  p_Z(z_1^{m}) \mathbbm{1}_A \otimes \log(\rho_{Z_{m+1}|z_1^m})\\&  \ \ \ \ \ \ \ \ \ \ \ \
 - p_Z(z_1^{m})  \log(\rho_{A|z_1^m}) \otimes \mathbbm{1}_{Z_{m+1}}
\big)
 \Big)
\bigg)
\\&
=\Tr \bigg(
\sum_{z_1^{m}} p_Z(z_1^{m+1}) p_Z(z_1^{m})
\ket{z_1^m} \bra{z_1^m} \otimes \Big( \rho_{AZ_{m+1}|z_1^{m+1}}  \big(
 \log(\rho_{AZ_{m+1}|z_1^m}) -  \log(\rho_{A|z_1^m} \otimes \rho_{Z_{m+1}|z_1^m})
\big)
 \Big)
\bigg)
\\&
= 
\sum_{z_1^{m}}  p_Z(z_1^{m}) \Tr \bigg( \sum_{z_{m+1}}
p_Z(z_1^{m+1}) p_Z(z_1^{m})\ket{z_1^m} \bra{z_1^m} \otimes \Big( \rho_{AZ_{m+1}|z_1^{m+1}}  \big(
 \log(\rho_{AZ_{m+1}|z_1^m}) -  \log(\rho_{A|z_1^m} \otimes \rho_{Z_{m+1}|z_1^m})
\big)
\Big)
\bigg)
\\&
= 
\sum_{z_1^{m}}  p_Z(z_1^{m}) \Tr \Big( \rho_{AZ_{m+1}|z_1^{m+1}}  \big(
 \log(\rho_{AZ_{m+1}|z_1^m}) -  \log(\rho_{A|z_1^m} \otimes \rho_{Z_{m+1}|z_1^m})
\big)
\Big)
\\&
=\mathbb{E}_{z_1^m} \Big\{ D( \rho_{AZ_{m+1}|z_1^m} \| \rho_{A|z_1^m}\otimes \rho_{Z_{m+1}|z_1^m}) \Big\} \leq \frac{\log(d_A)}{n}
\end{align*}
Thereby, we obtain a bound on the quantum relative entropy of the states we want to compare.
Now, we can use Pinsker's inequality \cite{BH16} to relate the quantum relative entropy to the trace distance of the states via
\[D(\rho||\sigma)\geq \frac{1}{2\ln(2)} \big( \Tr (|\sigma-\rho|) \big)^2 =  \frac{1}{2\ln(2)} \| \rho-\sigma\|_{\tr}^2 . \]
Then, we obtain the statement of the claim for a distance between density operators.
\end{proof}

%\begin{remark}
%Contrary to how it is stated in \cite{BBFS18}, this lemma requires no permutation invariance of the state for its proof.
%\end{remark}

To facilitate the use of Lemma \ref{lemma-expectationvalue}, we need a means to connect the physical systems $B_1^n$ to classical systems $Z_1^n$, which is done via measurement of the physical systems. There are two competing strategies to move between the systems, which relate to different changes in the trace distance (called measurement distortion).

\begin{lemma}[See \cite{BBFS18}, Lemma 3.2, and \cite{BH16}, Lemma 16]
\label{lemma-distortion1}
There\ exists\ a\ product\ mea-surement $M_A\otimes M_B$ with finitely many outcomes such that for any Hermitian and traceless operator $\xi_{AB}$, we have
\[ \Big\| (M_A\otimes M_B) \xi_{AB}\Big\|_{\tr} \geq  \frac{1}{18\sqrt{d_Ad_B}} \Big\|\xi_{AB} \Big\|_{\tr}.\]
\end{lemma}
This lemma has been proven in \cite{BH16}, Section 3.1. 

\begin{lemma}[See \cite{BBFS18}, Lemma 3.3]
\label{lemma-distortion2}
Consider a state two-design on $B$, i.e. a set of rank-one projectors $\{P_z\}$ such that $\frac{1}{n} \sum_{z=1}^Z P_z \otimes P_z =\frac{2P_{Symm}}{d_B(d_B+1)}$, where $P_{Symm}$ is the projector on the symmetric subspace of $B\otimes B$. Let $M_B$ be the measurement defined by
\[ M_B(x)=\sum_z \frac{d_B}{Z} \trace(P_z x) \ket{z}\bra{z} .\]
Then, for any Hermitian opertor $\xi_{AB}$, 
\[ \Big\| (\mathbbm{1}_A\otimes M_B) \xi_{AB}\Big\|_{\tr} \geq  \frac{1}{d_B^2(d_B+1)} \Big\|\xi_{AB} \Big\|_{\tr}.\]
\end{lemma}
This lemma has been proven in \cite{BBFS18}, Section 3.

It depends on the subsystem's underlying dimensions which of these distortions has a greater impact; we are interested in the minimum of the two.

Having established these three lemmata, we can state the proof for the original de Finetti statement with linear constraints, Theorem \ref{LinDeFinetti}.
\begin{proof}[Proof of Theorem \ref{LinDeFinetti}]

This proof has three key steps: First, we will show that bounding the left-hand side is related to bounding an expectation value of the squared trace distance of two states.
Then, Lemmas \ref{lemma-expectationvalue}, \ref{lemma-distortion1} and \ref{lemma-distortion2} are employed to bound this expectation value. Lastly, it remains to be shown that the linear constraints are obeyed.

For some $m\in[0,n-1]$, note that

\begin{equation*} \begin{split}
\Big\| \rho_{AB_{m+1}} - \sum_{z_1^m} p_Z(z_1^m) \rho_{A|z_1^m} \otimes \rho_{B_{m+1}|z_1^m} \Big\|_{\tr} & = \Big\| \rho_{AB_{m+1}} - \mathbb{E}_{z_1^m} \Big\{ \rho_{A|z_1^m} \otimes \rho_{B_{m+1}|z_1^m} \Big\} \Big\|_{\tr} \\
&\leq \mathbb{E}_{z_1^m} \Big\{ \Big\| \rho_{AB_{m+1}|z_1^m} -  \rho_{A|z_1^m} \otimes \rho_{B_{m+1}|z_1^m} \Big\|_{\tr}\Big\} \\
&\leq \sqrt{\mathbb{E}_{z_1^m} \Big\{ \Big\| \rho_{AB_{m+1}|z_1^m} -  \rho_{A|z_1^m} \otimes \rho_{B_{m+1}|z_1^m} \Big\|_{\tr}^2 \Big\}}.
\end{split}
\end{equation*} 

This follows from the fact that $\mathbb{E}_{z_1^m}\rho_{AB_{m+1}|z_1^m}=\rho_{AB_{m+1}}$, the convexity of the norm, and the convexity of the square function. Note that a bound on this expectation value of squared trace distance will result in a bound on the original statement, so we will now be bounding this expression instead.

To make use of Lemma \ref{lemma-expectationvalue}, which provides a bound on an expectation value of squared trace distance, the system on Bob's side must be classical. To this end, a measurement needs to be performed to map the $(m+1)$-th system from $B_{m+1}$ to $Z_{m+1}$ at the cost of some distortion. Here, either Lemma \ref{lemma-distortion1} or Lemma \ref{lemma-distortion2} could be applied, at different costs.

In the first case, using Lemma \ref{lemma-distortion1}, we obtain

\smallskip\noindent
\begin{align*}
&\mathbb{E}_{z_1^m} \Big\{ \Big\| \rho_{AB_{m+1}|z_1^m} -  \rho_{A|z_1^m} \otimes \rho_{B_{m+1}|z_1^m} \Big\|_{\tr}^2 \Big\}\\
& \leq ( 18\sqrt{d_A d_B})^2 \mathbb{E}_{z_1^m} \Big\{ \Big\| (M_A \otimes M_{B_{m+1}})\rho_{AB_{m+1}|z_1^m} -  \rho_{A|z_1^m} \otimes \rho_{B_{m+1}|z_1^m} \Big\|_{\tr}^2\Big\} 
 \\
& = (18\sqrt{d_A d_B})^2  \mathbb{E}_{z_1^m} \Big\{ \Big\| \rho_{AZ_{m+1}|z_1^m} -  \rho_{A|z_1^m} \otimes \rho_{Z_{m+1}|z_1^m} \Big\|_{\tr}^2 \Big\} 
 \\
& \leq (18\sqrt{d_A d_B})^2 {\frac{2\ln(2)\ln(d_A)}{n}}.
\end{align*} 

In the second case, using Lemma \ref{lemma-distortion2}, we obtain:

\smallskip\noindent
\begin{align*}
&\mathbb{E}_{z_1^m} \Big\{ \Big\| \rho_{AB_{m+1}|z_1^m} -  \rho_{A|z_1^m} \otimes \rho_{B_{m+1}|z_1^m} \Big\|_{\tr}^2 \Big\} \\
& \leq (d_B^2(d_B+1))^2 \mathbb{E}_{z_1^m} \Big\{ \Big\| (\mathbbm{1}_A \otimes M_{B_{m+1}})\rho_{AB_{m+1}|z_1^m} -  \rho_{A|z_1^m} \otimes \rho_{B_{m+1}|z_1^m} \Big\|_{\tr}^2 \Big\}
 \\
& =(d_B^2(d_B+1))^2  \mathbb{E}_{z_1^m} \Big\{ \Big\| \rho_{AZ_{m+1}|z_1^m} -  \rho_{A|z_1^m} \otimes \rho_{Z_{m+1}|z_1^m} \Big\|_{\tr}^2 \Big\} 
 \\
& \leq ( d_B^2(d_B+1))^2 {\frac{2\ln(2)\ln(d_A)}{n}}.
\end{align*}

Since we are interested in the best possible upper bound, it should be as small as possible; which upper bound is smaller depends on the underlying dimensions $d_A$ and $d_B$, which means that the best possible bound can be chosen depending on the setting of interest. In general, we take the minimum of the two possibilities:

\smallskip\noindent
\begin{align*}
&\mathbb{E}_{z_1^m} \Big\{ \Big\| \rho_{AB_{m+1}|z_1^m} -  \rho_{A|z_1^m} \otimes \rho_{B_{m+1}|z_1^m} \Big\|_{\tr}^2 \Big\} 
 \\
& \leq\min\Big\{(d_B^2(d_B+1))^2, (18\sqrt{d_A d_B})^2 \Big\} {\frac{2\ln(2)\ln(d_A)}{n}}
\end{align*}

Lastly, the two previous steps are combined in taking the square root of the above expectation value to obtain

\begin{equation*} \begin{split}
&\Big\| \rho_{AB_{m+1}} - \sum_{z_1^m} p_Z(z_1^m) \rho_{A|z_1^m} \otimes \rho_{B_{m+1}|z_1^m} \Big\|_{\tr} \\
& \leq \sqrt{\mathbb{E}_{z_1^m} \Big\{ \Big\| \rho_{AB_{m+1}|z_1^m} -  \rho_{A|z_1^m} \otimes \rho_{B_{m+1}|z_1^m} \Big\|_{\tr}^2 \Big\}}\\
& \leq\min\Big\{d_B^2(d_B+1), 18\sqrt{d_A d_B} \Big\} \sqrt{\frac{2\ln(2)\ln(d_A)}{n}} \\
&=\epsilon(d_B,d_A,n).
\end{split}
\end{equation*} 

Finally, it must be checked that the additional linear constraints are upheld, by showing that the initial condition $ \Lambda_{A\rightarrow C_A}(\rho_{AB_1^n})=X_{C_A}\otimes \rho_{B_1^n}$ implies $\Lambda_{A\rightarrow C_A}(\rho_{A|z_1^m})=X_{C_A}$. We find

\begin{equation*} \begin{split}
 \Lambda_{A\rightarrow C_A}(\rho_{A|z_1^m}) & =%\big(\Lambda_{A\rightarrow C_A} \otimes \mathbbm{1}_{B_1^n} \big)()% 
\Lambda_{A\rightarrow C_A} \bigg(\dfrac{\Tr_{Z_1^m B_{m+1}^n}  \big((\mathbbm{1}_A\otimes \Pi_{z_1^m} \otimes \mathbbm{1}_{B_{m+1}^n}) \rho_{AB_1^n} \big)}{\Tr_{A Z_1^m B_{m+1}^n}  \big((\mathbbm{1}_A\otimes \Pi_{z_1^m} \otimes \mathbbm{1}_{B_{m+1}^n} ) \rho_{AB_1^n} \big)}\bigg)\\
&  = \dfrac{\Tr_{Z_1^m B_{m+1}^n}  \big(\Lambda_{A\rightarrow C_A}\big((\mathbbm{1}_A\otimes \Pi_{z_1^m} \otimes \mathbbm{1}_{B_{m+1}^n})  \rho_{AB_1^n} )\big)}{\Tr_{A Z_1^m B_{m+1}^n}  \big((\mathbbm{1}_A\otimes \Pi_{z_1^m} \otimes \mathbbm{1}_{B_{m+1}^n}) \rho_{AB_1^n} \big)}\\
&= \dfrac{\Tr_{Z_1^m B_{m+1}^n}  \big( X_{C_A} \otimes \big((\Pi_{z_1^m} \otimes \mathbbm{1}_{B_{m+1}^n})\rho_{B_1^n} \big)\big)}{\Tr_{A Z_1^m B_{m+1}^n}  \big((\mathbbm{1}_A\otimes \Pi_{z_1^m} \otimes \mathbbm{1}_{B_{m+1}^n}) \rho_{AB_1^n} \big)}\\
&= X_{C_A} \dfrac{\Tr_{Z_1^m B_{m+1}^n}  \big(( \Pi_{z_1^m} \otimes \mathbbm{1}_{B_{m+1}^n}) \rho_{B_1^n} )\big)}{\Tr_{A Z_1^m B_{m+1}^n}  \big((\mathbbm{1}_A\otimes \Pi_{z_1^m} \otimes \mathbbm{1}_{B_{m+1}^n}) \rho_{AB_1^n} \big)}\\
&=X_{C_A} \dfrac{\Tr_{Z_1^m B_{m+1}^n}  \big(( \Pi_{z_1^m} \otimes \mathbbm{1}_{B_{m+1}^n}) ( \Tr_A (\rho_{AB_1^n}) )\big)}{\Tr_{A Z_1^m B_{m+1}^n}  \big((\mathbbm{1}_A\otimes \Pi_{z_1^m} \otimes \mathbbm{1}_{B_{m+1}^n}) \rho_{AB_1^n} \big)}\\
&=X_{C_A} \dfrac{\Tr_{A Z_1^m B_{m+1}^n}  \big((\mathbbm{1}_A\otimes \Pi_{z_1^m} \otimes \mathbbm{1}_{B_{m+1}^n})\rho_{AB_1^n} )\big)}{\Tr_{A Z_1^m B_{m+1}^n}  \big((\mathbbm{1}_A\otimes \Pi_{z_1^m} \otimes \mathbbm{1}_{B_{m+1}^n}) \rho_{AB_1^n} \big)}\\
&=X_{C_A}
\end{split}
\end{equation*}

and analogous for $ \Gamma_{B_{m+1} \rightarrow C_B}$. Note that this is the first and only time that the assumption of permutation invariance is needed: to specify that $ \Gamma_{B_{m+1} \rightarrow C_B}$ acts in the same way as $ \Gamma_{B_n \rightarrow C_B}$ for any $m$. Then, 

\begin{equation*} \begin{split}
 \Gamma_{B_{m+1} \rightarrow C_B}(\rho_{B_{m+1}|z_1^m})& =%\big(\Gamma_{B\rightarrow C_B} \otimes \mathbbm{1}_{B_1^n} \big)()% 
\Gamma_{B\rightarrow C_B} \bigg(\dfrac{\Tr_{Z_1^m B_{m+2}^n}  \big((  \Pi_{z_1^m}\otimes \mathbbm{1}_{B_{m+1}^n}) \rho_{B_1^n} \big)}{\Tr_{Z_1^m B_{m+1}^n}  \big((  \Pi_{B_1^m} \otimes \mathbbm{1}_{B_{m+1}^n} ) \rho_{B_1^n} \big)}\bigg)\\
%&= \big(\Gamma_{B\rightarrow C_B} \otimes \mathbbm{1}_{B_1^n} \big) \bigg(\dfrac{\Tr_{Z_1^m B_{m+2}^n}  \big(( \Pi_{z_1^m} \otimes \mathbbm{1}_{B_{m+1}^n}) \rho_{B_1^n} \big)}{\Tr_{Z_1^m B_{m+1}^n}  \big(( \Pi_{z_1^m} \otimes \mathbbm{1}_{B_{m+1}^n} ) \rho_{B_1^n} \big)}\bigg) \\
&  = \dfrac{\Tr_{Z_1^m B_{m+2}^n}  \big(( \Pi_{z_1^m} \otimes \mathbbm{1}_{B_{m+1}^n}) \Gamma_{B\rightarrow C_B}( \rho_{B_1^n} )\big)}{\Tr_{Z_1^m B_{m+1}^n}  \big(( \Pi_{z_1^m} \otimes \mathbbm{1}_{B_{m+1}^n} ) \rho_{B_1^n} \big)}\\
&= \dfrac{\Tr_{Z_1^m B_{m+2}^n}  \big(( \Pi_{z_1^m} \otimes \mathbbm{1}_{B_{m+1}^n}) \rho_{B_1^m}\otimes  Y_{C_B}\otimes \rho_{B_{m+2}^n} )\big)}{\Tr_{Z_1^m B_{m+1}^n}  \big(( \Pi_{z_1^m} \otimes \mathbbm{1}_{B_{m+1}^n} ) \rho_{B_1^n} \big)}\\
%&= Y_{C_B} \dfrac{\Tr_{Z_1^m B_{m+2}^n}  \big(( \Pi_{z_1^m} \otimes \mathbbm{1}_{B_{m+1}^n} \rho_{B_1^n} )\big)}{\Tr_{Z_1^m B_{m+1}^n}  \big(( \Pi_{z_1^m} \otimes \mathbbm{1}_{B_{m+1}^n} ) \rho_{B_1^n} \big)}\\
&=Y_{C_B} \dfrac{\Tr_{Z_1^m B_{m+2}^n}  \big((\Pi_{z_1^m} \otimes \mathbbm{1}_{B_{m+1}^n}) ( \Tr_{B_{m+1}} (\rho_{B_1^n}) )\big)}{\Tr_{Z_1^m B_{m+1}^n}  \big(( \Pi_{z_1^m} \otimes \mathbbm{1}_{B_{m+1}^n} ) \rho_{B_1^n} \big)}\\
&=Y_{C_B} \dfrac{\Tr_{Z_1^m B_{m+1}^n}  \big(( \Pi_{z_1^m} \otimes \mathbbm{1}_{B_{m+1}^n} )\rho_{B_1^n} )\big)}{\Tr_{Z_1^m B_{m+1}^n}  \big(( \Pi_{z_1^m} \otimes \mathbbm{1}_{B_{m+1}^n} ) \rho_{B_1^n} \big)}\\
&=Y_{C_B}.
\end{split}
\end{equation*} 

\end{proof}

\newpage
\chapter{Stabilizer de Finetti Theorem with Linear Constraints with Stochastic Orthogonal Invariance on Both Sides}
\label{appendix-bothsides}

It may happen that one is interested in a problem where Alice and Bob both have access to a system with stochastic orthogonal invariance, or where two systems should each be approximated by Clifford operations. %, %see for example \cite{PR17} about QEC for the amplitude damping channel using a Clifford encoder and a Clifford decoder.
Therefore, one could also be interested in studying maximum channel fidelity for such cases via an SDP hierarchy: %\cite{LNCY97,FSW07,PR17}. In this context, one could also be interested in studying maximum channel fidelity for such cases via an SDP hierarchy.

\begin{optimize}\label{Fidelity-bothsides}
\begin{equation*} \begin{split} 
F_{CC}(N,d_M)=\text{maximize} \ & F_s\Big(\Phi_{\tilde{B}R}, \big( ( \mathcal{D}_{B\rightarrow\tilde{B}} \circ \mathcal{N}_{\tilde{A}\rightarrow {B}} \circ \mathcal{E}_{A\rightarrow\tilde{A}})\otimes \mathbbm{1}_R \big) (\Phi_{AR})\Big)
\\
\text{subject to }%& \mathcal{E}_{A\rightarrow\tilde{A}} \text{ is a quantum channel}\\
&  \mathcal{E}_{A\rightarrow\tilde{A}}, \mathcal{D}_{B\rightarrow\tilde{B}} \text{ are Clifford channels}
\end{split}
\end{equation*}
\end{optimize}

Therefore, and for completeness' sake, it should also be mentioned that the states in previous stabilizer de Finetti theorems with linear constraints can easily be extended to be approximated by stabilizer states on both sides, using the triangle inequality, which still preserves separability in the cut between Alice and Bob.

\begin{theorem}[Stabilizer de Finetti Theorem with Linear Constraints, with stochastic orthogonal invariance on both sides]
\label{mainthm-both}
Let $\rho_{A_1^n B_1^n}$ be a quantum state on the Hilbert space $\mathcal{H}_A^{\otimes n}\otimes\mathcal{H}_B^{\otimes n}$ that commutes with the action of $O_n$ acting on the systems $A_1^n=A_1 A_2 \cdots A_n$, and $B_1^n=B_1 B_2 \cdots B_n$, respectively. $\mathcal{H}_A$ is a Hilbert space containing $r_A$ qudits with local dimension $d_A$, and $\mathcal{H}_B$ contains $r_B$ qudits with local dimension $d_B$. % The Hilbert spaces have dimensions each contain $r$ qudits, i.e. $\mathcal{H}_A=\mathcal{H_a}^{\otimes r}$\dim(\mathcal{H}_A)=d_A^r$ and $\dim(\mathcal{H}_B)=d_B^r$. 
Let $\Lambda_{A\rightarrow C_A}$ and $\Gamma_{B\rightarrow C_B}$ be linear maps, and $X_{C_A}$ and $Y_{C_B}$ be operators such that the following two linear constraints hold:
\[ \Lambda_{A\rightarrow C_A}(\rho_{AB_1^n})=X_{C_A}\otimes \rho_{B_1^n} ,\]
\[ \Gamma_{B_n \rightarrow C_B}(\rho_{B_1^n})=\rho_{B_1^{n-1}} \otimes Y_{C_B} .\]
Then, there exists a probability distribution $\{p_Z(z)\}_{z\in Z}$ and probability distributions $\{p_{S_A}(\sigma_A)\}$ and $\{p_{S_B}(\sigma_B)\}$ over the set of mixed stabilizer states on $r$ qudits on the respective spaces such that

 \begin{align*}
&\Big\|\rho_{AB}-\sum_{z,\sigma_A,\sigma_B} p_Z(z) p_{S_A}(\sigma_A) p_{S_B}(\sigma_B) \sigma_A \otimes  \sigma_B \Big\|_{\tr} \\
& \ \ \ \ \ \ \ \ \ \ \leq  \min\Big\{\epsilon(d_B^r,d_A^r,n),  \epsilon(d_A^r,d_B^r,n)\Big\} +\bar{\epsilon}(d_A,r_A,n)+\bar{\epsilon}(d_B,r_B,n)\\
 \end{align*}
where $\sigma_A$ are the mixed stabilizer states of $r_A$ $d_A$-dimensional qudits on $\mathcal{H}_A={(\mathbbm{C}^{d_A})}^{\otimes r}$ and $\sigma_B$ are the mixed stabilizer states of $r_B$ $d_B$-dimensional qudits on $\mathcal{H}_B={(\mathbbm{C}^{d_B})}^{\otimes r}$. $\epsilon(d_B^r,d_A^r,n)$ and $\epsilon(d_A^r,d_B^r,n)$ are defined in \eqref{epsilon}, and $\bar{\epsilon}(d_A,r_A,n)$ and $\bar{\epsilon}(d_B,r_B,n)$ are defined in \eqref{epsilon2}. In addition,

\[\Lambda_{A\rightarrow C_A}(\rho_{A|z})=X_{C_A},\ \Gamma_{B_{n} \rightarrow C_B}(\rho_{B|z})= Y_{C_B},\]

\[
\Big\| \sum_{z} p_Z(z)  \big(\rho_{A|z}- \sum_{\sigma_A} p_{S_A}(\sigma_A) \sigma_{A}\big) \Big\|_{\tr} \leq \bar{\epsilon}(d_A,r_A,n),
\]
\[
\Big\| \sum_{z} p_Z(z)  \big(\rho_{B|z_1^k}- \sum_{\sigma_B} p_{S_B}(\sigma_B) \sigma_{B}\big) \Big\|_{\tr} \leq \bar{\epsilon}(d_B,r_B,n).
\]
\end{theorem}

\begin{proof}[Proof of Theorem \ref{mainthm-both}]
With all the constraints in place, it is possible to apply Theorem \ref{mainthm}. Note that $\dim(\mathcal{H}_A)=d_A^{r_A}$ and $\dim(\mathcal{H}_B)=d_B^{r_B}$. Therefore, there exists an $m\in[0,n-1]$ such that 
\begin{equation} \begin{split} \label{both-proof-step1}
&\Big\|  \rho_{AB_{m+1}} - \sum_{z_1^m} p_Z(z_1^m) \rho_{A|z_1^m} \otimes \sum_{\sigma} p_S(\sigma) \sigma_{B_{m+1}}) \Big\|_{\tr}  \\
& \leq \epsilon(d_B^{r_B},d_A^{r_A},n)  + 2d_B^{2(r_B+1)^2}d_B^{-\frac{1}{2} (n-{(m+1)})} 
\end{split} \end{equation}
%\[ \Big\|\rho_{AB}-\sum_{z,\sigma} p_Z(z) p_{S}(\sigma)\rho_{A|z} \otimes  \sigma_B \Big\|_{\tr} \leq \min\Big\{d_B^{2r}(d_B^r+1), 18\sqrt{d_A d_B^r}\Big\} \sqrt{\frac{2\ln(2)\ln(d_A)}{n}}+4d_B^{2(r+1)^2}d_B^{-\frac{1}{2} (n-1)} \]
with $\sigma_B$ being the mixed stabilizer states of $r$ qudits on $\mathcal{H}_B={(\mathbbm{C}^{d_B})}^{\otimes r_B}$, and
\[\Lambda_{A\rightarrow C_A}(\rho_{A|z})=X_{C_A},\ \Gamma_{B_{n} \rightarrow C_B}(\rho_{B|x})= Y_{C_B},\]

\[
\Big\| \sum_{z_1^k} p_Z(z_1^k)  \big(\rho_{B|z_1^k}- \sum_{\sigma_B} p_S(\sigma_B) \sigma_{B}\big) \Big\|_{\tr} \leq 2d_B^{2(r_B+1)^2}d_B^{-\frac{1}{2} (n-(m+1)} .
\]

In comparison with the one-sided version of this theorem, there is now an additional constraint on Alice's side: $\rho_{A_1^n B_1^n}$ commutes with the action of $O_n$ acting on the systems $A=A_1^n=A_1 A_2 \cdots A_n$, which directly implies that $\rho_{A_1^n}=\Tr_{B_1^n}(\rho_{A_1^n B_1^n})$ is invariant under the stochastic orthogonal group. Therefore, the stabilizer de Finetti theorem in Theorem \ref{StabDeFinetti}, holds for the subsystems $A_1^n$:

\begin{equation} \label{both-proof-step2}
 \Big\|\rho_{A_1^{k+1}} - \sum_{\sigma_A} p_{S_A}(\sigma_A) \sigma_A^{\otimes (k+1)}\Big\|_{\tr} \leq 2d_A^{2(r_A+1)^2}d_A^{-\frac{1}{2} (n-(k+1))}  
\end{equation}

Using Observation \ref{CPTPmapLemma2}, and the same argumentation as in the final steps of the proof of Theorem \ref{mainthm} (see \eqref{proof-steplast}), there exists a CPTP map $\mathcal{M}'$, which transforms \eqref{both-proof-step2} into

\smallskip\noindent
\begin{equation} \begin{split}\label{both-proof-step3}
 2d_A^{2(r_A+1)^2}d_A^{-\frac{1}{2} (n-(k+1))} &\geq \Big\|\rho_{A_1^{k+1}} - \sum_{\sigma_A} p_{S_A}(\sigma_A) \sigma_{A_{k+1}} \Big\|_{\tr}\\
& \geq \Big\| \mathcal{M}' \big(\rho_{A_1^{k+1}} - \sum_{\sigma_A} p_{S_A}(\sigma_A)  \sigma_{A_{k+1}} \big) \Big\|_{\tr}\\
&=\Big\| \sum_{z_1^k} p_Z(z_1^k)  \big(\rho_{A_{k+1}|z_1^k}- \sum_{\sigma_A} p_{S_A}(\sigma_A) \sigma_{A_{k+1}}\big) \Big\|_{\tr}.
 \end{split}
\end{equation}

It is integral that the measurement contained in the maps used in \eqref{both-proof-step2} and here should be the same, to ensure that the resulting probability distribution $p_Z$ is the same.

Then, \eqref{both-proof-step1} and \eqref{both-proof-step3} can be combined via the triangle inequality:

\smallskip\noindent
\begin{align*}%\label{both-proof-step3}
&\Big\|\rho_{AB}-\sum_{z,\sigma_A,\sigma_B} p_Z(z) p_{S_A}(\sigma_A) p_{S_B}(\sigma_B) \sigma_A \otimes  \sigma_B \Big\|_{\tr}\\
&=\Big\| \sum_{z_1^k} p_Z(z_1^k)  \big(\rho_{B|z_1^k}- \sum_{\sigma_B} p_{S_B}(\sigma_B) \sigma_{B}\big) \\&\ \ \ \ \ \ \ \ \ \ \ \ \ \ \ \ \ \ \ \ + \Big( \sum_{z_1^k} p_Z(z_1^k)  \big(\rho_{A_{k+1}|z_1^k}- \sum_{\sigma_A} p_{S_A}(\sigma_A) \sigma_{A_{k+1}}\big) \Big) \otimes \sum_{\sigma_B} p_{S_B}(\sigma_B) \sigma_B \Big\|_{\tr}\\
&\leq \Big\| \sum_{z_1^k} p_Z(z_1^k)  \big(\rho_{B|z_1^k}- \sum_{\sigma_B} p_{S_B}(\sigma_B) \sigma_{B}\big)\Big\|_{\tr} \\& \ \ \ \ \ \ \ \ \ \ \ \ \ \ \ \ \ \ \ \ + \Big\| \Big( \sum_{z_1^k} p_Z(z_1^k)  \big(\rho_{A_{k+1}|z_1^k}- \sum_{\sigma_A} p_{S_A}(\sigma_A) \sigma_{A_{k+1}}\big) \Big) \otimes \sum_{\sigma_B} p_{S_B}(\sigma_B) \sigma_B \Big\|_{\tr}\\
&\leq \Big\| \sum_{z_1^k} p_Z(z_1^k)  \big(\rho_{B|z_1^k}- \sum_{\sigma_B} p_{S_B}(\sigma_B) \sigma_{B}\big)\Big\|_{\tr} + \Big\|  \sum_{z_1^k} p_Z(z_1^k)  \big(\rho_{A_{k+1}|z_1^k}- \sum_{\sigma_A} p_{S_A}(\sigma_A) \sigma_{A_{k+1}}\big)  \Big\|_{\tr}\\
&\leq \epsilon(d_B^{r_B},d_A^{r_A},n) + 2d_B^{2(r_B+1)^2}d_B^{-\frac{1}{2} (n-{(m+1)})}  +2d_A^{2(r_A+1)^2}d_A^{-\frac{1}{2} (n-(k+1))}
\end{align*}

Because permutations are a subgroup of stochastic orthogonal group, we can choose $m=k=0$.

Here, it has to be noted that the proof could also be done the other way around, switching out $A$ and $B$, and first using the orthogonal invariance of $A$ to get a statement about separability and closeness to convex combinations of stabilizer states, and then using the invariance on the $B$ side. If this was switched, the dimensions $d_A$ and $d_B$ would also be switched in the bound. In total, we are interested in the minimum bound, given by taking the minimum over all the available options. Therefore, there are four terms in the final statement of which one must choose the minimum. %of which the minimum must be chosen depending on the underlying dimensions.

Finally, the linear constraint on Alice's side can be obtained using Observation \ref{CPTPmapLemma}, in analogy to the final step in the proof of \ref{mainthm}.
\end{proof}

%\begin{remark}
Using Theorem \ref{mainthm-both}, the hierarchy described in \eqref{L_n} can be extended to a hierarchy where the stochastic orthogonal invariance of both Alice's and Bob's side implies the closeness to a separable state with stabilizer states on both sides.

\begin{optimize}
\begin{equation*} \begin{split}
\text{maximize} \ & d_{\tilde{A}} d_B^{r_B} \Tr \Big( \big(J_{\tilde{A}B}^N \otimes \Phi_{A\tilde{B}} \big) \big(  \rho_{A\tilde{A}B\tilde{B}} \big) \Big)\\
\text{subject to }&  \rho_{(A\tilde{A})_1^n (B\tilde{B})_1^n}  \geq 0, \ \Tr(\rho_{(A\tilde{A})_1^n (B\tilde{B})_1^n})=1 \\
& \rho_{(A\tilde{A})_1^n(B\tilde{B})_1^n} \text{ is invariant under the action of } O_n \text{ with respect to }A_1^n\text{ and }B_1^n\\
  & \Tr_{\tilde{A_1}}( \rho_{(A\tilde{A})_1^n(B\tilde{B})_1^n} ) =\frac{\mathbbm{1}_{A_1}}{d_{A}^{r_A}}\otimes \rho_{(A\tilde{A})_2^n(B\tilde{B})_1^n}\\&
\Tr_{\tilde{B}_n}( \rho_{(A\tilde{A})_1^n(B\tilde{B})_1^n})  = \rho_{(A\tilde{A})_1^n(B\tilde{B})_1^{n-1}} \otimes \frac{\mathbbm{1}_{B_{n}}}{d_B^{r_B}} \\
\end{split}
\end{equation*}
\end{optimize}

In the one-sided hierarchy, only the decoder (or, symmetrically, only the encoder) is approximated by Choi matrices of stabilizer states, which are Clifford operations. With stochastic orthogonal invariance on both sides, both the encoder and the decoder are approximated by Clifford operations.
%\end{remark}

It is also possible to extend stabilizer de Finetti theorem for qubits (Theorem \ref{mainthm-qubit}) such that both Alice's and Bob's part of the state are invariant under all permutations and anti-identity, leading to an approximation for qubits by stabilizer state tensor powers with separability between Alice and Bob.

\newpage
\pagenumbering{gobble}
\thispagestyle{empty}
\cleardoublepage
\newpage

\fancyhead{}
%\addcontentsline{toc}{chapter}{Note of Thanks}
\section*{Acknowledgement}

For this thesis to come into existence, many parties had to communicate over large distances via a complicated protocol. The goal of the protocol was to hand in this thesis. The state of my brain was prepared via local interaction in Cologne, transmitted to Zurich, transformed via local operations, and then (surprisingly) sent to Constance to work on the desired output. 
During this process, I have interacted with many people without whom the output would have been considerably more noisy.

I owe particular thanks to Professor David Gross for letting me join his group, getting me started on this project and offering me valuable support and life advice. In addition, I want to extend my thanks to the whole Gross group, especially Felipe Montealegre-Mora for helping me understand some representation theory stuff and Mariami Gachechiladze for trying to understand QKD with me.

I also owe many thanks to Professor Renato Renner for accepting me into his group, NCCR QSIT for giving me the opportunity to come to Zurich for one semester, and the whole Renner group, in particular Joe Renes for working with me on the SDP hierarchy. I am especially grateful to have had the opportunity to participate in the NCCR QSIT Winter School and General Meeting, where I learned a lot.

Furthermore, I want to thank my family and friends who supported me locally and remotely (especially in the final month) throughout my whole studies.
The output of this protocol also benefitted greatly from error correction performed by Joe Renes, Felipe Monetalegre-Mora, Wolfgang Belzig and Carsten Speckmann.

I also want to thank the Bonn-Cologne Graduate School for support of my master studies, Petra Neubauer-Günther for support and advice pertaining to the feasibility of this joint project, and everyone who facilitated this exchange.

\newpage
\pagenumbering{gobble}
\thispagestyle{empty}
\cleardoublepage
\newpage

\fancyhead{}
%\addcontentsline{toc}{chapter}{Note of Thanks}
\section*{Eidesstattliche Erklärung}

Hiermit versichere ich an Eides statt, dass ich die vorliegende Arbeit selbstständig und ohne die Benutzung anderer als der angegebenen Hilfsmittel angefertigt habe. Alle Stellen, die wörtlich oder sinngemäß aus veröffentlichten oder nicht veröffentlichten Schriften entnommen wurden, sind als solche kenntlich gemacht. Die Arbeit ist in gleicher oder ähnlicher Form oder auszugsweise im Rahmen einer anderen Prüfung noch nicht vorgelegt worden. Ich versichere, dass die eingereichte elektronische Fassung der eingereichten Druckfassung vollständig entspricht.
\vspace{1cm}

\begin{tabular}{lp{12em}l}
 \hspace{5cm}   && \hspace{5cm} \\\cline{1-1}\cline{3-3}
 &&\\
 Ort, Datum     && Unterschrift
\end{tabular}

\end{document}